\documentclass[12pt,letterpaper]{article}
\usepackage{xr-hyper}       
\usepackage[hypertexnames=false]{hyperref}
\usepackage{amsmath}

\hypersetup{
    colorlinks=true,
    linkcolor={red!50!black},
    citecolor={blue!50!black},
    urlcolor={blue!80!black}
}

\usepackage{graphicx}
\usepackage{enumerate}
\usepackage{natbib}

\usepackage{amsthm, amssymb}
\usepackage{booktabs}       
\usepackage{algorithm}      
\usepackage{algpseudocode}
\usepackage{multirow}
\usepackage{sectsty}
\usepackage{tabularx}
\usepackage{tikz}           
\usepackage{bm}             
\usepackage{xcolor}
\usepackage{microtype}
\usepackage{import}
\usepackage{titling}
\usetikzlibrary{arrows, backgrounds, patterns, matrix, shapes, fit, 
calc, shadows, plotmarks}

\graphicspath{{./figures/}}

\usepackage{multibib}
\newcites{App}{References}

\renewcommand\vec{\bm}
\newcommand{\simfn}{\mathtt{sim}} 
\newcommand{\truncsimfn}{\underline{\simfn}} 
\newcommand{\blockfn}{\mathtt{BlockFn}} 
\newcommand{\distfn}{\mathtt{dist}} 
\newcommand{\valset}{\mathcal{V}} 
\newcommand{\entset}{\mathcal{R}} 
\newcommand{\partset}{\mathcal{E}} 
\newcommand{\1}[1]{\mathbb{I}\!\left[#1\right]} 
\newcommand{\euler}{\mathrm{e}} 
\newcommand{\dblink}{\texttt{\upshape \lowercase{d-blink}}} 
\newcommand{\blink}{\texttt{\upshape \lowercase{blink}}} 
\def\spacingset#1{\renewcommand{\baselinestretch}%
  {#1}\small\normalsize} \spacingset{1}

\newtheorem*{remark}{Remark}
\newtheorem{proposition}{Proposition}
\newtheorem*{definition}{Definition}

\newcommand{\arxiv}{1}
\if0\arxiv
{
\externaldocument{supplementary}
} \fi

\sectionfont{\large\nohang\centering\MakeUppercase}

\title{\bf \dblink: Distributed End-to-End Bayesian Entity Resolution}
\author{Neil G.~Marchant\textsuperscript{a} \and
	Andee Kaplan\textsuperscript{b} \and 
	Daniel N.~Elazar\textsuperscript{c} \and
	Benjamin I.~P.~Rubinstein\textsuperscript{a} \and 
	Rebecca C.~Steorts\textsuperscript{d}}
\date{
	\textsuperscript{a}School of Computing and Information Systems, University 
	of Melbourne\\
	\textsuperscript{b}Department of Statistics, Colorado State University\\
	\textsuperscript{c}Methodology Division, Australian Bureau of Statistics\\
	\textsuperscript{d}Department of Statistical Science and Computer Science, Duke University\\Principal Mathematical Statistician, United States Census Bureau\\
	DRB \#: CBDRB-FY20-309\\[2ex]
	\today}

\begin{document}
\maketitle

\bigskip
\begin{abstract}
Entity resolution (ER; also known as record linkage or de-duplication) is 
the process of merging noisy databases, often in the absence of unique 
identifiers. 
A major advancement in ER methodology has been the application of Bayesian 
generative models, which provide a natural framework for inferring latent 
entities with rigorous quantification of uncertainty. 
Despite these advantages, existing models are severely limited in practice, 
as standard inference algorithms scale quadratically in the number of records. 
While scaling can be managed by fitting the model on separate blocks of 
the data, such a na\"{i}ve approach may induce significant error 
in the posterior. 
In this paper, we propose a principled model for scalable Bayesian ER, 
called ``\underline{d}istributed \underline{B}ayesian 
\underline{link}age'' or \dblink, which jointly performs 
blocking and ER without compromising posterior correctness. 
Our approach relies on several key ideas, including: 
(i)~an auxiliary variable representation that induces a partition of 
the entities and records into blocks; 
(ii)~a method for constructing well-balanced blocks based on k-d trees; 
(iii)~a distributed partially-collapsed Gibbs sampler with improved mixing; and 
(iv)~fast algorithms for performing Gibbs updates. 
Empirical studies on six data sets---including a case study 
on the 2010 Decennial Census---demonstrate the scalability and 
effectiveness of our approach.
\end{abstract}

\noindent%

\newpage
\if0\arxiv
\spacingset{1.5}
\fi

\section{Introduction}
\label{sec:introduction}

﻿When information about a statistical population is scattered across multiple 
databases, there may be immense value in combining them. 
A combined database can provide a more accurate and complete view of the 
population by improving coverage, bringing together analytic variables, and 
resolving erroneous and missing values. 
This allows statisticians to draw  richer and more reliable conclusions.  
Among the types of questions that can be addressed by combining such databases 
are the following: 
How accurate are census enumerations for minority groups 
\citep{winkler_overview_2006}? 
How many of the elderly are at high risk for sepsis in different parts of the 
country \citep{saria_trillion_2014}?  
How many people were victims of war crimes in recent conflicts in Syria 
\citep{price_2013}? 

An important step when combining databases is identifying records that refer 
to the same statistical unit. 
This is challenging in practice because consistent identifiers, such as 
social security numbers, are often not available. 
Identifiers may be omitted due to privacy concerns, they may be inconsistent 
across the databases, or they may have never been recorded. 
In such cases, practitioners must rely on  {\em entity resolution} (ER) to 
infer the relationships between records and statistical units 
(entities) using linking variables in the observed data. 
This problem is studied in the statistics, machine learning, database and natural language processing 
communities, and is also known as entity disambiguation, merge-purge, 
record linkage, deduplication and co-reference resolution \citep{christen_data_2012, dong_big_2015, soon_machine_2001}. 

ER is not only a crucial tool for statistical analysis, it is also a 
challenging statistical and computational problem in itself. 
This is because many databases lack reliable linking variables, the record 
comparison space scales quadratically in the number of records, and the 
number of parameters to be estimated grows with the number of records 
\citep{Herzog_2007, lahiri_2005, winkler_state_1999, winkler_2000}. 
To meet present and near-future needs, ER methods must be flexible and 
scalable to large databases. 
Furthermore, they must be able to handle uncertainty and be easily 
integrated with post-ER statistical analyses, such as regression. 
All of this must be done while achieving low error rates. 

Bayesian models offer a promising framework for ER as they support 
natural uncertainty propagation, flexible modeling assumptions, and 
incorporation of prior information. 
However, existing Bayesian ER models either ignore scalability 
\citep{steorts_entity_2015, zanella_flexible_2016, sadinle_bayesian_2017} or 
manage scalability in an unprincipled manner by applying blocking 
outside the Bayesian framework
\citep{fortini_bayesian_2001, larsen_advances_2005, larsen_experiment_2012, 
tancredi_hierarchical_2011, gutman_bayesian_2013, sadinle_detecting_2014, 
steorts_bayesian_2016}. 
Blocking improves scalability by partitioning records into blocks and 
assuming records in different blocks do not refer to the same entity 
\citep{christen_survey_2012}. 
However, when blocking is performed as a separate deterministic step it is 
not possible to propagate the uncertainty. 
Moreover, since the blocks are fixed, a poor blocking design may compromise 
the accuracy of the entire ER process. 
In other words, one sacrifices uncertainty propagation and accuracy 
for scalability. 

In this paper, we propose a principled approach to scaling Bayesian ER 
models, which does not suffer from the limitations of ad-hoc deterministic 
blocking. 
Using the \blink\ ER~model~\citep{steorts_entity_2015} as a foundation, 
we propose a scalable and distributed extension called 
``\underline{d}istributed \blink'' or \dblink\ for short, which integrates 
probabilistic blocking in a fully Bayesian framework. 
To our knowledge, \dblink\ is the first Bayesian ER model which 
supports propagation of uncertainty between the blocking and matching\slash 
linking stages of ER, without compromising the correctness of the 
posterior. 
In addition, \dblink\ supports distributed\slash parallel inference at 
the block level to further improve scalability to large databases.

We make several contributions to the literature. 
First, we propose an auxiliary variable representation of \blink, which 
induces a partitioning of the entities and records into blocks. 
These play a similar role as traditional deterministic blocks, however the 
assignments of latent entities and records to blocks are \emph{random} 
and inferred \emph{jointly} with the other model parameters. 
Second, we prove that our auxiliary variable representation 
preserves the marginal posterior distribution over the model parameters. 
This is a desirable property, as it means our inferences are theoretically 
independent of the blocking design. 
Third, we propose a method for constructing well-balanced blocks based 
on $k$-d trees. 
Fourth, we design a distributed partially-collapsed Gibbs sampler to perform 
inference, and demonstrate superior mixing times when compared to a 
standard Gibbs sampler. 
Fifth, we propose algorithms for improving computational efficiency of the 
Gibbs updates which leverage indexing data structures and a novel 
perturbation sampling algorithm. 

We implement our proposed methodology as an open-source Apache Spark 
package\footnote{Spark package source code available at 
  \url{https://github.com/cleanzr/dblink}.}
and provide an R interface for broad accessibility\footnote{R package 
  source code available at \url{https://github.com/cleanzr/dblinkR}.}. 
We conduct empirical evaluations on two synthetic and three real data sets, 
demonstrating efficiency gains in excess of 300$\times$ compared to 
\blink. 
To illustrate the effectiveness of our approach for realistic ER tasks, 
we present a case study using Census and administrative data from the 
U.S.\ state of Wyoming.

The paper is organized as follows. 
In Section~\ref{sec:related} we review related work in ER methodology 
and approximate inference algorithms. 
We then formulate ER in a Bayesian setting in Section~\ref{sec:model}, and 
present the \dblink\ model with integrated probabilistic blocking.
In Section~\ref{sec:partition-fn} we provide guidelines for selecting 
blocking functions. 
We then discuss inference and propose a distributed partially-collapsed 
Gibbs sampler in Section~\ref{sec:inference}. 
We suggest additional methods for improving computational efficiency of 
inference in Section~\ref{sec:tricks}.
Section~\ref{sec:experiments} provides a comprehensive empirical evaluation, 
and Section~\ref{sec:decennial} presents a case study to U.S.\ Census and 
administrative data. 
We make closing remarks in Section~\ref{sec:conclusions}.

\section{Related work}
\label{sec:related}
We review related work across three main areas---ER methodology, inference 
for Bayesian ER models, and distributed Markov chain Monte Carlo (MCMC).

\paragraph{Entity resolution methodology.}
The first probabilistic approach to ER was due to 
\citet{newcombe_automatic_1959}, who applied matching rules to 
pairs of records. 
This idea was later formalized in a seminal paper by 
\citet{fellegi_theory_1969} within a decision-theoretic framework.
Many variations of the Fellegi-Sunter (FS) approach have been proposed 
(for surveys, see~\citealp{winkler_overview_2006, winkler2014matching}), 
including a generalization to multiple 
databases~\citep{sadinle_generalized_2013}.
Others have addressed scalability of FS-type approaches using 
blocking\slash indexing methods (see 
\citealp{christen_data_2012,steorts_comparison_2014} for surveys) and 
efficient data structures \citep{enamorado_using_2019}.
However, traditional FS approaches do not naturally support propagation of 
ER uncertainty, and existing methods for scaling make approximations that 
sacrifice accuracy.

While the FS approach has been highly influential, it has also been 
criticized due to its lack of support for duplicates within databases; 
misspecified independence assumptions; and its dependence on subjective 
thresholds \citep{tancredi_hierarchical_2011}.
These limitations have prompted development of more sophisticated 
Bayesian models, including models for bipartite 
matching~\citep{fortini_bayesian_2001, larsen_advances_2005, 
  larsen_experiment_2012, tancredi_hierarchical_2011, gutman_bayesian_2013, 
  sadinle_bayesian_2017, mcveigh_scaling_2019}, 
deduplication~\citep{sadinle_detecting_2014, tancredi_unified_2020}
and matching across multiple databases~\citep{steorts_entity_2015, 
  steorts_bayesian_2016}.
Several of these models operate on attribute-level comparisons between 
pairs of records in a similar vein as the FS approach 
\citep{larsen_advances_2005, larsen_experiment_2012, gutman_bayesian_2013, 
  sadinle_detecting_2014, sadinle_bayesian_2017, mcveigh_scaling_2019}.
This contrasts with entity-centric generative models which assume the 
records arise as distortions to some latent entity attributes 
\citep{tancredi_hierarchical_2011, steorts_entity_2015, 
  steorts_bayesian_2016, tancredi_unified_2020}.

In scenarios where training data is scarce or unavailable, Bayesian 
generative models tend to be more robust than discriminative or 
likelihood-based methods, as the priors have a regularizing effect. 
Bayesian generative models are also amenable to theoretical analysis: 
recent work has obtained lower bounds on the probability of misclassifying 
the entity associated with a record \citep{steorts_performance_2017}.
However, a major downside of Bayesian ER models is the computational 
cost of performing inference (see discussion below).

Apart from these advances in Bayesian models for ER (largely undertaken 
in statistics), there have been an abundance of contributions from the 
database and machine learning communities 
\citep[see surveys by][]{getoor_entity_2012, christen_data_2012}.
Their focus has typically been on rule-based 
approaches~\citep{fan_reasoning_2009, singh_generating_2017}, 
supervised learning approaches~\citep{mudgal_deep_2018}, 
hybrid human-machine approaches~\citep{wang_crowder:_2012, 
  gokhale_corleone:_2014}, 
and scalability~\citep{papadakis_comparative_2016}.
Broadly speaking, all of these approaches rely on either humans in-the-loop 
or large amounts of labelled training data, which is not generally the 
case in the Bayesian setting.

\paragraph{Inference for Bayesian ER models.}
Most prior work on Bayesian generative models for ER
\citep[e.g.][]{tancredi_hierarchical_2011, gutman_bayesian_2013,
  steorts_entity_2015} has relied on Gibbs sampling for inference. 
Compared to other Markov chain Monte Carlo (MCMC) algorithms, 
Gibbs sampling is relatively easy to implement, however it may suffer 
from slow convergence and poor mixing owing to its highly local 
moves~\citep{liu_monte_2004}. 
Scalability is also a challenge, as a na\"{i}ve Gibbs update for the 
linkage structure requires all-to-all comparisons between records 
(or between records and entities for entity-centric models). 
This issue is often managed by applying deterministic blocking prior to 
Gibbs sampling, thereby sacrificing accuracy and proper 
treatment of uncertainty \citep{larsen_advances_2005, larsen_experiment_2012, 
  tancredi_hierarchical_2011, gutman_bayesian_2013, sadinle_detecting_2014}.

In the broader context of clustering models, the \emph{split-merge 
	algorithm}~\citep{jain_split-merge_2004} has been proposed 
as an alternative to Gibbs sampling.
It is a Metropolis-Hastings algorithm, which traverses the space of 
clusterings via proposals that split individual clusters or merge pairs 
of clusters. 
Since multiple cluster items are updated in a single move, it is less 
susceptible to becoming trapped in local modes.
\cite{steorts_bayesian_2016} applied this algorithm, in combination with 
deterministic blocking, to update the linkage structure in an ER model 
similar to \blink. 
A close relative of the split-merge algorithm is the \emph{chaperones 
  algorithm}, which was proposed for inference in microclustering models
\citep{zanella_flexible_2016}. 
The chaperones algorithm is expected to be more efficient, as it 
preferentially focuses on more likely cluster reassignments, 
through a user-specified biased distribution on the product space of 
cluster items. 
However, the biased distribution must be designed so that random 
item pairs can be drawn efficiently, without explicitly constructing 
the product space. 

More recently, \citet{zanella_informed_2020} proposed a general 
framework for designing informative proposals in a Metropolis-Hastings 
setting, which is suited for discrete spaces (e.g.\ the space of 
possible linkage structures). 
They show that \emph{locally-balanced proposals} are asymptotically-optimal 
within the class of pointwise informative proposals, and demonstrate 
significant improvements in efficiency when compared to a split-merge-type
algorithm.
However, computing a locally-balanced proposal for the linkage structure 
is computationally challenging due to quadratic scaling. 
This can be mitigated to some extent by running locally-balanced 
updates within randomly-selected sub-blocks of records. 
However to avoid poor mixing, care must be taken to ensure that  
randomly-selected sub-blocks contain likely matching records.

In contrast to much of the literature on Bayesian ER models, 
\citet{mcveigh_scaling_2019} proposed a method that combines deterministic 
blocking and restricted MCMC (based on earlier work by 
\citealp{mcveigh_practical_2017}). 
They balance approximation error by performing coarse-grained 
deterministic blocking\slash indexing as an initial step, followed by 
data-dependent post-hoc blocking.
During inference, the linkage structure is updated using locally-balanced 
proposals, restricted to the post-hoc blocks. 
They demonstrate improved scalability---to data sets with several hundred 
thousand of record---with minimal risk of approximation error. 
However, their approach is not directly compatible with distributed 
inference (see below) and may require modification for use with an 
entity-centric model.

\paragraph{Parallel\slash distributed MCMC.}
Recent literature has focused on using parallel and 
distributed computing to scale up MCMC algorithms, 
where applications have included Bayesian topic 
models~\citep{newman_distributed_2009, smola_architecture_2010, 
  ahn_distributed_2014} 
and mixture models~\citep{williamson_parallel_2013, 
  chang_parallel_2013, lovell_clustercluster:_2013, 
  ge_distributed_2015}.
We review the application to mixture models, as they are 
conceptually similar to ER models.

Existing work has concentrated on Dirichlet process (DP) 
mixture models and hierarchical DP mixture 
models.
The key to enabling distributed inference for these 
models is the realization that a DP mixture model can 
be reparameterized as a mixture of DPs.
Put simply, the reparameterized model induces a 
\emph{partitioning} of the clusters into blocks, such that clusters 
assigned to \emph{distinct blocks} are conditionally 
independent. 
As a result, variables within blocks can be updated in parallel.
\cite{williamson_parallel_2013} exploited this idea at the 
thread level to parallelize inference for a DP mixture model.
\cite{chang_parallel_2013} followed a similar approach, but 
included an additional level of parallelization within blocks 
using a parallelized version of the split-merge algorithm.
Others~\citep{lovell_clustercluster:_2013, ge_distributed_2015} 
have developed distributed implementations in the MapReduce 
framework. 

We do not consider DP mixture models in our work, as their behavior 
is ill-suited for ER applications.\footnote{With a DP prior, the 
number of clusters grows logarithmically in the number of records, 
but empirical observations call for near-linear growth~\citep{zanella_flexible_2016}.
}
However we do borrow the reparameterization idea, albeit with a more flexible 
partition specification which permits similar entities to be co-blocked, 
while facilitating load balancing.
It would be interesting to see whether similar ideas can be applied to 
microclustering models~\citep{zanella_flexible_2016}, 
however preserving the marginal posterior distribution seems 
challenging in this case.

\section{A scalable model for Bayesian ER}
\label{sec:model}

In this section, we present our scalable ER model called \dblink, 
which integrates probabilistic blocking in a fully Bayesian framework.
Our model can be viewed as an extension of the \blink\ model 
\citep{steorts_entity_2015} that incorporates an auxiliary partition 
of the latent entity parameter space into blocks. 
Unlike ad-hoc blocking approaches used previously in the literature 
\citep{larsen_advances_2005, larsen_experiment_2012, 
  tancredi_hierarchical_2011, gutman_bayesian_2013, sadinle_detecting_2014, 
  steorts_bayesian_2016}, 
the blocks in \dblink\ are random, and inferred jointly with the other 
model parameters. 
This enables propagation of uncertainty between the blocking and ER stages.
In addition, \dblink\ extends \blink\ with support for missing values and 
user-defined attribute similarity measures.

We describe notation and assumptions in Section~\ref{sec:problem}, before 
presenting \dblink\ in Section~\ref{sec:model-specification}.
We define attribute similarity measures in Section~\ref{sec:attribute-sim-measure}, 
including an optional truncation approximation which can improve scalability.
In Section~\ref{sec:posterior-dist}, we prove that the marginal posterior of  
\dblink\ (integrated over the blocks) reduces to \blink\ under certain 
conditions.
This is a desirable property, as it means our inferences are theoretically independent of the blocking design. 
Finally, in Section~\ref{sec:partition-motivation} we explain how the auxiliary 
blocks are beneficial in scaling and distributing inference.

\subsection{Notation and problem formulation}
\label{sec:problem}
In this section, we define notation and formulate ER in a Bayesian 
setting.
Consider a collection of $T$ tables\footnote{We define a \emph{table} 
as an ordered (indexed) collection of records, which may contain
duplicates (records for which all attributes are identical).} 
(databases) indexed by $t$, each with $R_t$ records (rows) indexed by $r$
and $A$ aligned attributes (columns) indexed by $a$.
Associated with the records is a fixed population of entities of size $E$ 
indexed by $e$. 
Each entity $e$ is described by a set of attributes 
$\vec{y}_e = [y_{ea}]_{a = 1 \ldots A}$, which are aligned with the 
record attributes. 
The population of entities is partitioned into $B$ blocks for 
computational convenience, using a blocking function $\blockfn$ that maps an 
entity $e$ to a block based on its attributes $\vec{y}_e$.
We assume each record $(t,r)$ belongs to a block $\gamma_{tr}$ and 
is associated with an entity $\lambda_{tr}$ within that block.
The value of the $a$-th attribute for record $(t,r)$ is 
denoted by $x_{tra}$, and is assumed to be a noisy observation 
of the associated entity's true attribute value $y_{\lambda_{tr}a}$.
We allow for the fact that some attributes $x_{tra}$ may be missing
completely at random through a corresponding indicator variable 
$o_{tra}$~\cite[p.~12]{little_statistical_2002}.

Table~\ref{tbl:notation} summarizes our notation, including 
model-specific parameters which will be introduced shortly.
We adopt the following rules to compactly refer to sets of variables:
\begin{itemize}
  \item A boldface lower-case variable denotes the set of 
  \emph{all attributes}: e.g.~$\vec{x}_{tr} = [x_{tra}]_{a=1\ldots A}$.
  \item A boldface capital variable denotes the set of 
  \emph{all index combinations}: 
  e.g.~$\vec{X} = [x_{tra}]_{t=1\ldots T; r=1 \ldots R_t; a=1 \ldots A}$.
\end{itemize}
We also define notation to separate the record attributes $\mathbf{X}$ 
into an observed part $\mathbf{X}^{(o)}$ (those $x_{tra}$'s for which 
$o_{tra} = 1$) and a missing part $\mathbf{X}^{(m)}$ (those 
$x_{tra}$'s for which $o_{tra}=0$).

After specifying a generative model (see next section), we perform 
ER by inferring the \emph{joint} posterior distribution over:
\begin{itemize}
  \item the block assignments $\vec{\Gamma} = [\gamma_{tr}]_{t = 1 \ldots T; r = 1 \ldots R_t}$,
  \item the linkage structure $\vec{\Lambda} = [\lambda_{tr}]_{t = 1 \ldots T; r = 1 \ldots R_t}$, and 
  \item the true entity attribute values $\vec{Y} = [y_{ea}]_{e = 1 \ldots E; a = 1 \ldots A}$,
\end{itemize}
conditional on the observed record attribute values $\mathbf{X}^{(o)}$. 
Note that we operate in a fully unsupervised setting, since we do 
not condition on ground truth data for the links or entities. 
Inferring $\vec{\Gamma}$ is equivalent to the \emph{blocking} stage 
of ER, where the records are partitioning into blocks to limit the 
comparison space.
Inferring $\vec{\Lambda}$ is equivalent to the \emph{matching\slash linking} 
stage of ER, where records that refer to the same entities are linked together. 
Inferring $\vec{Y}$ is equivalent to the \emph{merging} stage, where 
linked records are combined to produce a single representative record. 
By inferring $\vec{\Gamma}$, $\vec{\Lambda}$ and $\vec{Y}$ jointly, we are 
able to propagate uncertainty between the three stages.

\subsection{Model specification}
\label{sec:model-specification}
We now present our proposed model \dblink\ by describing the generative 
process.
We provide a visual representation of the model in 
Figure~\ref{fig:plate-diagram}, with key differences from \blink\ 
highlighted in a dashed blue line style.

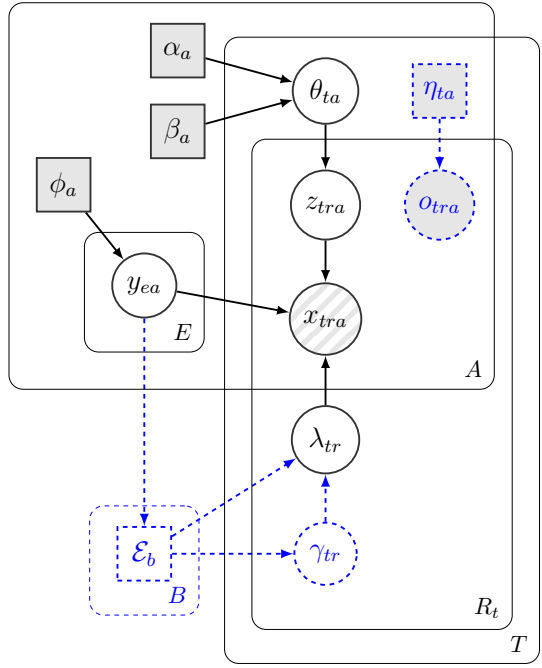
\begin{figure}
  \centering
  \resizebox{0.45\linewidth}{!}{%
    \begin{tikzpicture}
	\tikzset{
		hatch distance/.store in=\hatchdistance,
		hatch distance=8pt,
		hatch thickness/.store in=\hatchthickness,
		hatch thickness=2pt
  }

	\makeatletter
	\pgfdeclarepatternformonly[\hatchdistance,\hatchthickness]{flexible hatch}
	{\pgfqpoint{0pt}{0pt}}
	{\pgfqpoint{\hatchdistance}{\hatchdistance}}
	{\pgfpoint{\hatchdistance-1pt}{\hatchdistance-1pt}}%
	{
		\pgfsetcolor{\tikz@pattern@color}
		\pgfsetlinewidth{\hatchthickness}
		\pgfpathmoveto{\pgfqpoint{0pt}{0pt}}
		\pgfpathlineto{\pgfqpoint{\hatchdistance}{\hatchdistance}}
		\pgfusepath{stroke}
	}
	\makeatother

	\tikzstyle{main}=[circle, minimum size=9mm, thick, draw=black!80, node distance=17mm]
	\tikzstyle{hyparam}=[rectangle, minimum size=5mm, thick, draw=black!80, fill=black!10, node distance=17mm]
	\tikzstyle{param}=[rectangle, minimum size=8mm, thick, draw=black!80, node distance=17mm]
	\tikzstyle{connect}=[-latex, thick]
	\tikzstyle{shortconnect}=[-latex, thin]
	\tikzstyle{plate}=[rectangle, inner sep=4mm, yshift=-1mm, draw=black, rounded corners=6pt]
	\tikzstyle{highlight}=[draw=blue,text=blue, dash pattern=on 2pt off 2pt]
	\tikzstyle{platelab}=[anchor=south east, xshift=-0.4mm, yshift=0.4mm]
	\tikzstyle{switch}=[circle, minimum size=1mm, fill=black, draw=black]
	\node[param, fill=black!10] (alpha) {$\alpha_{a}$};
	\node[main] (theta) [below right of=alpha, xshift=10mm, yshift=6mm] {$\theta_{ta}$};
	\node[param, fill=black!10] (beta) [below left of=theta, xshift=-10mm, yshift=6mm] {$\beta_{a}$};
	\node[main] (z) [below of=theta] {$z_{tra}$};
	\node[main, pattern=flexible hatch, pattern color=black!10] (x) [below of=z] {$x_{tra}$};
	\node[main] (lambda) [below of=x, yshift=-1mm] {$\lambda_{tr}$};
	\node[main, highlight] (gamma) [below of=lambda] {$\gamma_{tr}$};
	\node[main] (y) [left of=x, xshift=-10mm, yshift=5mm] {$y_{ea}$};
	\node[param, fill=black!10] (phi) [above left of=y, yshift=3mm] {$\phi_{a}$};
	\node[param, highlight] (P) [left of=gamma, xshift=-10mm] {$\partset_{b}$};
	\node[param, fill=black!10, highlight] (eta) [right of=theta] {$\eta_{ta}$};
	\node[main, fill=black!10, highlight] (o) [below of=eta] {$o_{tra}$};
	\path (alpha) edge [connect] (theta)
	(beta) edge [connect] (theta)
	(theta) edge [connect] (z)
	(z) edge [connect] (x)
	(lambda) edge [connect] (x)
	(gamma) edge [connect, highlight, fill=blue] (lambda)
	(eta) edge [connect, highlight, fill=blue] (o)
	(y) edge [connect] (x)
	(y) edge [connect, highlight, fill=blue] (P)
	(P) edge [connect, highlight, fill=blue] (gamma)
  (P) edge [connect, highlight, fill=blue] (lambda)
	(phi) edge [connect] (y);
	\node[plate, fit=(y)] (plate-e) {};
	\node[platelab] at (plate-e.south east) (e) {\footnotesize$E$};
	\node[plate, fit=(z) (gamma) (x) (lambda) (o), inner sep=5.5mm] (plate-r) {};
	\node[platelab] at (plate-r.south east) {\footnotesize$R_t$};
	\node[plate, fit=(theta) (eta) (plate-r), inner sep=4mm] (plate-t) {};
	\node[platelab] at (plate-t.south east) {\footnotesize$T$};
  \node[plate, fit=(P), highlight] (plate-P) {};
  \node[platelab, text=blue] at (plate-P.south east) {\footnotesize$B$};
	\node[plate, fit=(phi) (alpha) (x) (eta) (plate-e), yshift=0mm, inner sep=4mm] (plate-a) {};
	\node[platelab] at (plate-a.south east) {\footnotesize$A$};
\end{tikzpicture}
  }
  \caption{Plate diagram for \dblink. 
    Extensions to \blink\ are highlighted in a dashed blue line style.
    Circular nodes represent random variables; square nodes represent 
    deterministic variables; (un)shaded nodes represent (un)observed variables; 
    arrows represent conditional dependence; and plates represent replication 
    over an index.
  }
  \label{fig:plate-diagram}
\end{figure}

\begin{table}
  \caption{Summary of notation.}
  \label{tbl:notation}
  \begin{center}
  \newcolumntype{Z}{>{\arraybackslash}X}
  \spacingset{1}
  \footnotesize
  \begin{tabularx}{\linewidth}{@{}l Z @{} l Z@{}}
  	\toprule
  	Symbol                       & Description & Symbol & Description\\
  	\midrule
    $t \in 1 \ldots T$           & index over tables 
    & $\gamma_{tr}$              & assigned block for record $r$ in table $t$ \\
    $r \in 1 \ldots R_t$         & index over records in table $t$ 
    & $\lambda_{tr}$             & assigned entity for record $r$ in table $t$ \\
    $e \in 1 \ldots E$           & index over entities 
    & $\theta_{ta}$              & prob.\ attribute $a$ in table $t$ is distorted\\
    $b \in 1 \ldots B$           & index over block 
    & $\alpha_a, \beta_a$        & distortion hyperparams.\ for attribute $a$ \\
    $a \in 1 \ldots A$           & index over attributes 
    & $\eta_{ta}$                & prob.\ attribute $a$ in table $t$ is observed \\
    $v \in 1 \ldots |\valset_a| $ & index over domain of attribute $a$ 
    & $\valset_a$                & domain of attribute $a$ \\
    $R = \sum_{t} R_t$           & total number of records 
    & $\phi_{a}(\cdot)$          & distribution over domain of attribute $a$ \\
    $x_{tra}$                    & attribute $a$ for record $r$ in table $t$     
    & $\simfn_a(\cdot, \cdot)$   & similarity measure for attribute $a$ \\
    $z_{tra}$                    & distortion indicator for $x_{tra}$ 
    & $\entset_e$                & set of records assigned to entity $e$\\
    $o_{tra}$                    & observed indicator for $x_{tra}$ 
    & $\partset_b$               & set of entities assigned to block $b$ \\     
    $y_{ea}$                     & attribute $a$ for entity $e$ 
    & $\blockfn(\cdot)$           & block assignment function \\
    \bottomrule
  \end{tabularx}
  \end{center}
\end{table}

\paragraph{Entities.} 
The population of entities is assumed to be of fixed size $E$.
Each entity $e$ is described by a vector of ``true'' attributes 
$\vec{y}_e \in \bigotimes_{a=1}^{A} \valset_{a}$.
The value of the $a$-th attribute $y_{ea}$ is assumed to be drawn 
independently from a distribution 
$\phi_{a}$ over the attribute domain $\valset_{a}$:
\begin{equation}
y_{ea} \overset{\mathrm{ind.}}{\sim}
  \operatorname{Discrete}_{v \in \valset_{a}} [\phi_{a}(v)].
\end{equation}
Following the \blink\ model, we set the population size $E$ and 
the distributions over the attribute domains $\phi_{a}$ empirically.
Recommendations for setting these parameters are provided in 
Appendix~\ref{app-sec:parameter-settings}.

\paragraph{Blocks.} 
The parameter space associated with the entities 
$\bigotimes_{a=1}^{A} \valset_{a}$ is partitioned into $B$ blocks. 
The partition is parameterized using a deterministic 
\emph{blocking function}:
\begin{equation}
  \blockfn: \bigotimes_{a} \valset_{a} \to \{1,\ldots, B\},
  \label{eqn:part-fn}
\end{equation}
which is a free parameter and may be selected for inferential 
  convenience. 
We provide recommendations for selecting the blocking function in 
Section~\ref{sec:partition-fn}, including an example based on $k$-d trees.

We shall often need to refer to the entities assigned to a particular block. 
To do this concisely, we introduce the notation 
$\partset_{b}(\vec{Y}) = \{e: \blockfn(\vec{y}_{e}) = b\}$ to denote 
the set of entities assigned to block $b$. 
This is random due to the dependence on $\vec{Y}$, however we shall often 
omit the dependence for brevity.


\paragraph{Distortion.} Associated with each table $t$ and 
attribute $a$ is a distortion probability $\theta_{ta}$, 
with assumed prior distribution:
\begin{equation}
\theta_{ta}|\alpha_{a}, \beta_{a} \overset{\mathrm{ind.}}{\sim}
\operatorname{Beta}[\alpha_{a}, \beta_{a}],
\end{equation}
where $\alpha_{a}$ and $\beta_{a}$ are hyperparameters.
We provide recommendations for setting $\alpha_{a}$ and $\beta_{a}$ in 
Appendix~\ref{app-sec:experiments-setup}.
The distortion probabilities feed into the record-generation process below.

\paragraph{Records.} We assume a record is generated 
by selecting an entity uniformly at random and 
copying the entity's attributes subject to distortion. 
The process for generating record $r$ in table $t$ is outlined below. 
Steps (i), (ii), and (v) deviate from \blink.
\begin{enumerate}[(i)]
  \item Choose a block assignment $\gamma_{tr}$ at random 
  in proportion to the block sizes:
  \begin{equation}
  \gamma_{tr}|\vec{Y} \overset{\mathrm{ind.}}{\sim}
    \operatorname{Discrete}_{b \in \{1 \ldots B\}}[|\partset_b|/E].
  \end{equation}
  \item Choose an entity assignment $\lambda_{tr}$ uniformly 
  at random from block $\gamma_{tr}$:
  \begin{equation}
  \lambda_{tr}|\gamma_{tr}, \vec{Y} \overset{\mathrm{ind.}}{\sim}
    \operatorname{DiscreteUniform}[\partset_{\gamma_{tr}}].
  \end{equation}
  \item For each attribute $a$, draw a distortion indicator
  $z_{tra}$:
  \begin{equation}
  z_{tra}|\theta_{ta} \overset{\mathrm{ind.}}{\sim}
    \operatorname{Bernoulli}[\theta_{ta}].
  \end{equation}
  \item For each attribute $a$, draw a record value $x_{tra}$:
  \begin{equation}
  x_{tra}|z_{tra}, y_{\lambda_{tr}a} \overset{\mathrm{ind.}}{\sim}
  (1 - z_{tra}) \delta(y_{\lambda_{tr}a}) + z_{tra} 
  \operatorname{Discrete}_{v \in \valset_{a}}[\psi_{a}(v|y_{\lambda_{tr}a})]
  \end{equation}
  where $\delta(\cdot)$ represents a point mass.
  If $z_{tra} = 0$, $x_{tra}$ is copied directly from the 
  entity.
  Otherwise, $x_{tra}$ is drawn from the domain $\valset_{a}$ 
  according to the distortion distribution $\psi_{a}$.
  In the literature, this is known as a hit-miss 
  model~\citep{copas_record_1990}.
  \item For each attribute $a$, draw an observed indicator $o_{tra}$:
  \begin{equation}
    o_{tra} \overset{\mathrm{ind.}}{\sim} \operatorname{Bernoulli}[\eta_{ta}].
  \end{equation}
  If $o_{tra} = 1$, $x_{tra}$ is observed, otherwise it is 
  missing.
\end{enumerate}

\paragraph{Detail on the distortion distribution.} $\psi_{a}(\cdot|w)$ 
chooses a distorted value for attribute $a$ conditional on 
the true value $w$.
In our parameterization of the model, it is defined as
\begin{equation}
\psi_{a}(v|w) = h_{a}(w) \phi_{a}(v) \euler^{\simfn_{a}(v,w)},
\label{eqn:distortion-dist}
\end{equation}
where $h_{a}(w) = 1 / \sum_{v \in \valset_{a}} \phi_{a}(v) 
\euler^{\simfn_{a}(v,w)}$ is a normalization constant and 
$\simfn_{a}$ is the similarity measure for attribute~$a$
(see Section~\ref{sec:attribute-sim-measure}).
Intuitively, this distribution chooses values in proportion to 
their empirical frequency, while placing more weight on those 
that are ``similar'' to $w$.
This reflects the notion that distorted values are likely 
to be close to the truth, as is the case when modeling
typographical errors.

\paragraph{Posterior distribution.} 
The generative process described above corresponds to a posterior distribution 
over the model parameters, conditioned on the observed records. 
By reading the conditional dependence structure off the plate diagram 
(Figure~\ref{fig:plate-diagram}) and marginalizing over 
the missing record attributes $\vec{X}^{(m)}$, one can show that the 
posterior distribution is of the following form:
\begin{equation}
\begin{split}
p(\vec{\Gamma}, \vec{\Lambda}, \vec{Y}, \vec{Z}, \vec{\Theta}|\vec{X}^{(o)}, \vec{O}) 
& \propto \prod_{e, a} p(y_{ea}|\phi_{a}) \times \prod_{t, a} p(\theta_{ta}|\alpha_{a}, \beta_{a}) \times \prod_{\substack{t, r, a\\o_{tra}=1}} p(x_{tra}|z_{tra}, \lambda_{tr}, y_{\lambda_{tr}a}) \\
& \qquad {} \times \prod_{t, r} \Big\{ p(\gamma_{tr}|\vec{Y}) 
p(\lambda_{tr}|\gamma_{tr}, \vec{Y}) \prod_{a} p(z_{tra}|\theta_{ta}) \Big\}.
\end{split} 
\label{eqn:partitioned-posterior}
\end{equation}
For further detail on the derivation and an expanded form of the 
posterior, we refer the reader to Appendix~\ref{app-sec:full-posterior}.

\subsection{Attribute similarity measures}
\label{sec:attribute-sim-measure}
We now discuss the attribute similarity measures that appear 
in the distortion distribution of Equation~\ref{eqn:distortion-dist}.
The purpose of these measures is to quantify the 
propensity that some value $v$ in the attribute domain is 
chosen as a distorted alternative to the true value $w$.
\begin{definition}[Attribute similarity measure]
  \label{def:attribute-sim-measure}
  Let $\valset$ be the domain of an attribute.
  An \emph{attribute similarity measure} on $\valset$ is a 
  function $\simfn: \valset \times \valset \to [0, s_\mathrm{max}]$
  that satisfies $0 \leq s_\mathrm{max} < \infty$ and 
  $\simfn(v,w) = \simfn(w,v)$ for all $v, w \in \valset$.
\end{definition}

Note that our parameterization in terms of attribute \emph{similarity} 
measures differs from \blink, which uses \emph{distance} measures.
This allows us to make use of a more efficient sampling method, as 
described in Section~\ref{sec:pert-sampling}. 
The next proposition states that the two parameterizations 
are equivalent, so long as the distance measure is 
bounded and symmetric (a proof is provided in 
Appendix~\ref{app-sec:proof-sim-dist-equiv}).
\begin{proposition}
  \label{thm:sim-dist-equiv}
  Let $\distfn_{a}: \valset \times \valset \to [0, d_{\mathrm{max};a}]$ 
  be the attribute distance measure that appears in \blink, 
  and assume that $0 \leq d_{\mathrm{max};a} < \infty$ and 
  $\distfn_{a}(v, w) = \distfn_{a}(w, v)$ for all $v,w \in \valset$. 
  Define the corresponding attribute similarity measure 
  for \dblink\ as
  \begin{equation}
  \simfn_{a}(v, w) := d_{\mathrm{max};a} - \distfn_{a}(v, w).
  \label{eqn:simfn-distfn-correspondence}
  \end{equation}
  Then the parameterization of 
  $\psi_{a}$ used in \dblink\ is equivalent to \blink.
\end{proposition}

In this paper, we restrict our attention to the following 
similarity measures for simplicity:
\begin{itemize}
  \item \emph{Constant similarity measure.} This measure 
  is appropriate for categorical attributes, where there is 
  no reason to believe one value is more likely than any 
  other as a distortion to the true value $w$.
  Without loss of generality, it may be defined as 
  $\simfn_\mathrm{const}(v, w) = s_\mathrm{max}$ for all 
  $v, w \in \valset$.
  \item \emph{Normalized edit similarity measure.} 
  This measure is based on the edit distance metric, and is 
  suitable for modeling distortion in generic string attributes.
  Following~\cite{yujian_normalized_2007}, we define a normalized 
  edit distance metric,
  \begin{equation*}
    \distfn_{\mathrm{nEd}}(v,w) = \frac{2\ \distfn_{\mathrm{Ed}}(v,w)}
    {|v| + |w| + \distfn_{\mathrm{Ed}}(v,w)},
  \end{equation*}
  where $\distfn_{\mathrm{Ed}}$ denotes the regular edit distance 
  and $|v|$ denotes the length of string $v$.
  Note that alternative definitions of the normalized edit distance 
  could be used \citep[see references in][]{yujian_normalized_2007}, 
  however the above definition is unique in that it yields a 
  proper metric.
  Since the normalized edit distance is bounded on the interval 
  $[0,1]$ we can define a corresponding normalized edit similarity 
  measure:
  \begin{equation}
    \simfn_{\mathrm{nEd}}(v, w) = 1 - \distfn_{\mathrm{nEd}}(v, w).
  \label{eqn:norm-edit-sim-fn}
  \end{equation}
\end{itemize}
Ideally, one should select attribute similarity measures based on 
the data at hand.
There are many possibilities to consider, such as  
Jaccard similarity, numeric similarity 
measures~\citep{lesot_similarity_2008} and other 
domain-specific measures~\citep{bilenko_adaptive_2003}.

\subsection{Model equivalence}
\label{sec:posterior-dist}
We have purposely constructed \dblink\ so that it  
reduces to \blink\ under certain conditions. 
Assuming the records are fully observed, the posterior distribution 
of \dblink\ as specified in Equation~\ref{eqn:partitioned-posterior} 
is similar to \blink. 
The difference lies in the factors involving the block assignments 
$\gamma_{tr}$ and the entity assignments $\lambda_{tr}$. 
However, if one marginalizes out the auxiliary block assignments---as is 
done automatically in Markov chain Monte Carlo---the 
posterior distributions are identical.
This statement is made precise below (proof provided in 
Appendix~\ref{app-sec:proof-posterior-equiv}):
\begin{proposition}
  \label{thm:posterior-equiv}
  Suppose the conditions of Proposition~\ref{thm:sim-dist-equiv} hold
  and that $\alpha_{a} = \alpha$ and $\beta_{a} = \beta$ for 
  all $a$. 
  Assume furthermore that all record attributes are observed, i.e.\ 
  $o_{tra} = 1$ for all $t, r, a$.
  Then the marginal posterior of $\vec{\Lambda}$, $\vec{Y}$, 
  $\vec{Z}$ and $\vec{\Theta}$ for \dblink\ (i.e. marginalized 
  over $\vec{\Gamma} = [\gamma_{tr}]_{t = 1 \ldots T; r = 1 \ldots R_t}$)
  is identical to the posterior for \blink.
\end{proposition}

This is an important result, as it shows our inferences for the 
meaningful model parameters are the same as we would obtain from \blink. 
Thus we are able to apply blocking to scale the model, without compromising 
the correctness of the posterior distribution.

\subsection{Rationale for introducing block}
\label{sec:partition-motivation}
We now briefly explain the role of the auxiliary block in \dblink.
First, we note that without the block ($B = 1$), the Markov 
blanket for $\lambda_{tr}$ includes the attribute values for 
\emph{all} of the entities $\vec{Y}$.
This presents a major obstacle when it comes to distributing the 
inference on a compute cluster, as the data is not separable.
By incorporating block, we restrict the Markov blanket 
for $\lambda_{tr}$ to include only a subset of the entity 
attribute values---those in the same block as record $(t,r)$.
As a result, it becomes natural to distribute the inference 
so that each compute node is responsible for a single block 
(see Section~\ref{sec:distributed-pcg} for details).
Secondly, we can interpret the block as performing probabilistic 
blocking in the context of MCMC sampling 
(introduced in Section~\ref{sec:inference}), which improves 
computational efficiency.
In a given iteration, the possible links for a record are 
restricted to the entities residing in the same block.
However, unlike conventional blocking, the block assignments 
are not fixed---between iterations the entities and linked records 
may move between blocks.

\section{Blocking functions}
\label{sec:partition-fn}
In Section~\ref{sec:model-specification} we introduced a generic blocking 
function (Equation~\ref{eqn:part-fn}) that is responsible for assigning 
entities to blocks.
This function may be regarded as a free parameter, 
since it has no bearing on model equivalence according to 
Proposition~\ref{thm:posterior-equiv}.
However, from a practical perspective the blocking function ought to be 
chosen carefully, as it can impact inferential efficiency---both in 
terms of computational and mixing time.
We suggest some guidelines for choosing a blocking function in 
Section~\ref{sec:interpret}, before presenting an example based on $k$-d trees 
in Section~\ref{sec:kd}.

\subsection{Interpretation and guidelines}
\label{sec:interpret}
Recall that the blocking function assigns an entity to 
a block according to its attributes 
$\vec{y}_{e} = [y_{ea}]_{a = 1 \ldots A}$.
Since $\vec{y}_{e}$ is \emph{unobserved}, it must be treated 
as a random variable over the space of possible attributes 
$\valset_\otimes := \bigotimes_{a=1}^{A} \valset_{a}$.
This means the blocking function should not be interpreted 
as partitioning the entities directly.
Rather, it should be interpreted as partitioning the space 
$\valset_\otimes$ in which the entities reside, while taking 
the distribution over $\valset_\otimes$ into account.
With this interpretation in mind, we argue that the 
blocking function should ideally satisfy the following 
properties:
\begin{enumerate}[(i)]
  \item \emph{Balanced weight.} The blocks should have 
  equal weight (probability mass) under the distribution 
  over $\valset_\otimes$, thereby ensuring the entities are 
  distributed evenly (in expectation) among the blocks.
  This is a desirable property, as it ensures proper load balancing 
  for our distributed inference algorithm (see 
  Section~\ref{sec:distributed-pcg}).
  \item \emph{Entity separation.} A pair of entities drawn at 
  random from the same block should have a high degree of 
  similarity, while entities drawn from different blocks 
  should have a low degree of similarity.
  This improves the likelihood that similar records will 
  end up in the same block, and allows them to more 
  readily form likely entities.
\end{enumerate}

These properties need not be satisfied strictly: the extent to which 
they are satisfied is merely expected to improve the efficiency of 
the inference.
For example, satisfying the first property requires 
knowledge of the marginal posterior distribution 
over $\vec{y}_{e}$, which is infeasible to calculate.
We note that there is likely to be tension between the two properties, 
so that a balance must be struck between them.

\subsection{Example: \emph{k}-d tree blocking function}
\label{sec:kd}
We now describe a blocking function based on $k$-d trees, 
which is used in our experiments in Section~\ref{sec:experiments}.

\paragraph{Background.} A $k$-d tree is a binary tree that 
recursively partitions a $k$-dimensional affine 
space~\citep{bentley_multidimensional_1975, 
  friedman_algorithm_1977}.
In the standard set-up, each node of the tree is associated 
with a data point that implicitly splits the input space 
into two half-spaces along a particular dimension.
Owing to its ability to hierarchically group nearby 
points, it is commonly used to speed up nearest-neighbor 
search.
This makes a $k$-d tree a good candidate for a blocking function, 
since it can be balanced while grouping similar points.

\paragraph{Setup.} Our setup differs from a standard $k$-d 
tree in several aspects.
First, we consider a discrete space $\valset_\otimes$ 
  (not an affine space), where the ``$k$ dimensions'' 
  are the $A$ attributes. 
Second, we do not store data points in the tree.
  We only require that the tree implicitly stores the 
  boundaries of the blocks, so that it can assign an 
  arbitrary $\vec{y} \in \valset_\otimes$ to the 
  correct partition (a leaf node).
Finally, since we are working in a discrete space, 
  the input space to a node is a countable set. 
  The node must split the input set into two parts 
  based on the values of one of the attributes.

\paragraph{Fitting the tree.} Since it is infeasible to calculate 
the marginal posterior distribution over $\vec{y}_{e}$ exactly, 
we use the empirical distribution from the tables as an 
approximation.
As a result, we treat the records (tables) as a sample from the 
distribution over $\vec{y}_{e}$, and fit the tree so that it 
remains balanced with respect to this sample.
The depth of the tree $d$ determines the number of blocks
($2^d$).

\paragraph{Achieving balanced splits.}
When fitting the tree, each node receives an input set of samples 
and a rule must be found that splits the set 
into two roughly equal (balanced) parts based on an attribute.
We consider two types of splitting rules: the \emph{ordered median} and 
the \emph{reference set} (see Appendix~\ref{app-sec:splitting-rules}).
We allow the practitioner to specify an ordered list of attributes to 
be used for splitting.
To ensure balanced splits, we recommend selecting attributes with a 
large domain.
If possible, we recommend preferencing attributes which are known a priori 
to be reliable (low distortion), as this will reduce the shuffling of 
entities\slash records between blocks.
In principle, it is possible to automate the process of fitting a tree: 
one could grow several trees with randomly-selected splits and use the 
one that is most balanced.
We examine balance empirically in Appendix~\ref{app-sec:balance-partitions}.

\section{Inference}
\label{sec:inference}
We now turn to approximating the full joint posterior 
distribution over the unobserved variables $\vec{Z}$, 
$\vec{Y}$, $\vec{\Theta}$, $\vec{\Gamma}$ and $\vec{\Lambda}$, 
as given in Equation~\ref{eqn:partitioned-posterior}.
Since it is infeasible to sample from this distribution 
directly, we design MCMC algorithms based on partially-collapsed Gibbs (PCG)
sampling~\citep{dyk_partially_2008}.
In addition, we show how to exploit the conditional 
independence induced by the blocks to distribute the PCG sampling 
across multiple threads or machines.

\subsection{Partially-collapsed Gibbs sampling}
\label{sec:pcg-sampling}
Following the \blink\ paper~\citep{steorts_entity_2015}, we initially 
experimented with regular Gibbs sampling.\footnote{
  We define \emph{regular} Gibbs sampling as the most basic variation 
  where variables are updated iteratively one-at-a-time by sampling from 
  their conditional distributions.
} 
However, the resulting Markov chains exhibited slow 
convergence and poor mixing.
This is a known shortcoming of Gibbs sampling which 
may be remedied by collapsing variables and\slash or 
updating correlated variables in groups~\citep{liu_monte_2004}.
These ideas form the basis for a framework called 
\emph{partially-collapsed Gibbs (PCG) sampling}---a 
generalization of Gibbs sampling that has 
better convergence properties~\citep{dyk_partially_2008}.

Under the PCG framework, variables are updated in 
groups by sampling from their conditional distributions.
These conditional distributions may be taken with respect 
to the joint posterior (like regular Gibbs), or with 
respect to \emph{marginal distributions} of the joint 
posterior (unique to PCG).
The latter case is called \emph{trimming} and must be 
handled with care so as not to alter the stationary 
distribution of the Markov chain.

In applying PCG sampling to \dblink, we must decide how to 
apply the three tools: 
\emph{marginalization} (equivalent to grouping), 
\emph{permutation} (changing the order of the updates) and 
\emph{trimming} (removing marginalized variables).
In theory, the convergence rate should improve with more 
marginalization and trimming, however this must be balanced 
with the following: 
(i)~whether the resulting conditionals can be sampled from 
efficiently, and 
(ii)~whether the resulting dependence structure is 
compatible with our distributed set-up (see 
Section~\ref{sec:distributed-pcg}).
We consider two samplers, PCG-I and PCG-II, described below.
Of the two, we recommend PCG-I as it is more efficient 
in our empirical evaluations (see Section~\ref{sec:efficiency-expts}).
We include the PCG-II sampler, as one would expect the PCG-II 
sampler to perform better than the PCG-I sampler in terms of mixing, 
however when computational efficiency is taken into account the performance 
is worse (see Figure~\ref{fig:speed-up-vs-sampler}).

\begin{figure*}
  \def\svgwidth{\linewidth}
  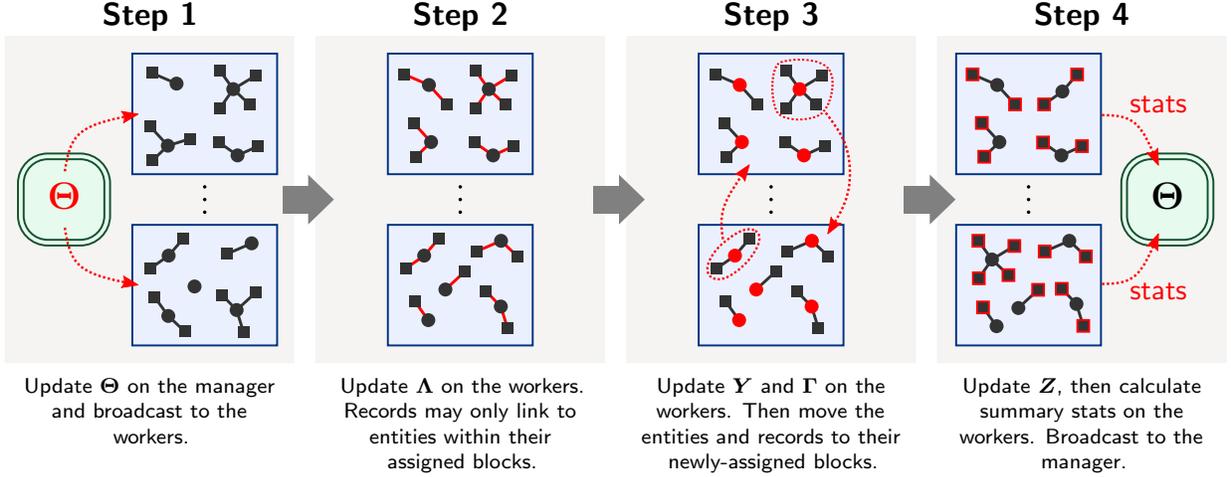
  \caption{Schematic depicting a single iteration of distributed 
    PCG sampling.
    The entity attributes ($\vec{Y}$---circular nodes), 
    record attributes and their distortion indicators 
    ($\vec{X}$, $\vec{Z}$---square nodes),
    and links from records to entities 
    ($\vec{\Lambda}$---node connectors) 
    are distributed across the workers (blue rectangular plates) 
    according to their assigned blocks.
    The distortion probabilities ($\vec{\Theta}$) reside on the 
    manager (green rounded-rectangular plate).
  }
  \label{fig:distributed-pcg}
\end{figure*}

\subsubsection{PCG-I sampler} 
The PCG-I sampler uses regular Gibbs updates for $\theta_{ta}$, 
$\lambda_{tr}$ and $z_{tra}$ for all $t$, $r$ and $a$.
The conditional distributions for these updates are listed in 
Appendix~\ref{app-sec:gibbs}.
When updating the entity attributes $y_{ea}$ and the block 
assignments $\gamma_{tr}$, marginalization and trimming are 
used.
Specifically, we apply marginalization by jointly updating 
$\vec{y}_{e}$ and $\{\gamma_{tr}, \vec{z}_{tr}\}_{\entset_{e}}$ 
(the set of $\gamma_{tr}$'s and $\vec{z}_{tr}$'s for records $(t,r)$ 
linked to entity $e$).
We then trim (analytically integrate over) $\{\vec{z}_{tr}\}_{\entset_{e}}$.

We shall now derive this update. 
Referring to Equation~\ref{eqn:partitioned-posterior}, the joint posterior 
of $\vec{y}_{e}$, $\{\gamma_{tr}, \vec{z}_{tr}\}_{\entset_{e}}$ 
conditioned on the other parameters has the form
\begin{equation*}
\begin{split}
& p\!\left(\vec{y}_{e}, \{\gamma_{tr}, \vec{z}_{tr}\}_{\entset_{e}} \middle| 
  \vec{Z}^{\neg \entset_{e}}, \vec{\Gamma}^{\neg \entset_{e}}, \vec{\Theta}, 
    \vec{\Lambda}, \vec{X}^{(o)}, \vec{O}\right) \propto \\
& \quad \prod_{a} \Big\{ p(y_{ea}|\phi_{a}) \times  \prod_{(t,r) \in \entset_{e}} 
  p(\gamma_{tr}|\vec{Y}) p(\lambda_{tr}|\gamma_{tr}, \vec{Y}) p(z_{tra}|\theta_{ta}) 
    \times \prod_{\substack{(t,r) \in \entset_{e}\\o_{tra}=1}} p(x_{tra} | z_{tra}, \lambda_{tr}, y_{ea}) \Big\},
\end{split}
\end{equation*}
where superscript $\neg \entset_{e}$ denotes exclusion of 
any records $(t,r) \in \entset_{e}$ (those currently linked to entity $e$).
Substituting the distributions and trimming $\{\vec{z}_{tr}\}_{\entset_{e}}$ 
yields
\begin{equation}
p\!\left(\vec{y}_{e}, \{\gamma_{tr}\}_{\entset_{e}} \middle| 
  \vec{Z}^{\neg \entset_{e}}, \vec{\Gamma}^{\neg \entset_{e}}, \vec{\Theta}, 
  \vec{\Lambda}, \vec{X}^{(o)}, \vec{O} \right) = p(\{\gamma_{tr}\}_{\entset_{e}}|\entset_{e}, \vec{y}_{e}) 
  \prod_{a} p(y_{ea}|\entset_{e}, \vec{\Theta}, \vec{X}^{(o)}, \vec{O})
\label{eqn:y-gamma-update}
\end{equation}
where
\begin{equation*}
p(y_{ea}|\entset_{e}, \vec{\Theta}, \vec{X}^{(o)}, \vec{O}) 
  \propto \phi_{a}(y_{ea}) \prod_{\substack{(t,r) \in \entset_{e}\\o_{tra}=1}} 
  \left\{(1 - \theta_{ta}) \1{x_{tra} = y_{ea}} 
  + \theta_{ta} \psi_{a}(x_{tra}|y_{ea}) \right\}
\end{equation*}
and
\begin{equation*}
p(\{\gamma_{tr}\}_{\entset_{e}}|\entset_{e}, \vec{y}_{e}) \propto \prod_{(t,r) \in 
\entset_{e}} \1{\gamma_{tr} = \blockfn(\vec{y}_{e})}.
\end{equation*}
Note that the update for $\{\gamma_{tr}\}_{\entset_{e}}$ is deterministic, 
conditional on $\vec{y}_{e}$ and $\entset_e$.

Since we have applied trimming, we must permute the updates so that 
the trimmed variables $\vec{Z}$ are not conditioned on in later updates.
This means the updates for $\vec{y}_{e}$ and 
$\{\gamma_{tr}, \vec{z}_{tr}\}_{\entset_{e}}$ must come \emph{after} the 
updates for $\theta_{ta}$ and $\lambda_{tr}$, but \emph{before} the updates 
for $z_{tra}$.

\subsubsection{PCG-II sampler}
The PCG-II sampler is identical to PCG-I, except that it replaces the 
regular Gibbs update for $\lambda_{tr}$ with an update that marginalizes 
and trims $\vec{z}_{tr}$.
To derive the distribution for this update, we first consider the 
joint posterior of $\lambda_{tr}$ and $\vec{z}_{tr}$ conditioned on the 
other parameters:
\begin{equation*}
\begin{split}
& p(\lambda_{tr}, \vec{z}_{tr}|\vec{\Gamma}, \vec{Y}, \vec{\Theta}, \vec{Z}^{\neg (t,r)}, \vec{X}^{(o)}, \vec{O}) \propto \\
& \qquad p(\lambda_{tr}|\gamma_{tr}, \vec{Y}) \times \prod_{a} p(z_{tra}|\theta_{ta}) \times \prod_{\substack{a\\o_{tra}=1}} p(x_{tra}|z_{tra}, \lambda_{tr}, y_{\lambda_{tr}a})
\end{split}
\end{equation*}
where superscript $\neg(t,r)$ denotes exclusion of record $(t,r)$. 
Substituting the distributions and trimming $\vec{z}_{tr}$ yields
\begin{equation}
\begin{split}
& p(\lambda_{tr} | \vec{\Gamma}, \vec{Y}, \vec{\Theta}, \vec{Z}^{\neg (t,r)}, \vec{X}^{(o)}, \vec{O}) \propto \\
& \qquad \1{\lambda_{tr} \in \partset_{\gamma_{tr}}(\vec{Y})} 
  \times \prod_{\substack{a\\o_{tra}=1}} 
  \Big\{(1 - \theta_{ta}) \1{x_{tra} = y_{\lambda_{tr}a}} + \theta_{ta} \psi_{a}(x_{tra}|y_{\lambda_{tr}a})\Big\}.
\end{split}
\label{eqn:lambda-pcg-update}
\end{equation}

\begin{table}
  \spacingset{1}
  \centering
  \caption{Dependencies for the conditional updates used in the PCG-I 
    sampler.}
  \label{tbl:conditional-dependencies}
  \begin{center}
  \begin{tabular}{rl}
    \toprule
    Update variables & Dependencies \\
    \midrule
    $\theta_{ta}$ & $z_{t \cdot a} = \sum_{r} z_{tra}$ \\
    $\lambda_{tr}$   & $\vec{z}_{tr}$, $\vec{x}_{tr}$, $\gamma_{tr}$, 
    $\partset_{\gamma_{tr}}$, 
    $\{\vec{y}_{e}\}_{e \in \partset_{\gamma_{tr}}}$ \\
    $y_{ea}$, $\{\gamma_{tr}, z_{tra}\}_{(t,r) \in \entset_{e}}$ &  
    $\entset_{e}$, $\{x_{tra}\}_{(t,r) \in \entset_{e}}$,
    $\{\theta_{ta}\}_{(t,r) \in \entset_{e}}$ \\
    $z_{tra}$     & $x_{tra}$, $\lambda_{tr}$, $y_{\lambda_{tr}a}$, 
    $\theta_{ta}$ \\
    \bottomrule
  \end{tabular}
  \end{center}
\end{table}

\subsection{Distributing the sampling}
\label{sec:distributed-pcg}
By examining the conditional distributions derived in the previous section 
and those listed in Appendix~\ref{app-sec:gibbs}, one can show 
that the updates for the variables associated with entities 
and records ($z_{tra}$, $\lambda_{tr}$, $\gamma_{tr}$ and 
$y_{ea}$) only depend on variables associated 
with entities and records assigned to the \emph{same} 
block (excluding $\vec{\Theta}$).
These dependencies are summarized in 
Table~\ref{tbl:conditional-dependencies} for the PCG-I sampler.
The distortion probability $\theta_{ta}$ is 
an exception---it is not associated with any block 
and may depend on $z_{tra}$'s across \emph{all} blocks.

This dependence structure---in particular, the conditional 
independence of entities and records across blocks---makes 
the PCG sampling amenable to distributed computing.
As such, we propose a manager-worker architecture where: 
\begin{itemize}
  \item the \emph{manager} is responsible for storing and 
  updating variables \emph{not} associated with any 
  block (i.e.\ $\vec{\Theta}$); and
  \item each \emph{worker} represents a block, and is 
  responsible for storing and updating variables associated 
  with the entities and records assigned to it.
\end{itemize}
The manager\slash workers may be processes running on a single machine 
or on machines in a cluster.
If using a cluster, we recommend that the nodes be tightly 
coupled, as frequent communication between them is required.

Figure~\ref{fig:distributed-pcg} depicts a single 
iteration of PCG sampling using our proposed 
manager-worker architecture. 
Of the four steps depicted, steps~2 and~3---where the links, 
entity attributes and block assignments are updated---are the most 
computationally intensive. 
We therefore expect to achieve a significant speed-up by 
distributing these steps across the workers.

To ensure good load balancing of these steps it is important 
that the blocks are well-balanced (see Section~\ref{sec:interpret}), 
otherwise workers responsible for smaller blocks must wait idly 
for other workers to finish before the next iteration can begin. 
This is because step~1 requires global synchronization of state 
across the workers.
The blocks also have an effect on communication costs, which 
are most significant in step~3, where the entities and linked records 
are shuffled to their newly-assigned blocks.
A well-chosen blocking function can minimize this cost, by ensuring 
similar records\slash entities are co-blocked.

\section{Computational efficiency considerations}\label{sec:tricks}

\subsection{Efficient pruning of candidate links}
\label{sec:pruning-links}
In this section, we describe a trick that is aimed at improving the 
computational efficiency of the Gibbs update for $\lambda_{tr}$ 
(used in the Gibbs and PCG-I samplers).
This particular trick does \emph{not} apply to the joint PCG update 
for $\lambda_{tr}$ and $\vec{z}_{tr}$ (used in the PCG-II sampler).

Consider the conditional distribution for the 
$\lambda_{tr}$ update in Equation~\ref{app-eqn:lambda-update} of 
Appendix~\ref{app-sec:gibbs}:
\begin{equation}
\begin{split}
& p(\lambda_{tr} = e | \vec{\Gamma}, \vec{Y}, \vec{Z}, \vec{X}^{(o)}, \vec{O}) \propto \\
& \qquad \1{e \in \partset_{\gamma_{tr}}(\vec{Y})}
  \times \prod_{\substack{a\\o_{tra}=1}} \Big\{(1 - z_{tra}) \1{x_{tra} = y_{ea}} 
  + z_{tra} \psi_{a}(x_{tra}|y_{ea}) \Big\}.
\end{split}
\label{eqn:lambda-km}
\end{equation}
The support of this distribution is the \emph{set of candidate 
  links} for record $(t,r)$, which we denote by $\mathcal{L}_{tr}$.
Looking at the first indicator function above, we see that 
$\mathcal{L}_{tr} \subseteq \partset_{\gamma_{tr}}$, 
i.e.\ the candidate links are restricted to the entities 
in the \emph{same block} as record $(t,r)$. 
Thus, a na\"{i}ve sampling approach for this distribution 
takes $O(|\partset_{\gamma_{tr}}|)$ time.

We can improve upon the na\"{i}ve approach by exploiting 
the fact that $\mathcal{L}_{tr}$ is often considerably 
smaller than $\partset_{\gamma_{tr}}$.  
To see why this is the case, note that the second indicator function 
in Equation~\ref{eqn:lambda-km} 
further restricts $\mathcal{L}_{tr}$ if any of the 
distortion indicators for the observed record attributes 
are zero.
Specifically, if $z_{tra} = 0$ and $o_{tra}=1$, $\mathcal{L}_{tr}$ cannot 
contain any entity  whose $a$-th attribute $y_{ea}$ 
does not match the record's $a$-th attribute $x_{tra}$.
This implies $\mathcal{L}_{tr}$ is likely to be small in 
the case of low distortion.

Putting aside the computation of $\mathcal{L}_{tr}$ for the moment, 
this means we can reduce the time required to update 
$\lambda_{tr}$ to $O(|\mathcal{L}_{tr}|)$.
To compute $\mathcal{L}_{tr}$ efficiently, we propose maintaining 
an inverted index over the entity attributes within each block.
Specifically, the index for the $a$-th attribute in 
block~$b$ should accept a query value $v \in \valset_{a}$ 
and return the set of entities that match on $v$:
\begin{equation}
\mathcal{M}_{pa}(v) = \{n \in \partset_{p}: y_{ea} = v\}.
\end{equation}
Once the index is constructed, we can efficiently retrieve the 
set of candidate links for record $(t,r)$ by computing a 
multiple set intersection:
\begin{equation}
\mathcal{L}_{tr} = \bigcap_{\{a : z_{tra} = 0 \wedge o_{tra} = 1\}}
\mathcal{M}_{\gamma_{tr}a}(x_{tra}).
\label{eqn:set-intersection}
\end{equation}
This assumes at least one of the observed record attributes 
is not distorted. 
Otherwise $\mathcal{L}_{tr} = \partset_{\gamma_{tr}}$.

Since the sizes of the sets $\mathcal{M}_{\gamma_{tr}a}(x_{tra})$ 
are likely to vary significantly, we advise computing 
the intersection iteratively in increasing order of size.
That is, we begin with the smallest set and retain the elements 
that are also in the next largest set, and so on.
With a hash-based set implementation, this scales linearly 
in the size of the first (smallest) set.

\subsection{Caching and truncation of attribute similarities}
\label{sec:trunc-sim}
We have not yet emphasized that the updates for $\vec{\Lambda}$, 
$\vec{Y}$ and $\vec{\Gamma}$ depend on the attribute 
similarities between pairs of values in the attribute domains.
Specifically, for each attribute $a$, we need to access 
the indexed set $\mathcal{S}_{a} = 
  \{\simfn_{a}(v,w): v, w \in \valset_{a} \times \valset_{a} \}$.
These similarities may be expensive to evaluate on-the-fly, 
so we cache the results in memory on the workers.

To manage the quadratic scaling of $\mathcal{S}_{a}$, and 
in anticipation of another trick introduced in 
Section~\ref{sec:pert-sampling}, we transform the similarities so that 
those \emph{below} a cut-off $s_{\mathrm{cut};a}$ are regarded 
as completely disagreeing.
We achieve this by applying the following truncation transformation 
to the raw attribute similarity $\simfn_{a}(v,w)$:
\begin{equation}
\truncsimfn_{a}(v,w) = 
  \max \left(0, \ \frac{\simfn_{a}(v,w) - s_{\mathrm{cut};a}}
    {1 - s_{\mathrm{cut};a}/s_{\mathrm{max};a}} \right).
\end{equation}
as illustrated in Figure~\ref{fig:truncated-sim-fn}.
\begin{figure}
  \centering
  \def\svgwidth{0.45\linewidth}
  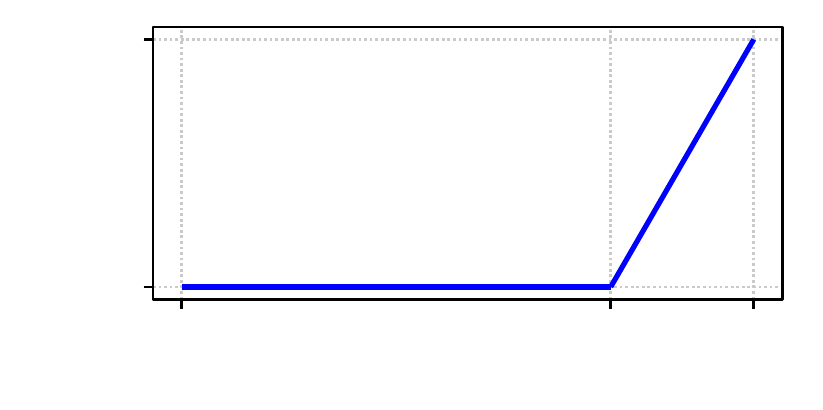
  \caption{Transformation from a raw similarity function 
    ($\simfn$) to a truncated similarity 
    function ($\truncsimfn$).}
  \label{fig:truncated-sim-fn}
\end{figure}
Whenever a raw attribute similarity is called for, we 
replace it with this truncated version.
Only pairs of values with positive truncated similarity are stored 
in the cache---those not stored in the cache have a truncated 
similarity of zero by default.
Note that attributes with a constant similarity function 
$\simfn_{\mathrm{const}}$ are treated specially---there is no 
need to cache the index set of similarities, since they are all 
identical.

It is important to acknowledge that the truncated similarities 
are an approximation to the original model.
We claim that the approximation is reasonable on the following grounds:
\begin{itemize}
  \item \emph{Low loss of information.} 
  Below a certain cut-off, the attribute similarity function is
  unlikely to encode much useful information for modeling the 
  distortion process.
  For example, the fact that
  $\simfn_\mathrm{nEd}(\text{``Smith''},\text{``Chiu''}) = 0.385$ 
  whereas 
  $\simfn_\mathrm{nEd}(\text{``Smith''},\text{``Chen''}) = 0.286$, 
  doesn't necessarily suggest that ``Chiu'' is more likely 
  than ``Chen'' as a distorted alternative to ``Smith''.
  \item \emph{Precedent.}
  In the record linkage literature, value pairs with similarities 
  below a cut-off are regarded as completely 
  disagreeing~\citep{winkler_methods_2002, enamorado_using_2019}.
  \item \emph{Efficiency gains.} 
  As we shall soon see in Section~\ref{sec:pert-sampling}, 
  we can perform the combined $\vec{Y}$, $\vec{\Gamma}$, $\vec{Z}$ update 
  more efficiently by eliminating pairs below the cut-off from 
  consideration. 
\end{itemize}

\subsection{Fast updates of entity attributes using perturbation sampling}
\label{sec:pert-sampling}
We now present a novel sampling algorithm that allows us to 
efficiently perform the PCG update for $y_{ea}$ and 
$\{\gamma_{tr}, z_{tra}\}_{\entset_{e}}$.
The algorithm relies on the observation that the 
conditional distribution for $y_{ea}$ can be 
expressed as a mixture over two components:
\begin{enumerate}[(i)]
  \item a \emph{base distribution} over $\valset_{a}$ which is 
  ideally constant for all entities; and
  \item a \emph{perturbation distribution} which varies for each 
  entity, but has a much smaller support than $\valset_{a}$.
\end{enumerate}
With this representation, we can avoid computing and 
sampling from the full distribution over 
$\valset_{a}$, which varies for each $y_{ea}$ update.
Rather, we only need to compute the perturbation 
distribution over a much smaller support, and then 
sample from the mixture, which can be done efficiently 
using the Vose-Alias method~\citep{vose_linear_1991}.
We refer to this algorithm as \emph{perturbation sampling}.

\subsubsection{Perturbation sampling}
Although we're interested in applying perturbation sampling 
to a specific conditional distribution, we describe the idea in 
generality below. 

Consider a target probability mass function (pmf) $p(x|\omega)$
with finite support $\mathcal{X}$, which varies as a function of 
parameters $\omega \in \Omega$.
In general, one must recompute the probability tables 
to draw a new variate whenever $\omega$ changes---a 
computation that takes $O(|\mathcal{X}|)$ time.
However, if the dependence on $\omega$ is of a certain 
restricted form, we show that it is possible to achieve 
better scalability by expressing the target as a mixture.
This is made precise in the following result.
\begin{proposition}
  \label{thm:pert-sampling}
  Let $p(x|\omega)$ be a pmf with finite support $\mathcal{X}$, 
  which depends on parameters $\omega \in \Omega$.
  Suppose there exists a ``base'' pmf $q(x)$ over $\mathcal{X}$ 
  which is independent of $\omega$ and a non-negative bounded 
  perturbation term $\epsilon(x|\omega)$, such that $p(x|\omega)$ 
  can be factorized as 
  $p(x|\omega) \propto q(x) (1 + \epsilon(x|\omega))$.
  Then $p(x|\omega)$ can be expressed as a mixture over the base 
  pmf $q(x)$ and a ``perturbation'' pmf $v(x|\omega) := c \, q(x) 
  \epsilon(x|\omega)$ over $\mathcal{X}^\star = \{x \in \mathcal{X}: 
  \epsilon(x|\omega) > 0\}$ as follows:
  \begin{equation}
  p(x|\omega) = \frac{c}{1 + c} q(x) + \frac{1}{1 + c} v(x|\omega)
  \end{equation}
  where $c^{-1} := 
    \sum_{x \in \mathcal{X}^\star} q(x) \epsilon(x|\omega)$.
\end{proposition}
\begin{proof}
  The result is straightforward to verify by substitution.
\end{proof}

Algorithm~\ref{app-alg:pert-sampling} (in Appendix~\ref{app-sec:pert-sampling})
shows how to apply this result to draw random variates from a target pmf.
Briefly, it consists of three steps:
(i)~the perturbation pmf $v$ and its normalization 
constant $c$ are computed;
(ii)~a biased coin is tossed to choose between the base 
pmf $q$ and the perturbation pmf $v$; and
(iii)~a random variate is drawn from the selected pmf.
If $q$ is selected, a pre-initialized Alias sampler is used to 
draw the random variate (reused for all $\omega$).
Otherwise if $v$ is selected, a new Alias sampler is instantiated. 
The result below states the time complexity of this algorithm 
(see Appendix~\ref{app-sec:pert-sampling} for a proof).
\begin{proposition}
  \label{thm:complexity-pert-sampling}
  Algorithm~\ref{app-alg:pert-sampling} (in 
  Appendix~\ref{app-sec:pert-sampling}) returns a random variate
  from the target pmf $p(x|\omega)$ for any $\omega \in \Omega$ 
  in $O(|\mathcal{X}^\star|)$ time.
\end{proposition}
This is a promising result, since the size of the 
perturbation support $|\mathcal{X}^\star|$ is typically of 
order 10 for our application, while the size of the full
 support $|\mathcal{X}|$ may be as large as $10^5$.
Hence, we expect a significant speed-up over the na\"{i}ve 
approach.

\subsubsection{Application of perturbation sampling}
We now return to our original objective: performing 
the joint PCG update for $y_{ea}$ and $\{\gamma_{tr}, 
z_{tra}\}_{\entset_{e}}$.
Referring to Equation~\ref{eqn:y-gamma-update}, we can express the 
conditional distribution for $y_{ea}$ (i.e. the target distribution) as
\begin{equation}
p(y_{ea} = v|\entset_{e}, \vec{\Theta}, \vec{X}^{(o)}, \vec{O}) \propto 
  q_{a}(v|\entset_{e}, \mathbf{O}) \left(1 + 
  \epsilon_{a}(v|\entset_{e}, \vec{\Theta}, \vec{X}^{(o)}, \vec{O}) \right).
\label{eqn:y-gamma-perturb}
\end{equation}
The base distribution is given by
\begin{equation}
q_{a}(v|\entset_{e}, \vec{O}) \propto 
\phi_{a}(v) \left( h_{a}(v) \right)^{n_a(\entset_{e}, \vec{O})}
\label{eqn:y-gamma-base}
\end{equation}
where $n_a(\entset_{e}, \vec{O}) = |\{(t,r) \in \entset_{e}: o_{tra}=1\}|$ 
is the number of records linked to entity $e$ with observed values for 
attribute $a$;
and the perturbation term is given by
\begin{equation}
\epsilon_{a}(v|\entset_{e}, \vec{\Theta}, \{x_{tra}\}_{\entset_{e}}) = 
  \prod_{\substack{(t,r) \in \entset_{e}\\o_{tra}=1}} \Bigg\{ 
  \euler^{\truncsimfn_{a}(x_{tra}, v)} + 
  \frac{(\theta_{ta}^{-1} - 1) \, \1{x_{tra} = v}}{\phi_{a}(x_{tra}) h_{a}(x_{tra})} 
  \Bigg\} - 1.
\label{eqn:y-gamma-pert-term}
\end{equation}

The full support of the target pmf is $\valset_{a}$, 
while the perturbation support is given by
\begin{equation*}
\{x_{tra}: (t,r) \in \entset_{e} \wedge o_{tra} = 1\} \ \cup 
  \{v \in \valset_{a}: \truncsimfn_{a}(v, x_{tra}) > 0 \wedge 
  o_{tra} = 1 \text{ for any } (t,r) \in \entset_{e}\}.
\end{equation*}
In words, this set consists of the observed values for attribute~$a$ 
in the records linked to entity~$e$, plus any 
sufficiently similar values from the attribute domain (for which 
the truncated similarity is non-zero).
The size of the perturbation set will vary depending 
on the cut-off used for the truncation transformation---the 
higher the cut-off, the smaller the set.
This implies that there is a trade-off between 
efficiency (small perturbation set) and accuracy 
(lower cut-off).

\begin{remark}
  The astute reader may have noticed that the base 
  distribution $q_{a}$ given in Equation~\ref{eqn:y-gamma-base} 
  is not completely independent of the conditioned parameters, 
  as is required by Proposition~\ref{thm:pert-sampling}.
  In particular, $q_{a}$ depends on $n_a(\entset_{e}, \vec{O})$---roughly
  the size of entity $e$.
  Fortunately, we expect the range of regularly 
  encountered entity sizes to be small, so we 
  sacrifice some memory by instantiating multiple 
  Alias samplers for each $n_a(\entset_{e}, \vec{O})$ in some 
  expected range.
  In the worst case, when a value is encountered 
  outside the expected range and the base distribution 
  is required (unlikely since the weight on the base component is 
  typically small),
  we instantiate the base distribution on-the-fly 
  (same asymptotic cost as the na\"{i}ve approach).
\end{remark}


%
\section{Empirical evaluation}
\label{sec:experiments}

\newcolumntype{Y}{>{\centering\arraybackslash}X}
\begin{table*}
  \spacingset{1}
  \centering
  \caption{Summary of data sets. 
    Those marked with a `$\star$' are synthetic.}
  \footnotesize
  \label{tbl:data-sets}
  \begin{center}
  \begin{tabularx}{\textwidth}{l *{3}{c} *{2}{Y}}
    \toprule
    Data set & \# records ($R$) & \# tables ($T$) & \# entities & 
    \multicolumn{2}{c}{\# attributes ($A$)} \\ 
    \cmidrule{5-6}
    & & & & {\footnotesize categorical} & {\footnotesize string } \\
    \midrule
    $\star$ \texttt{ABSEmployee} & 660,000 & 3 & 400,000 & 4 & 0 \\
    \texttt{NCVR}        & 448,134 & 2 & 296,433 & 3 & 3 \\
    \texttt{NLTCS}       &  57,077 & 3 &  34,945 & 6 & 0 \\
    \texttt{SHIW0810}    &  39,743 & 2 &  28,584 & 8 & 0 \\
    $\star$ \texttt{RLdata10000} &  10,000 & 1 &   9,000 & 2 & 3 \\
    \bottomrule
  \end{tabularx}
  \end{center}
\end{table*}

We present an evaluation of \dblink\ using two synthetic and three real 
data sets, as summarized in Table~\ref{tbl:data-sets}. 
The data sets include applications such as ER of employees in 
administrative and survey data (\texttt{ABSEmployee}), ER of voters in 
registration databases (\texttt{NCVR}) and ER of respondents in 
anonymized survey data (\texttt{SHIW0810}).
All results presented here were obtained using a local server in 
pseudocluster mode, however some were replicated on a cluster in the 
Amazon public cloud (see Appendix~\ref{app-sec:cloud-results}) to test 
the effect of higher communication costs.
Further details about the data sets, hardware, implementation 
and parameter settings are provided in 
Appendix~\ref{app-sec:experiments-setup}.

\subsection{Computational and sampling efficiency}
\label{sec:efficiency-expts}
Following \cite{turek_efficient_2016}, we measured the efficiency using the 
rate of effective samples produced per unit time (ESS rate), which 
balances sampling efficiency (related to mixing\slash autocorrelation) and 
computational efficiency.
We used the \texttt{mcmcse} R package~\citep{mcmcse} to compute the 
effective sample size (ESS), which implements a multivariate method 
proposed by \citet{vats_multivariate_2019}.

Since the number of variables in the model is unwieldy (there are at least 
$(E + R + T) A + R$ unobserved variables) we computed the ESS for the 
following summary statistics:
\begin{itemize}
  \item the number of observed entities (scalar);
  \item the aggregate distortion for each attribute (vector); and
  \item the cluster size distribution (vector containing 
  frequency of 0-clusters, 1-clusters, 2-clusters, etc.).
\end{itemize}

\paragraph{\dblink\ versus \blink.}
We compared \dblink\ (using the PCG-I sampler) to our own implementation of 
\blink\ (i.e.\ a Gibbs sampler without any of the tricks described in 
Section~\ref{sec:tricks}).
For a fair comparison, we switched off blocking in \dblink.
We used the relatively small \texttt{RLdata10000} data set, as \blink\ 
cannot cope with larger data sets.
Figure~\ref{fig:convergence-vs-impl} contains trace plots for two 
summary statistics as a function of running time.
It is evident that \blink\ has not converged to the equilibrium 
distribution within the allotted time of 11 hours, while \dblink\ 
converges to equilibrium in 100 seconds.
Looking solely at the time per iteration, \dblink\ is at least 200$\times$ 
faster than \blink.

\begin{figure}[t]
  \centering
  \includegraphics{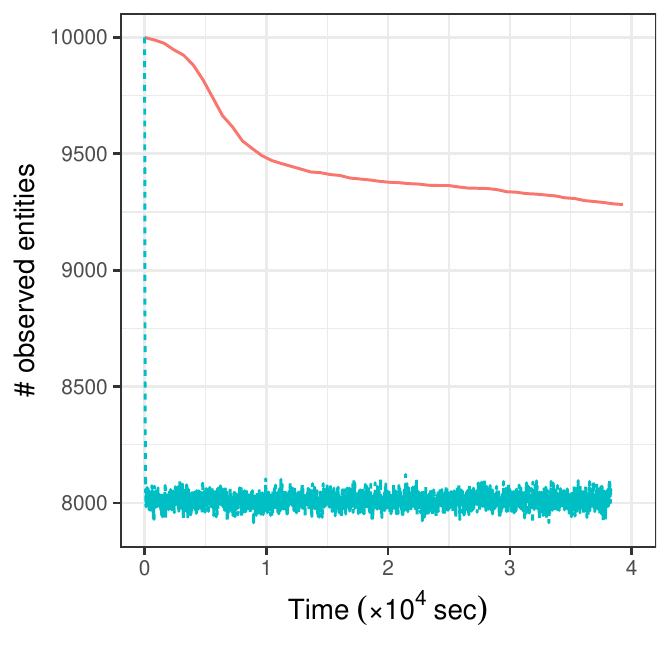} \quad
  \includegraphics{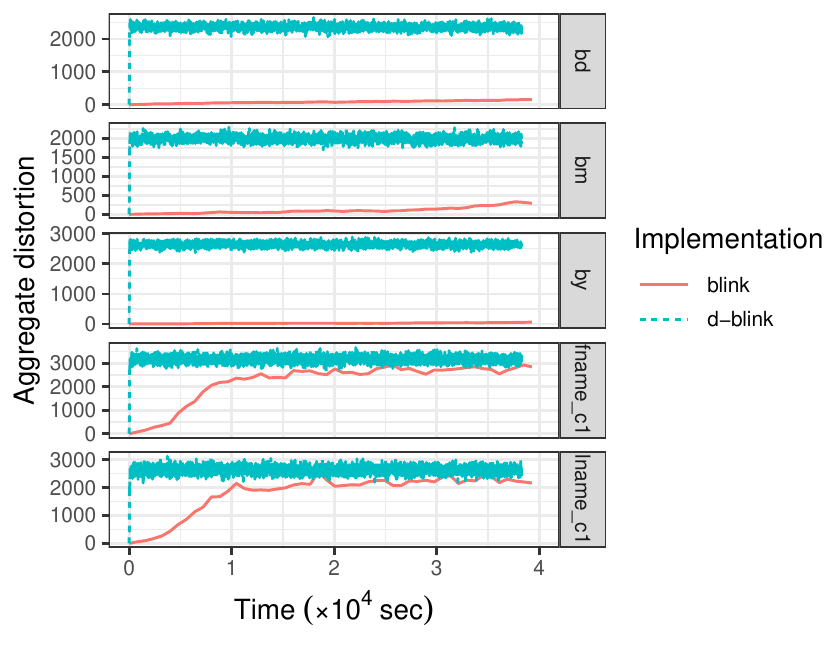}
  \caption{Comparison of convergence rates for \dblink\ and 
  \blink. 
  The summary statistics for \dblink\ (number of observed entities 
  on the left and attribute distortions on the right) rapidly converge 
  to equilibrium, 
  while those for \blink\ fail to converge within 11 hours.}
  \label{fig:convergence-vs-impl}
\end{figure}

\paragraph{Blocking and efficiency.}
We tested the effect of varying the number of blocks $B$ on the efficiency 
of \dblink.
For each value of $B$, we computed the ESS rate 
averaged over 3000 iterations.
We used the \texttt{NLTCS} data set and the PCG-I sampler.
Figure~\ref{fig:speed-up-vs-num-partitions} presents the results in terms of 
the speed-up relative to the ESS rate for $B = 1$. 
We observe a near-linear speed-up in $B$, with the exception of $B=32$.
The speed-up is expected to taper off with increasing numbers of blocks, 
as parallel gains in efficiency are overcome by losses due to communication 
costs and\slash or poorer mixing.
This tipping point seems difficult to predict for a given set up, as it 
depends on complex factors such as the data distribution, the splitting rules 
used, and the hardware characteristics.

\begin{figure}[t]
  \centering
  \includegraphics{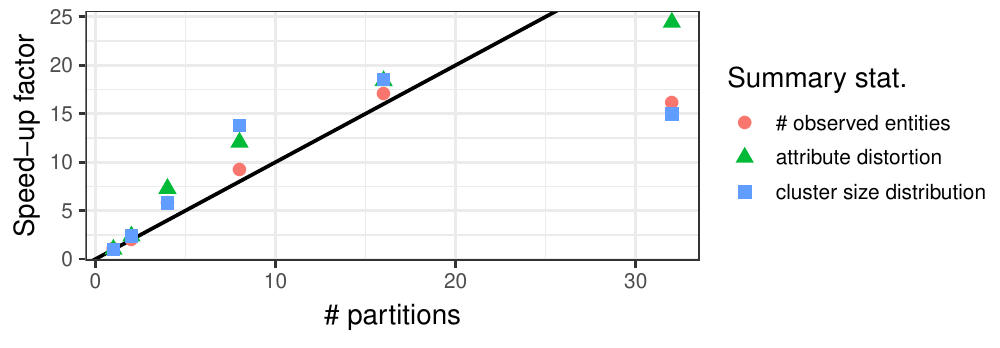}
  \caption{Efficiency of \dblink\ as a function of the
    number of blocks $B$ and summary statistic of interest 
    (larger is better).
    The speed-up measures the ESS rate relative to the ESS rate 
    for $B = 1$ (no blocking) for the \texttt{NLTCS} data set.}
  \label{fig:speed-up-vs-num-partitions}
\end{figure}

\paragraph{Sampling methods and efficiency.}
We evaluated the efficiency of the three samplers introduced in 
Section~\ref{sec:pcg-sampling} (Gibbs, PCG-I and PCG-II).
As above, we computed the ESS rate as an average over 3000 iterations.
We set $B = 16$ and used the \texttt{NLTCS} data set.
The results, shown in Figure~\ref{fig:speed-up-vs-sampler}, 
indicate that the PCG-I sampler is considerably more efficient (by a factor 
of 1.5--2$\times$) than the baseline Gibbs sampler for this data set.
We also observe that the PCG-II sampler performs quite poorly in comparison: 
between 20--30$\times$ slower than the Gibbs sampler.
This is because the marginalization and trimming for the $\vec{\Lambda}$ 
update for PCG-II prevents us from applying the trick described in 
Section~\ref{sec:pruning-links}.
Thus although PCG-II is expected to be more efficient in terms of 
reducing autocorrelation, it is less efficient overall as each iteration 
is too computationally expensive.

\begin{figure}[t]
  \centering
  \includegraphics{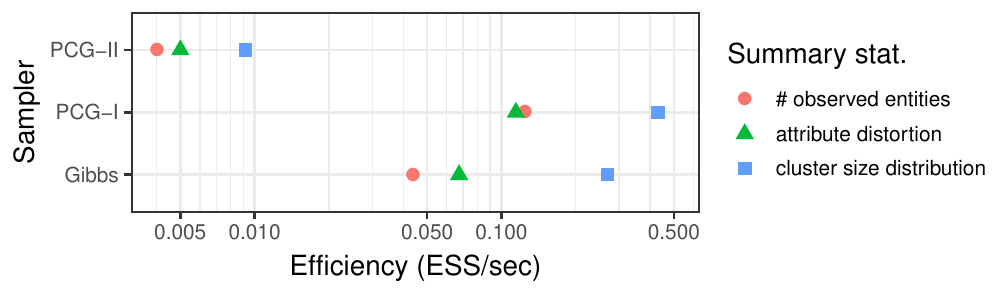}
  \caption{Efficiency of \dblink\ as a function of the 
    sampler and summary statistic of interest (larger is better).
    All measurements are for the \text{NLTCS} data set with 
    $B = 16$.
    }
  \label{fig:speed-up-vs-sampler}
\end{figure}

\subsection{Linkage quality}
\label{sec:linkage-quality}
Though not our primary focus, we assessed the performance of \dblink\ 
in terms of its predictions for the linkage structure (the matching step) 
for the data sets in Table~\ref{tbl:data-sets}.
This was not previously possible with \blink, as it could 
only scale to small data sets of around 1000 records.

\paragraph{Point estimate methodology.}
To evaluate the matching performance of \dblink\ 
with respect to the ground truth, we extracted a point 
estimate of the linkage structure from the posterior using 
the \emph{shared most probable maximal matching sets (sMPMMS)} 
method~\citep{steorts_bayesian_2016}.
This method circumvents the problem of label 
switching~\citep{jasra_markov_2005}---where the identities 
of the entities do not remain constant along the 
Markov chain.

The sMPMMS method involves two main steps.
In the first step, the most probable entity cluster is 
computed for each record based on the posterior samples.
In general, these entity clusters will conflict with one 
another---e.g.\ the most probable entity cluster for $r_1$ 
might be $(r_1, r_2)$ while for $r_2$ it is $(r_1, r_2, r_3)$.
The second step resolves these conflicts by assigning 
precedence to links between records and their most probable 
entity clusters.
The result is a globally-consistent estimate of the linkage 
structure---i.e.\ it satisfies transitivity.

We distributed the computation of the sMPMMS method in 
the Spark framework.
We used $9000$ approximate posterior samples which were 
derived from a Markov chain of length $10^5$ by 
discarding the first $10^4$ iterations as 
burn-in\footnote{We applied a burn-in of 210k iterations 
  for \texttt{NCVR} as it was slow to converge.} 
and applying a thinning interval of 10.
These parameters were chosen by inspection of trace plots, 
some of which are reported in Appendix~\ref{app-sec:trace-plots}.
By contrast to the point estimates reported here, we also examined full 
posterior estimation in Appendix~\ref{app-sec:trace-plots}.

\paragraph{Baseline methods.}
We compared the linkage quality of \dblink\ with three 
baseline methods as described below.
We focused on (scalable) unsupervised methods as we assumed very little to no
training data was available.
\begin{table}[t]
	\centering
	\caption{Comparison of matching quality.
		``ARI'' stands for adjusted Rand index and ``Err. \# clust.'' 
		is the percentage error in the number of clusters.}
	\label{tbl:linkage-quality}
	\spacingset{1}
  \footnotesize
  \begin{center}
	\begin{tabular}{l l *{5}{c}}
		\toprule
		Data set    & Method & \multicolumn{3}{c}{Pairwise measures} & \multicolumn{2}{c}{Cluster measures} \\
		\cmidrule(lr){3-5} \cmidrule(lr){6-7}
		& & Precision & Recall & F1-score & ARI & Err. \# clust.\ \\
		\midrule
		\multirow{5}{*}{\texttt{ABSEmployee}} 
		& \dblink\              & 0.9763 & 0.8530 & \textbf{0.9105} & \textbf{0.9105} & \textbf{+1.667\%} \\
		& Fellegi-Sunter (10)   & \textbf{0.9963} & 0.8346 & 0.9083 & --- & --- \\
		& Fellegi-Sunter (100)  & \textbf{0.9963} & 0.8346 & 0.9083 & --- & --- \\
		& Near Matching         & 0.0378 & \textbf{0.9930} & 0.0728 & --- & --- \\
		& Exact Matching        & 0.9939 & 0.8346 & 0.9074 & 0.9074 & +9.661\% \\
		\midrule
		\multirow{5}{*}{\texttt{NCVR}}
		& \dblink\              & 0.9146 & \textbf{0.9654} & \textbf{0.9393} & \textbf{0.9392} & \textbf{--3.587\%}\\
		& Fellegi-Sunter (10)   & 0.9868 & 0.7874 & 0.9083  & --- & ---\\
		& Fellegi-Sunter (100)  & 0.9868 & 0.7874 & 0.9083 & --- & --- \\
		& Near Matching         & 0.9899 & 0.7443 & 0.8497 & --- & --- \\
		& Exact Matching        & \textbf{0.9925} & 0.0017 & 0.0034 & 0.0034 & +51.09\% \\
		\midrule
		\multirow{5}{*}{\texttt{NLTCS}}
		& \dblink\              & 0.8319 & 0.9103 & 0.8693 & 0.8693 & --22.09\% \\
		& Fellegi-Sunter (10)   & \textbf{0.9094} & 0.9087 & \textbf{0.9090} & --- & --- \\
		& Fellegi-Sunter (100)  & \textbf{0.9094} & 0.9087 & \textbf{0.9090} & --- & --- \\
		& Near Matching         & 0.0600 & \textbf{0.9563} & 0.1129 & --- & --- \\
		& Exact Matching        & 0.8995 & 0.9087 & 0.9040 & \textbf{0.9040} & \textbf{+2.026\%} \\ 
		\midrule
		\multirow{5}{*}{\texttt{SHIW0810}}
		& \dblink\              & \textbf{0.2514} & 0.5396 & \textbf{0.3430} & \textbf{0.3429} & --37.65\% \\
		& Fellegi-Sunter (10)   & 0.0028 & 0.9050 & 0.0056 & --- & --- \\
		& Fellegi-Sunter (100)  & 0.0025 & \textbf{0.9161} & 0.0050 & --- & --- \\
		& Near Matching         & 0.0043 & 0.9111 & 0.0086 & --- & --- \\
		& Exact Matching        & 0.1263 & 0.7608 & 0.2166 & 0.2166 & \textbf{--37.40\%} \\ 
		\midrule
		\multirow{5}{*}{\texttt{RLdata10000}}
		& \dblink\              & 0.6334 & \textbf{0.9970} & 0.7747 & \textbf{0.7747} & 
		\textbf{--10.97\%} \\
		& Fellegi-Sunter (10)   & 0.9957 & 0.6174 & 0.7622 & --- & --- \\
		& Fellegi-Sunter (100)  & 0.9364 & 0.8734 & 0.9038 & --- & --- \\
		& Near Matching         & 0.9176 & 0.9690 & \textbf{0.9426} & --- & --- \\
		& Exact Matching        & \textbf{1.0000} & 0.0080 & 0.0159 & 0.0159 & +11.02\% \\ 
		\bottomrule
  \end{tabular}
  \end{center}
\end{table}

\begin{itemize}
  \item \emph{Exact Matching.} Links records that match on 
  all $A$ attributes.
  It is unsupervised and ensures transitivity.
  \item \emph{Near Matching.} Links records that match on 
  at least $L - 1$ attributes.
  It is unsupervised, but does not guarantee transitivity.
  \begin{sloppypar}
  \item \emph{Fellegi-Sunter.} Links records according to a 
  pairwise match score that is a weighted sum of attribute-level 
  dis\slash agreements.
  The weights are specified by the Fellegi-Sunter 
  model~\citep{fellegi_theory_1969} and were estimated using the 
  expectation-maximization algorithm, as implemented in the 
  \texttt{RecordLinkage} R package~\citep{sariyar_recordlinkage_2010}.
  We chose the threshold on the match score to optimize 
  the F1-score using a small amount of training data (size 10 and 100). 
  This makes the method semi-supervised.
  Note that the training data was sampled in a biased manner 
  to deal with the imbalance between the matches and non-matches
  (half with match scores above zero and half below).
  The method does not guarantee transitivity.
  \end{sloppypar}
\end{itemize}

\paragraph{Results.}
Table~\ref{tbl:linkage-quality} presents performance measures 
categorized by data set and method.
The pairwise performance measures (precision, recall and 
F1-score) are provided for all methods, however 
the cluster performance measures 
(adjusted Rand Index, see~\citealp{vinh_information_2010}, and percentage error in 
the number 
of clusters) are only valid for methods that guarantee transitivity of 
closure (\dblink\ and Exact Matching). 
Despite being fully unsupervised, \dblink\ achieves competitive 
performance when compared to the semi-supervised Fellegi-Sunter 
method.
The two simple baselines, Near Matching and Exact Matching, 
are acceptable for data sets with low noise but perform 
poorly otherwise (e.g.\ \texttt{NCVR} and \texttt{RLdata10000}).
We conducted an empirical sensitivity analysis for \dblink\ 
with respect to variations in the hyperparameters.
The results for \texttt{RLdata10000} (included in 
Appendix~\ref{app-sec:sensitivity}) show that \dblink\ is somewhat 
sensitive to all of the hyperparameters tested, however sensitivity is 
in general predictable, following clear and intuitive trends.
One interesting observation is the fact that \dblink\ tends to overestimate 
the amount of distortion.
This is perhaps not surprising given the absence of ground truth.

\section{Application to the 2010 U.S.\ Decennial Census}
\label{sec:decennial}
National statistics agencies frequently need to link inter- or intra-agency 
data sets, for a number of purposes such as quality control. 
One critical problem in the United States (U.S.) occurs every ten years, when 
the U.S.\ Census Bureau must enumerate the population in each State as 
mandated under the U.S.\ Constitution, Article I, Section 2. 
The enumeration is used to apportion the representation of legislators, 
and to allocate resources for housing, highways, schools, assistance 
programs, and other projects that are vital to the prosperity, welfare, 
and economic growth of the U.S. 
As the country grows and becomes more diverse, it becomes more challenging 
to produce an accurate enumeration. 
Many individuals elect not to fill out census forms, which results in 
them not being counted in the enumeration. 
Other individuals may be counted multiple times due to duplicate responses. 
For example, students attending universities or private schools (living 
in group quarters) are often double counted as they are legally required 
to be counted by their university\slash school, while also being counted 
by their parents\slash guardians as part of a household. 

Motivated by these data duplication issues, we apply \dblink\ to conduct 
an enumeration in the state of Wyoming. 
In order to improve coverage, we combine records from the 
2010 Decennial Census with administrative records from the 
Social Security Administration's Numerical Identification System 
(Numident).\footnote{The Numident is the Social Security Administration's 
computer database file of an abstract of the information contained in an 
application for a U.S.\ Social Security number.} 
In total, we consider 1,050,000 records representing the population of 
Wyoming: 
a subset of 494,000 records from the 2010 Decennial Census and 556,000 
records from the Numident.\footnote{These figures 
have been rounded to the nearest thousand as they are protected under 
Title~13.}
Our goal is to recover the unique individuals represented in these records
using unsupervised ER.
We apply \dblink\ using the overlapping attributes from the Census and the Numident: first and last name, date of birth, gender, and zip code. 
We treat first and last name as string-type attributes and the remaining 
attributes as categorical. 
To manage scalability, we utilize the $k$-d tree blocking function outlined in Section~\ref{sec:kd}, 
splitting recursively on gender and birth year at each level of the tree. 
Inference and MCMC diagnostics are discussed in Appendix~\ref{app-sec:census-mcmc}.

After performing ER using \dblink, we are able to provide a posterior 
estimate of the total number of unique individuals represented in both 
data sets. 
Table~\ref{tbl:census-results} reports a point estimate based on the mean. 
The standard error is quite narrow, which is consistent with 
knowledge of the uniform prior \citep{steorts_bayesian_2016}. 
We find that our estimate is significantly larger than the unadjusted count 
of 563,626 reported by \citet{census_report_2010}. 
The difference may be explained by several factors. 
Firstly, our approach may capture individuals who are not represented 
in the Census, but who are represented in the Numident (assuming they 
have a Social Security number).
Indeed, the participation rate for the Census is known to be lower in 
Wyoming than for other states \citep{census_response}.
Secondly, there may be some double-counting for records that cannot 
be reliably linked---e.g.~due to missing or unreliable attribute values.
Thirdly, there may be minor differences in the Census data---e.g.\ 
whether blank forms are discarded or not.

To assess the reliability of ER, we report pairwise evaluation measures (precision, recall and F1-score) in Table~\ref{tbl:census-results}. 
These measures are computed using ground truth identifiers, 
which are available for a limited subset of the records.
To our knowledge, these are the first performance measures that have been 
published for ER of Census and administrative data at the state-level.
However, we note that the measures should be interpreted with caution, 
as the limited ground truth may not be representative of all records 
(hence the need for unsupervised methods).

\begin{table}
  \centering
  \caption{Results for ER of 2010 Census and Numident data in Wyoming. 
  Pairwise evaluation measures are computed using ground truth identifiers 
  available for a subset of the records.}
  \label{tbl:census-results}
  \footnotesize
  \begin{center}
  \begin{tabular}{*{3}{c} *{2}{c}}
    \toprule
    \multicolumn{3}{c}{Pairwise measures} & \multicolumn{2}{c}{Posterior population size} \\
    \cmidrule(lr){1-3} \cmidrule(lr){4-5}
    Precision & Recall & F1-score &    Mean & Std.~error \\
    \midrule
    0.97 &   0.84 &     0.90 & 616,000 & 5,000 \\
    \bottomrule
  \end{tabular}
  \end{center}
\end{table}

We believe that \dblink\ shows promise in producing enumerations at the 
state-level, while accounting for ER uncertainty.
Moving forward, it would be beneficial to study the accuracy and scalability 
of \dblink\ in other states, to further assess the reliability of our 
methodology for conducting linkage tasks within national statistical 
agencies. 
While it is beyond the scope of this paper, we are also interested in 
incorporating additional sources of administrative data, such as tax records, 
in future work.

\section{Concluding remarks}
\label{sec:conclusions}
In this paper we have proposed \dblink---a method for performing scalable 
ER with integrated blocking in a fully Bayesian framework. 
Our approach leverages an auxiliary variable representation, which 
partitions the latent entities and records into auxiliary blocks. 
Since the auxiliary blocks are not fixed, but inferred during inference, 
we are able to propagate uncertainty between the blocking and ER stages. 
This stands in contrast with the existing literature, where blocking 
and ER are performed in two separate stages without uncertainty 
propagation. 
In addition, we have shown that our approach does not compromise 
the correctness of the marginal posterior over the model parameters. 
In other words, approximate posterior samples produced by \dblink\ 
are independent of the blocking design in the asymptotic limit.

To further improve scalability, we discussed inference for 
\dblink\ in a distributed\slash parallel setting.
We proposed a blocking function based on $k$-d trees, which achieves 
good load balancing at the block level.  
We designed a distributed partially-collapsed Gibbs sampler, 
with superior mixing properties compared to a standard Gibbs sampler. 
We also presented fast algorithms for the Gibbs updates, which 
leverage indexing data structures and perturbation sampling. 
Our empirical evaluation on five data sets demonstrated efficiency 
gains for \dblink\ in excess of 300$\times$ when compared to existing 
methods.
We also demonstrated the potential of \dblink\ in a population 
enumeration case study, using data from the 2010 Decennial 
Census and Social Security Administration. 
The resulting enumeration was output with uncertainty, and achieved 
high precision and recall. 

An implementation of \dblink\ is provided as an open-source Apache Spark 
package. 
We also provide an interface for R users, for broad accessibility. 
Our software has been put in place within the United States Census 
Bureau for research purposes.

\bigskip
\begin{center}
{\large\bf SUPPLEMENTAL MATERIALS}
\end{center}

\begin{description}

\item[Appendices:] Includes proofs, further details about the 
experimental setup, and additional results. (PDF file)

\item[Code:] An implementation of \dblink\ in Apache Spark 
and a corresponding R interface. (Zip file)

\item[Data:] An archive containing data sets that we have permission to 
redistribute. (Zip file)

\end{description}

\bigskip
\begin{center}
{\large\bf ACKNOWLEDGEMENTS}
\end{center}
The authors would also like to thank the anonymous reviewers, Associate 
Editor and Editor for their valuable comments and helpful suggestions. 
N.~Marchant acknowledges the support of an Australian Government 
Research Training Program Scholarship and the AMSIIntern program
hosted by the Australian Bureau of Statistics. 
R.~C.~Steorts and A.~Kaplan acknowledge the support of NSF SES-1534412 and 
CAREER-1652431.
B.~Rubinstein acknowledges the support of Australian Research Council grant DP150103710.
N.~Marchant and B.~Rubinstein also acknowledge support of Australian Bureau 
of Statistics project ABS2018.363.

\bibliographystyle{jasa}
{
\footnotesize
\bibliography{dblink}

\begin{thebibliography}{7}
\newcommand{\enquote}[1]{``#1''}
\expandafter\ifx\csname natexlab\endcsname\relax\def\natexlab#1{#1}\fi
\expandafter\ifx\csname url\endcsname\relax
  \def\url#1{{\tt #1}}\fi
\expandafter\ifx\csname urlprefix\endcsname\relax\def\urlprefix{URL }\fi

\bibitem[{{Banca d'Italia}(n.d.)}]{bancaitalia_2010}
{Banca d'Italia}.
\newblock \enquote{{Bank of Italy -- Survey on Household Income and Wealth}.}
\newblock http://www.bancaditalia.it/pubblicazioni/indagine-famiglie/index.html
  (n.d.).
\newblock Accessed: 2018-03-09.

\bibitem[{Christen(2014)}]{christen_preparation_2014}
Christen, P.
\newblock \enquote{Preparation of a real temporal voter data set for record
  linkage and duplicate detection research.}
\newblock Technical report, Australian National University (2014).

\bibitem[{Manton(2010)}]{manton_nltcs_2010}
Manton, K.~G.
\newblock \enquote{{National} {Long-Term} {Care} {Survey}: 1982, 1984, 1989,
  1994, 1999 and 2004.} (2010).

\bibitem[{Sariyar and Borg(2010)}]{sariyar_recordlinkage_2010_app}
Sariyar, M. and Borg, A.
\newblock \enquote{The {RecordLinkage} {Package}: {Detecting} {Errors} in
  {Data}.}
\newblock {\em The R Journal\/}, 2(2):61--67 (2010).

\bibitem[{Steorts et~al.(2016)Steorts, Hall, and
  Fienberg}]{steorts_bayesian_2016_app}
Steorts, R.~C., Hall, R., and Fienberg, S.~E.
\newblock \enquote{A {Bayesian} {Approach} to {Graphical} {Record} {Linkage}
  and {Deduplication}.}
\newblock {\em Journal of the American Statistical Association\/},
  111(516):1660--1672 (2016).

\bibitem[{Vose(1991)}]{vose_linear_1991_app}
Vose, M.~D.
\newblock \enquote{A linear algorithm for generating random numbers with a
  given distribution.}
\newblock {\em IEEE Transactions on Software Engineering\/}, 17(9):972--975
  (1991).

\bibitem[{Zanella et~al.(2016)Zanella, Betancourt, Wallach, Miller, Zaidi, and
  Steorts}]{zanella_flexible_2016_app}
Zanella, G., Betancourt, B., Wallach, H., Miller, J., Zaidi, A., and Steorts,
  R.~C.
\newblock \enquote{Flexible {Models} for {Microclustering} with {Application}
  to {Entity} {Resolution}.}
\newblock In {\em Proceedings of the 30th {International} {Conference} on
  {Neural} {Information} {Processing} {Systems}\/}, {NIPS}'16, 1425--1433. NY,
  USA: Curran Associates Inc. (2016).

\end{thebibliography}


\begin{thebibliography}{66}
\newcommand{\enquote}[1]{``#1''}
\expandafter\ifx\csname natexlab\endcsname\relax\def\natexlab#1{#1}\fi
\expandafter\ifx\csname url\endcsname\relax
  \def\url#1{{\tt #1}}\fi
\expandafter\ifx\csname urlprefix\endcsname\relax\def\urlprefix{URL }\fi

\bibitem[{Ahn et~al.(2014)Ahn, Shahbaba, and Welling}]{ahn_distributed_2014}
Ahn, S., Shahbaba, B., and Welling, M.
\newblock \enquote{Distributed {Stochastic} {Gradient} {MCMC}.}
\newblock In {\em Proceedings of the 31st {International} {Conference} on
  {International} {Conference} on {Machine} {Learning} - {Volume} 32\/},
  {ICML}'14, II--1044--II--1052. Beijing, China: JMLR.org (2014).

\bibitem[{Bentley(1975)}]{bentley_multidimensional_1975}
Bentley, J.~L.
\newblock \enquote{Multidimensional {Binary} {Search} {Trees} {Used} for
  {Associative} {Searching}.}
\newblock {\em Commun. ACM\/}, 18(9):509--517 (1975).

\bibitem[{Bilenko and Mooney(2003)}]{bilenko_adaptive_2003}
Bilenko, M. and Mooney, R.~J.
\newblock \enquote{Adaptive {Duplicate} {Detection} {Using} {Learnable}
  {String} {Similarity} {Measures}.}
\newblock In {\em Proceedings of the {Ninth} {ACM} {SIGKDD} {International}
  {Conference} on {Knowledge} {Discovery} and {Data} {Mining}\/}, {KDD} '03,
  39--48. New York, NY, USA: ACM (2003).

\bibitem[{Chang and Fisher(2013)}]{chang_parallel_2013}
Chang, J. and Fisher, J.~W., III.
\newblock \enquote{Parallel {Sampling} of {DP} {Mixture} {Models} {Using}
  {Sub}-clusters {Splits}.}
\newblock In {\em Proceedings of the 26th {International} {Conference} on
  {Neural} {Information} {Processing} {Systems}\/}, volume~1 of {\em
  {NIPS}'13\/}, 620--628. NY, USA: Curran Associates Inc. (2013).

\bibitem[{Christen(2012{\natexlab{a}})}]{christen_survey_2012}
Christen, P.
\newblock \enquote{{A Survey of Indexing Techniques for Scalable Record Linkage
  and Deduplication}.}
\newblock {\em IEEE Transactions on Knowledge and Data Engineering\/},
  24(9):1537--1555 (2012{\natexlab{a}}).

\bibitem[{Christen(2012{\natexlab{b}})}]{christen_data_2012}
---.
\newblock {\em Data {Matching}: {Concepts} and {Techniques} for {Record}
  {Linkage}, {Entity} {Resolution}, and {Duplicate} {Detection}\/}.
\newblock Data-{Centric} {Systems} and {Applications}. Berlin Heidelberg:
  Springer-Verlag (2012{\natexlab{b}}).

\bibitem[{Copas and Hilton(1990)}]{copas_record_1990}
Copas, J.~B. and Hilton, F.~J.
\newblock \enquote{Record {Linkage}: {Statistical} {Models} for {Matching}
  {Computer} {Records}.}
\newblock {\em Journal of the Royal Statistical Society. Series A (Statistics
  in Society)\/}, 153(3):287--320 (1990).

\bibitem[{Dong and Srivastava(2015)}]{dong_big_2015}
Dong, X.~L. and Srivastava, D.
\newblock \enquote{Big {Data} {Integration}.}
\newblock {\em Synthesis Lectures on Data Management\/}, 7(1):1--198 (2015).

\bibitem[{Enamorado et~al.(2019)Enamorado, Fifield, and
  Imai}]{enamorado_using_2019}
Enamorado, T., Fifield, B., and Imai, K.
\newblock \enquote{Using a {Probabilistic} {Model} to {Assist} {Merging} of
  {Large-Scale} {Administrative} {Records}.}
\newblock {\em American Political Science Review\/}, 113(2):353–371 (2019).

\bibitem[{Fan et~al.(2009)Fan, Jia, Li, and Ma}]{fan_reasoning_2009}
Fan, W., Jia, X., Li, J., and Ma, S.
\newblock \enquote{Reasoning {About} {Record} {Matching} {Rules}.}
\newblock {\em Proc. VLDB Endow.\/}, 2(1):407--418 (2009).

\bibitem[{Fellegi and Sunter(1969)}]{fellegi_theory_1969}
Fellegi, I.~P. and Sunter, A.~B.
\newblock \enquote{A {Theory} for {Record} {Linkage}.}
\newblock {\em Journal of the American Statistical Association\/},
  64(328):1183--1210 (1969).

\bibitem[{Flegal et~al.(2017)Flegal, Hughes, Vats, and Dai}]{mcmcse}
Flegal, J.~M., Hughes, J., Vats, D., and Dai, N.
\newblock {\em mcmcse: Monte Carlo Standard Errors for MCMC\/}.
\newblock Riverside, CA, Denver, CO, Coventry, UK, and Minneapolis, MN (2017).
\newblock R package version 1.3-2.

\bibitem[{Fortini et~al.(2001)Fortini, Liseo, Nuccitelli, and
  Scanu}]{fortini_bayesian_2001}
Fortini, M., Liseo, B., Nuccitelli, A., and Scanu, M.
\newblock \enquote{{On Bayesian Record Linkage}.}
\newblock {\em Research in Official Statistics\/}, 4(1):185--198 (2001).

\bibitem[{Friedman et~al.(1977)Friedman, Bentley, and
  Finkel}]{friedman_algorithm_1977}
Friedman, J.~H., Bentley, J.~L., and Finkel, R.~A.
\newblock \enquote{An {Algorithm} for {Finding} {Best} {Matches} in
  {Logarithmic} {Expected} {Time}.}
\newblock {\em ACM Trans. Math. Softw.\/}, 3(3):209--226 (1977).

\bibitem[{Ge et~al.(2015)Ge, Chen, Wan, and Ghahramani}]{ge_distributed_2015}
Ge, H., Chen, Y., Wan, M., and Ghahramani, Z.
\newblock \enquote{Distributed {Inference} for {Dirichlet} {Process} {Mixture}
  {Models}.}
\newblock In Bach, F. and Blei, D. (eds.), {\em Proceedings of the 32nd
  {International} {Conference} on {Machine} {Learning}\/}, volume~37 of {\em
  Proceedings of {Machine} {Learning} {Research}\/}, 2276--2284. Lille, France:
  PMLR (2015).

\bibitem[{Getoor and Machanavajjhala(2012)}]{getoor_entity_2012}
Getoor, L. and Machanavajjhala, A.
\newblock \enquote{Entity {Resolution}: {Theory}, {Practice} \& {Open}
  {Challenges}.}
\newblock {\em Proc. VLDB Endow.\/}, 5(12):2018--2019 (2012).

\bibitem[{Gokhale et~al.(2014)Gokhale, Das, Doan, Naughton, Rampalli, Shavlik,
  and Zhu}]{gokhale_corleone:_2014}
Gokhale, C., Das, S., Doan, A., Naughton, J.~F., Rampalli, N., Shavlik, J., and
  Zhu, X.
\newblock \enquote{Corleone: {Hands}-off {Crowdsourcing} for {Entity}
  {Matching}.}
\newblock In {\em Proceedings of the 2014 {ACM} {SIGMOD} {International}
  {Conference} on {Management} of {Data}\/}, {SIGMOD} '14, 601--612. New York,
  NY, USA: ACM (2014).

\bibitem[{Gutman et~al.(2013)Gutman, Afendulis, and
  Zaslavsky}]{gutman_bayesian_2013}
Gutman, R., Afendulis, C.~C., and Zaslavsky, A.~M.
\newblock \enquote{A {Bayesian} {Procedure} for {File} {Linking} to {Analyze}
  {End}-of-{Life} {Medical} {Costs}.}
\newblock {\em Journal of the American Statistical Association\/},
  108(501):34--47 (2013).

\bibitem[{Herzog et~al.(2007)Herzog, Scheuren, and Winkler}]{Herzog_2007}
Herzog, T.~N., Scheuren, F.~J., and Winkler, W.~E.
\newblock {\em {Data Quality and Record Linkage Techniques}\/}.
\newblock New York: Springer-Verlag (2007).

\bibitem[{Jain and Neal(2004)}]{jain_split-merge_2004}
Jain, S. and Neal, R.~M.
\newblock \enquote{A {Split}-{Merge} {Markov} chain {Monte} {Carlo} {Procedure}
  for the {Dirichlet} {Process} {Mixture} {Model}.}
\newblock {\em Journal of Computational and Graphical Statistics\/},
  13(1):158--182 (2004).

\bibitem[{Jasra et~al.(2005)Jasra, Holmes, and Stephens}]{jasra_markov_2005}
Jasra, A., Holmes, C.~C., and Stephens, D.~A.
\newblock \enquote{Markov {Chain} {Monte} {Carlo} {Methods} and the {Label}
  {Switching} {Problem} in {Bayesian} {Mixture} {Modeling}.}
\newblock {\em Statistical Science\/}, 20(1):50--67 (2005).

\bibitem[{Lahiri and Larsen(2005)}]{lahiri_2005}
Lahiri, P. and Larsen, M.~D.
\newblock \enquote{Regression Analysis With Linked Data.}
\newblock {\em Journal of the American Statistical Association\/},
  100(469):222--230 (2005).

\bibitem[{Larsen(2005)}]{larsen_advances_2005}
Larsen, M.~D.
\newblock \enquote{{Advances in Record Linkage Theory: Hierarchical Bayesian
  Record Linkage Theory}.}
\newblock In {\em Proceedings of the Survey Research Methods Section\/},
  3277--3284. American Statistical Association (2005).

\bibitem[{Larsen(2012)}]{larsen_experiment_2012}
---.
\newblock \enquote{An experiment with hierarchical {Bayesian} record linkage.}
  (2012).
\newblock {arXiv:1212.5203}.

\bibitem[{Lesot et~al.(2008)Lesot, Rifqi, and Benhadda}]{lesot_similarity_2008}
Lesot, M.-J., Rifqi, M., and Benhadda, H.
\newblock \enquote{Similarity measures for binary and numerical data: a
  survey.}
\newblock {\em International Journal of Knowledge Engineering and Soft Data
  Paradigms\/}, 1(1):63--84 (2008).

\bibitem[{Little and Rubin(2002)}]{little_statistical_2002}
Little, R. J.~A. and Rubin, D.~B.
\newblock {\em Statistical {Analysis} with {Missing} {Data}\/}.
\newblock Wiley (2002).

\bibitem[{Liu(2004)}]{liu_monte_2004}
Liu, J.~S.
\newblock {\em Monte {Carlo} {Strategies} in {Scientific} {Computing}\/}.
\newblock Springer {Series} in {Statistics}. New York: Springer-Verlag (2004).

\bibitem[{Lovell et~al.(2013)Lovell, Malmaud, Adams, and
  Mansinghka}]{lovell_clustercluster:_2013}
Lovell, D., Malmaud, J., Adams, R.~P., and Mansinghka, V.~K.
\newblock \enquote{ClusterCluster: Parallel Markov Chain Monte Carlo for
  Dirichlet Process Mixtures.} (2013).
\newblock {arXiv:1304.2302}.

\bibitem[{McVeigh and Murray(2017)}]{mcveigh_practical_2017}
McVeigh, B.~S. and Murray, J.~S.
\newblock \enquote{Practical {Bayesian} {Inference} for {Record} {Linkage}.}
  (2017).
\newblock {arXiv:1710.10558}.

\bibitem[{McVeigh et~al.(2019)McVeigh, Spahn, and
  Murray}]{mcveigh_scaling_2019}
McVeigh, B.~S., Spahn, B.~T., and Murray, J.~S.
\newblock \enquote{{Scaling Bayesian Probabilistic Record Linkage with Post-Hoc
  Blocking: An Application to the California Great Registers}.} (2019).
\newblock {arXiv:1905.05337}.

\bibitem[{Mudgal et~al.(2018)Mudgal, Li, Rekatsinas, Doan, Park, Krishnan,
  Deep, Arcaute, and Raghavendra}]{mudgal_deep_2018}
Mudgal, S., Li, H., Rekatsinas, T., Doan, A., Park, Y., Krishnan, G., Deep, R.,
  Arcaute, E., and Raghavendra, V.
\newblock \enquote{Deep {Learning} for {Entity} {Matching}: {A} {Design}
  {Space} {Exploration}.}
\newblock In {\em Proceedings of the 2018 {International} {Conference} on
  {Management} of {Data}\/}, {SIGMOD} '18, 19--34. New York, NY, USA: ACM
  (2018).

\bibitem[{Newcombe et~al.(1959)Newcombe, Kennedy, Axford, and
  James}]{newcombe_automatic_1959}
Newcombe, H.~B., Kennedy, J.~M., Axford, S.~J., and James, A.~P.
\newblock \enquote{Automatic {Linkage} of {Vital} {Records}: {Computers} can be
  used to extract "follow-up" statistics of families from files of routine
  records.}
\newblock {\em Science\/}, 130(3381):954--959 (1959).

\bibitem[{Newman et~al.(2009)Newman, Asuncion, Smyth, and
  Welling}]{newman_distributed_2009}
Newman, D., Asuncion, A., Smyth, P., and Welling, M.
\newblock \enquote{Distributed algorithms for topic models.}
\newblock {\em Journal of Machine Learning Research\/}, 10(Aug):1801--1828
  (2009).

\bibitem[{Papadakis et~al.(2016)Papadakis, Svirsky, Gal, and
  Palpanas}]{papadakis_comparative_2016}
Papadakis, G., Svirsky, J., Gal, A., and Palpanas, T.
\newblock \enquote{Comparative {Analysis} of {Approximate} {Blocking}
  {Techniques} for {Entity} {Resolution}.}
\newblock {\em Proc. VLDB Endow.\/}, 9(9):684--695 (2016).

\bibitem[{Price et~al.(2013)Price, Klinger, Qtiesh, and Ball}]{price_2013}
Price, M., Klinger, J., Qtiesh, A., and Ball, P.
\newblock \enquote{{Updated Statistical Analysis of Documentation of Killings
  in the Syrian Arab Repulic}.} (2013).

\bibitem[{Rastogi et~al.(2012)Rastogi, O'Hara, Noon, Zapata, Espinoza,
  Marshall, Schellhamer, and Brown}]{census_report_2010}
Rastogi, S., O'Hara, A., Noon, J., Zapata, E.~A., Espinoza, C., Marshall,
  L.~B., Schellhamer, T.~A., and Brown, J.~D.
\newblock \enquote{{2010 Census Match Study}.}
\newblock Technical report, Center for Administrative Records Research and
  Applications, United States Census Bureau (2012).

\bibitem[{Sadinle(2014)}]{sadinle_detecting_2014}
Sadinle, M.
\newblock \enquote{Detecting duplicates in a homicide registry using a Bayesian
  partitioning approach.}
\newblock {\em Ann. Appl. Stat.\/}, 8(4):2404--2434 (2014).

\bibitem[{Sadinle(2017)}]{sadinle_bayesian_2017}
---.
\newblock \enquote{Bayesian {Estimation} of {Bipartite} {Matchings} for
  {Record} {Linkage}.}
\newblock {\em Journal of the American Statistical Association\/},
  112(518):600--612 (2017).

\bibitem[{Sadinle and Fienberg(2013)}]{sadinle_generalized_2013}
Sadinle, M. and Fienberg, S.~E.
\newblock \enquote{A {Generalized} {Fellegi}-{Sunter} {Framework} for
  {Multiple} {Record} {Linkage} {With} {Application} to {Homicide} {Record}
  {Systems}.}
\newblock {\em Journal of the American Statistical Association\/},
  108(502):385--397 (2013).

\bibitem[{Saria(2014)}]{saria_trillion_2014}
Saria, S.
\newblock \enquote{A \$3 trillion challenge to computational scientists:
  {Transforming} healthcare delivery.}
\newblock {\em IEEE Intelligent Systems\/}, 29(4):82--87 (2014).

\bibitem[{Sariyar and Borg(2010)}]{sariyar_recordlinkage_2010}
Sariyar, M. and Borg, A.
\newblock \enquote{The {RecordLinkage} {Package}: {Detecting} {Errors} in
  {Data}.}
\newblock {\em The R Journal\/}, 2(2):61--67 (2010).

\bibitem[{Singh et~al.(2017)Singh, Meduri, Elmagarmid, Madden, Papotti,
  Quian{\'e}-Ruiz, Solar-Lezama, and Tang}]{singh_generating_2017}
Singh, R., Meduri, V., Elmagarmid, A., Madden, S., Papotti, P.,
  Quian{\'e}-Ruiz, J.-A., Solar-Lezama, A., and Tang, N.
\newblock \enquote{Generating {Concise} {Entity} {Matching} {Rules}.}
\newblock In {\em Proceedings of the 2017 {ACM} {International} {Conference} on
  {Management} of {Data}\/}, {SIGMOD} '17, 1635--1638. New York, NY, USA: ACM
  (2017).

\bibitem[{Smola and Narayanamurthy(2010)}]{smola_architecture_2010}
Smola, A. and Narayanamurthy, S.
\newblock \enquote{An {Architecture} for {Parallel} {Topic} {Models}.}
\newblock {\em Proc. VLDB Endow.\/}, 3(1-2):703--710 (2010).

\bibitem[{Soon et~al.(2001)Soon, Ng, and Lim}]{soon_machine_2001}
Soon, W.~M., Ng, H.~T., and Lim, D. C.~Y.
\newblock \enquote{{A Machine Learning Approach to Coreference Resolution of
  Noun Phrases}.}
\newblock {\em Computational linguistics\/}, 27(4):521--544 (2001).

\bibitem[{Steorts(2015)}]{steorts_entity_2015}
Steorts, R.~C.
\newblock \enquote{Entity {Resolution} with {Empirically} {Motivated}
  {Priors}.}
\newblock {\em Bayesian Analysis\/}, 10(4):849--875 (2015).

\bibitem[{Steorts et~al.(2017)Steorts, Barnes, and
  Neiswanger}]{steorts_performance_2017}
Steorts, R.~C., Barnes, M., and Neiswanger, W.
\newblock \enquote{Performance {Bounds} for {Graphical} {Record} {Linkage}.}
\newblock In Singh, A. and Zhu, J. (eds.), {\em Proceedings of the 20th
  {International} {Conference} on {Artificial} {Intelligence} and
  {Statistics}\/}, volume~54 of {\em Proceedings of {Machine} {Learning}
  {Research}\/}, 298--306. Fort Lauderdale, FL, USA: PMLR (2017).

\bibitem[{Steorts et~al.(2016)Steorts, Hall, and
  Fienberg}]{steorts_bayesian_2016}
Steorts, R.~C., Hall, R., and Fienberg, S.~E.
\newblock \enquote{A {Bayesian} {Approach} to {Graphical} {Record} {Linkage}
  and {Deduplication}.}
\newblock {\em Journal of the American Statistical Association\/},
  111(516):1660--1672 (2016).

\bibitem[{Steorts et~al.(2014)Steorts, Ventura, Sadinle, and
  Fienberg}]{steorts_comparison_2014}
Steorts, R.~C., Ventura, S.~L., Sadinle, M., and Fienberg, S.~E.
\newblock \enquote{A {Comparison} of {Blocking} {Methods} for {Record}
  {Linkage}.}
\newblock In Domingo-Ferrer, J. (ed.), {\em Privacy in {Statistical}
  {Databases}\/}, Lecture {Notes} in {Computer} {Science}, 253--268. Cham:
  Springer International Publishing (2014).

\bibitem[{Tancredi and Liseo(2011)}]{tancredi_hierarchical_2011}
Tancredi, A. and Liseo, B.
\newblock \enquote{A hierarchical {Bayesian} approach to record linkage and
  population size problems.}
\newblock {\em The Annals of Applied Statistics\/}, 5(2B):1553--1585 (2011).

\bibitem[{Tancredi et~al.(2020)Tancredi, Steorts, and
  Liseo}]{tancredi_unified_2020}
Tancredi, A., Steorts, R., and Liseo, B.
\newblock \enquote{{A Unified Framework for De-Duplication and Population Size
  Estimation}.}
\newblock {\em Bayesian Analysis\/} (2020).

\bibitem[{Turek et~al.(2016)Turek, Valpine, and
  Paciorek}]{turek_efficient_2016}
Turek, D., Valpine, P.~d., and Paciorek, C.~J.
\newblock \enquote{Efficient {Markov} chain {Monte} {Carlo} sampling for
  hierarchical hidden {Markov} models.}
\newblock {\em Environmental and Ecological Statistics\/}, 23(4):549--564
  (2016).

\bibitem[{{United States Census Bureau}(n.d.)}]{census_response}
{United States Census Bureau}.
\newblock \enquote{{2010 Census Participation Rates}.}
\newblock
  https://www.census.gov/data/datasets/2010/dec/2010-participation-rates.html
  (n.d.).
\newblock Accessed: 2020-05-25.

\bibitem[{van Dyk and Park(2008)}]{dyk_partially_2008}
van Dyk, D.~A. and Park, T.
\newblock \enquote{Partially {Collapsed} {Gibbs} {Samplers}.}
\newblock {\em Journal of the American Statistical Association\/},
  103(482):790--796 (2008).

\bibitem[{Vats et~al.(2019)Vats, Flegal, and Jones}]{vats_multivariate_2019}
Vats, D., Flegal, J.~M., and Jones, G.~L.
\newblock \enquote{Multivariate output analysis for Markov chain Monte Carlo.}
\newblock {\em Biometrika\/}, 106(2):321--337 (2019).

\bibitem[{Vinh et~al.(2010)Vinh, Epps, and Bailey}]{vinh_information_2010}
Vinh, N.~X., Epps, J., and Bailey, J.
\newblock \enquote{Information theoretic measures for clusterings comparison:
  {Variants}, properties, normalization and correction for chance.}
\newblock {\em Journal of Machine Learning Research\/}, 11(Oct):2837--2854
  (2010).

\bibitem[{Vose(1991)}]{vose_linear_1991}
Vose, M.~D.
\newblock \enquote{A linear algorithm for generating random numbers with a
  given distribution.}
\newblock {\em IEEE Transactions on Software Engineering\/}, 17(9):972--975
  (1991).

\bibitem[{Wang et~al.(2012)Wang, Kraska, Franklin, and
  Feng}]{wang_crowder:_2012}
Wang, J., Kraska, T., Franklin, M.~J., and Feng, J.
\newblock \enquote{{CrowdER}: {Crowdsourcing} {Entity} {Resolution}.}
\newblock {\em Proc. VLDB Endow.\/}, 5(11):1483--1494 (2012).

\bibitem[{Williamson et~al.(2013)Williamson, Dubey, and
  Xing}]{williamson_parallel_2013}
Williamson, S., Dubey, A., and Xing, E.
\newblock \enquote{Parallel {Markov} {Chain} {Monte} {Carlo} for
  {Nonparametric} {Mixture} {Models}.}
\newblock In Dasgupta, S. and McAllester, D. (eds.), {\em Proceedings of the
  30th {International} {Conference} on {Machine} {Learning}\/}, volume~28 of
  {\em Proceedings of {Machine} {Learning} {Research}\/}, 98--106. Atlanta,
  Georgia, USA: PMLR (2013).

\bibitem[{Winkler(1999)}]{winkler_state_1999}
Winkler, W.~E.
\newblock \enquote{The {State} of {Record} {Linkage} and {Current} {Research}
  {Problems}.}
\newblock Technical report, Statistical Research Division, U.S. Bureau of the
  Census (1999).

\bibitem[{Winkler(2000)}]{winkler_2000}
---.
\newblock \enquote{Machine Learning, Information Retrieval, and Record
  Linkage.}
\newblock In {\em Proceedings of the Section on Survey Research Methods,
  20--29\/}. American Statistical Association (2000).

\bibitem[{Winkler(2002)}]{winkler_methods_2002}
---.
\newblock \enquote{Methods for {Record Linkage} and {Bayesian Networks}.}
\newblock Technical Report Statistics \#2002-05, U.S. Bureau of the Census
  (2002).

\bibitem[{Winkler(2006)}]{winkler_overview_2006}
---.
\newblock \enquote{{Overview of Record Linkage and Current Research
  Directions}.}
\newblock Technical Report Statistics \#2006-2, Statistical Research Division,
  U.S. Census Bureau (2006).

\bibitem[{Winkler(2014)}]{winkler2014matching}
---.
\newblock \enquote{Matching and record linkage.}
\newblock {\em Wiley Interdisciplinary Reviews: Computational Statistics\/},
  6(5):313--325 (2014).

\bibitem[{Yujian and Bo(2007)}]{yujian_normalized_2007}
Yujian, L. and Bo, L.
\newblock \enquote{A {Normalized} {Levenshtein} {Distance} {Metric}.}
\newblock {\em IEEE Transactions on Pattern Analysis and Machine
  Intelligence\/}, 29(6):1091--1095 (2007).

\bibitem[{Zanella(2020)}]{zanella_informed_2020}
Zanella, G.
\newblock \enquote{{Informed Proposals for Local MCMC in Discrete Spaces}.}
\newblock {\em Journal of the American Statistical Association\/},
  115(530):852--865 (2020).

\bibitem[{Zanella et~al.(2016)Zanella, Betancourt, Wallach, Miller, Zaidi, and
  Steorts}]{zanella_flexible_2016}
Zanella, G., Betancourt, B., Wallach, H., Miller, J., Zaidi, A., and Steorts,
  R.~C.
\newblock \enquote{Flexible {Models} for {Microclustering} with {Application}
  to {Entity} {Resolution}.}
\newblock In {\em Proceedings of the 30th {International} {Conference} on
  {Neural} {Information} {Processing} {Systems}\/}, {NIPS}'16, 1425--1433. NY,
  USA: Curran Associates Inc. (2016).

\end{thebibliography}
}

\if1\arxiv
\newpage
\setcounter{figure}{0}
\setcounter{equation}{0}
\setcounter{table}{0}
\setcounter{algorithm}{0}
\setcounter{page}{1}
\makeatletter

\sectionfont{\Large\bfseries}
\renewcommand{\theequation}{S\arabic{equation}}
\renewcommand{\thefigure}{S\arabic{figure}}
\renewcommand{\thetable}{S\arabic{table}}
\renewcommand{\thealgorithm}{S\arabic{algorithm}}
\renewcommand{\bibnumfmt}[1]{[S#1]}
\renewcommand{\citenumfont}[1]{S#1}

\title{\bf Appendices for ``d-blink: Distributed End-to-End Bayesian Entity Resolution''}
\author{Neil G.~Marchant\textsuperscript{a} \and
	Andee Kaplan\textsuperscript{b} \and 
	Daniel N.~Elazar\textsuperscript{c} \and
	Benjamin I.~P.~Rubinstein\textsuperscript{a} \and 
	Rebecca C.~Steorts\textsuperscript{d}}
\date{
	\textsuperscript{a}School of Computing and Information Systems, University 
	of Melbourne\\
	\textsuperscript{b}Department of Statistics, Colorado State University\\
	\textsuperscript{c}Methodology Division, Australian Bureau of Statistics\\
	\textsuperscript{d}Department of Statistical Science and Computer Science, Duke University\\Principal Mathematical Statistician, United States Census Bureau\\[2ex]
	\today}
\maketitle

\newpage

\appendix
\section{Derivation of the posterior distribution}
\label{app-sec:full-posterior}
Here we sketch the derivation of the joint posterior distribution over the 
unobserved variables conditioned on the observed record attributes 
$\vec{X}^{(o)}$, which is given in 
Equation~\ref{eqn:partitioned-posterior} of the paper.
First we read the factorization off the plate diagram in 
Figure~\ref{fig:plate-diagram}, together with the conditional dependence 
assumptions detailed in Section~\ref{sec:model-specification} of the paper.
We obtain the following expression, up to a normalisation constant:
\begin{equation*}
\begin{split}
& p(\vec{\Gamma}, \vec{\Lambda}, \vec{Y}, \vec{Z}, \vec{\Theta}, \vec{X}^{(m)}|\vec{X}^{(o)}, \vec{O}) \propto 
  \prod_{e,a} p(y_{ea}|\phi_{a}) 
  \times \prod_{t,a} p(\theta_{ta}|\alpha_{a}, \beta_{a}) \\
& \quad {} \times 
  \prod_{t,r} \Big\{ p(\gamma_{tr}|\vec{Y})p(\lambda_{tr}|\gamma_{tr}, \vec{Y}) 
  \prod_{a} p(z_{tra}|\theta_{ta}) \Big\} 
  \times \prod_{\substack{t,r,a\\o_{tra}=1}} p(x_{tra}|z_{tra}, \lambda_{tr}, y_{\lambda_{tr}a}) \\
& \qquad {} \times 
  \prod_{\substack{t,r,a\\o_{tra}=0}} p(x_{tra}|z_{tra}, \lambda_{tr}, y_{\lambda_{tr}a}).
\end{split}
\end{equation*}
Ideally, we'd like to marginalize out all variables except $\vec{\Lambda}$ and 
$\vec{Y}$ (the variables of interest), however this is not tractable 
analytically.
Fortunately, we can marginalize out the missing record attributes 
$\vec{X}^{(m)}$ which yields Equation~\ref{eqn:partitioned-posterior} 
from the paper:
\begin{equation*}
\begin{split}
& p(\vec{\Gamma}, \vec{\Lambda}, \vec{Y}, \vec{Z}, \vec{\Theta}|\vec{X}^{(o)}, \vec{O}) \propto \prod_{e,a} p(y_{ea}|\phi_{a}) 
  \times \prod_{t,a} p(\theta_{ta}|\alpha_{a}, \beta_{a}) \\
& \quad {} \times 
  \prod_{t,r} \Big\{ p(\gamma_{tr}|\vec{Y})p(\lambda_{tr}|\gamma_{tr}, \vec{Y}) 
  \prod_{a} p(z_{tra}|\theta_{ta}) \Big\} 
  \times \prod_{\substack{t,r,a\\o_{tra}=1}} p(x_{tra}|z_{tra}, \lambda_{tr}, y_{\lambda_{tr}a}).
\end{split}
\end{equation*}

We can expand this further by substituting the conditional distributions 
given in Section~\ref{sec:model-specification} of the paper.
This yields:
\begin{equation}
\begin{split}
& p(\vec{\Gamma}, \vec{\Lambda}, \vec{Y}, \vec{Z}, \vec{\Theta}|\vec{X}^{(o)}, \vec{O}) \propto 
  \prod_{e,a} \phi_{a}(y_{ea}) \times \prod_{t,a} \theta_{ta}^{\alpha_{a} - 1} (1 - \theta_{ta})^{\beta_{a} - 1} \\
& \quad {} \times \prod_{t,r} \Big\{ \1{\lambda_{tr} \in 
  \partset_{\gamma_{tr}}(\vec{Y})} \prod_{a} \theta_{ta}^{z_{tra}} (1 - \theta_{ta})^{1-z_{tra}} \Big\} \\
& \mspace{40mu} {} \times  \prod_{\substack{t,r,a\\o_{tra}=1}} 
  \Big\{ (1 - z_{tra}) \1{x_{tra} = y_{\lambda_{tr}a}} 
  + z_{tra} \, \psi_{a}(x_{tra}|y_{\lambda_{tr}a}) \Big\}.
\end{split}
\label{app-eqn:expanded-posterior}
\end{equation}

\section{Equivalence of \texorpdfstring{\dblink}{d-blink} and \texorpdfstring{\blink}{blink}}
In this section, we present proofs of 
Propositions~\ref{thm:sim-dist-equiv} and~\ref{thm:posterior-equiv}, 
which show that the inferences we obtain from \dblink\ are equivalent to those 
we would obtain from \blink\ under certain conditions.

\subsection{Proof of Proposition~\ref{thm:sim-dist-equiv}: equivalence of 
distance\slash similarity representations}
\label{app-sec:proof-sim-dist-equiv}
It is straightforward to show that $\simfn$ as defined in 
Equation~\ref{eqn:simfn-distfn-correspondence} of the paper satisfies 
the requirements of Definition~\ref{def:attribute-sim-measure}.
All that remains is to show that the two parameterizations of the distortion 
distribution $\psi_{a}$ are equivalent.
Beginning with $\psi_{a}$ as parameterized in \blink, we substitute 
Equation~\ref{eqn:simfn-distfn-correspondence} and observe that
\begin{equation*}
\psi_{a}(v|w) \propto \phi_{a}(v) \euler^{-\distfn_{a}(v, w)} 
   = \phi_{a}(v) \euler^{d_\mathrm{max;a} + \simfn_{a}(v, w)} 
   \propto \phi_{a}(v) \euler^{\simfn_{a}(v, w)}.
\end{equation*}
This is identical to our parameterization in 
Equation~\ref{eqn:distortion-dist}. \qed

\subsection{Proof of Proposition~\ref{thm:posterior-equiv}: equivalence of 
\texorpdfstring{\dblink}{d-blink} and \texorpdfstring{\blink}{blink}}
\label{app-sec:proof-posterior-equiv}
Given that 
\begin{itemize}
  \item Proposition~\ref{thm:sim-dist-equiv} holds, 
  \item the distortion hyperparameters are the same for all attributes, and 
  \item all record attributes are observed,
\end{itemize}
the only factor in the posterior that differs from \blink\ is:
\begin{equation}
\prod_{t,r} p(\lambda_{tr}| \gamma_{tr}, \vec{Y}) 
  p(\gamma_{tr}|\vec{Y}).
\label{app-eqn:posterior-part-factor}
\end{equation}
Substituting the density for the conditional distributions 
for a single $t,r$ factor yields:
\begin{equation*}
p(\lambda_{tr}|\gamma_{tr},\vec{Y}) p(\gamma_{tr}|\vec{Y})
= \frac{\1{\lambda_{tr} \in \partset_{\gamma_{tr}}(\vec{Y})}}
{|\partset_{\gamma_{tr}}(\vec{Y})|} 
\times \frac{|\partset_{\gamma_{tr}}(\vec{Y})|}{E} = \frac{1}{E} \1{\lambda_{tr} \in \partset_{\gamma_{tr}}(\vec{Y})}.
\end{equation*}
Putting this in Equation~\ref{app-eqn:posterior-part-factor} and 
marginalizing over $\vec{\Gamma}$ we obtain:
\begin{equation*}
\prod_{t,r} \sum_{\gamma_{tr} = 1}^{B} p(\lambda_{tr}|\gamma_{tr},\vec{Y}) p(\gamma_{tr}|\vec{Y}) 
= \prod_{t,r} \frac{1}{E} \sum_{\gamma_{tr} = 1}^{B} \1{\lambda_{tr} \in \partset_{\gamma_{tr}}(\vec{Y})} 
= \prod_{t,r} \frac{1}{E} \1{\lambda_{tr} \in \{1,\ldots,E\}},
\end{equation*}
which is the factor that appears in the posterior for \blink. \qed

\section{Splitting rules for the \textit{k}-d tree blocking function}
\label{app-sec:splitting-rules}
In Section~\ref{sec:kd} of the paper we outline a blocking function inspired 
by $k$-d trees.
When inserting a node in the tree, we require a splitting rule that partitions 
the input set of values.
In ordinary $k$-d trees, the median is often used for this purpose, 
however it is not appropriate for the discrete input sets that we 
encounter.
As a result, we propose the following alternative splitting rules:
\begin{enumerate}
  \item \emph{Ordered median.}
  This rule is appropriate if the set of input attribute 
  values is large and\slash or has a natural ordering.
  If there is no natural ordering, an artificial ordering 
  must be applied (e.g. lexicographic ordering).
  The splitting rule is determined by sorting the input 
  values and finding the median, accounting for the 
  frequency of each value.
  Attribute values ordered before (after) the median are 
  passed to the left (right) child node.
  \item \emph{Reference set.}
  This rule is appropriate if the set of input attribute 
  values is small with no natural ordering.
  The splitting rule is determined by using a first-fit 
  bin-packing algorithm to split the values into two roughly 
  equal-sized bins, accounting for the frequency of each value.
  One of these bins is then labeled the ``reference set''.
  Attribute values (not) in the reference set are passed to 
  the left (right) child node.
\end{enumerate}

\section{Gibbs update distributions}
\label{app-sec:gibbs}
Here we list the conditional distributions for the Gibbs updates.
These are derived by referring to the posterior distribution in 
Equation~\ref{app-eqn:expanded-posterior}.

\subsection{Update for \texorpdfstring{$\theta_{ta}$}{distortion probabilities}}
\begin{equation}
\theta_{ta}|\vec{Z}, \vec{\Lambda}, \vec{\Gamma}, \vec{Y}, \vec{X}^{(o)}, \vec{O} \sim \operatorname{Beta}[z_{t \cdot a} + \alpha_{a}, R_t - z_{t \cdot a} + \beta_{a}]
\label{app-eqn:theta-update}
\end{equation}
where $z_{t \cdot a} := \sum_{r = 1}^{R_t} z_{tra}$.

\subsection{Update for \texorpdfstring{$z_{tra}$}{distortion indicators}}
\begin{equation}
  \begin{split}
  & z_{tra}|\vec{\Lambda}, \vec{\Gamma}, \vec{Y}, \vec{\Theta}, \vec{X}^{(o)}, \vec{O} \sim 
      (1-o_{tra}) \operatorname{Bernoulli}[\theta_{ta}] +
      o_{tra} \operatorname{Bernoulli}[\zeta_a(\theta_{ta}, x_{tra}, y_{\lambda_{tr}a})]
  \end{split}
  \label{app-eqn:z-update}
\end{equation}
where
$\zeta_a(\theta, x, y) = \begin{cases}
    1, & \text{if } x \neq y, \\
    \frac{\theta \psi_{a}(x|y)}{\theta \psi_{a}(x|y) - \theta + 1}, &\text{otherwise.}
  \end{cases}$

\subsection{Update for \texorpdfstring{$\lambda_{tr}$}{entity assignments}}
\begin{equation}
\begin{split}
& p(\lambda_{tr} | \vec{\Gamma}, \vec{Y}, \vec{\Theta}, \vec{Z}, \vec{X}^{(o)}, \vec{O}) \propto \\
& \qquad \1{\lambda_{tr} \in \partset_{\gamma_{tr}}(\vec{Y})} 
\prod_{\substack{a\\o_{tra}=1}} \Big\{ (1 - z_{tra}) \1{x_{tra} = y_{\lambda_{tr}a}} + z_{tra} \psi_{a}(x_{tra}|y_{\lambda_{tr}a})\Big\}.
\end{split}
\label{app-eqn:lambda-update}
\end{equation}

\section{Perturbation sampling algorithm}
\label{app-sec:pert-sampling}
In Proposition~\ref{thm:pert-sampling} of the paper, we show how to express a 
target pmf $p$ (from which we'd like to draw random variates) as a mixture over 
a base pmf $q$ and a perturbation pmf $v$.
Algorithm~\ref{app-alg:pert-sampling} demonstrates how to efficiently 
draw random variates from the target pmf using this mixture representation.

\begin{algorithm}[H]
  \caption{Perturbation sampling for $p(x|\omega)$}
  \label{app-alg:pert-sampling}
  \begin{algorithmic}[1]
    \Statex \textbf{Input:} 
    map from $x, \omega \in \mathcal{X}^\star \times \Omega \to 
    \epsilon(x|\omega)$; 
    map from $x \in \mathcal{X} \to q(x)$; 
    pre-initialized Alias sampler for $q$.
    \State $v \gets \emptyset$ 
    \Comment{empty map}
    \For {$x \in \mathcal{X}^\star$}
    \State $v(x) \gets q(x) \epsilon(x|\omega)$
    \EndFor
    \State $c \gets 1/\sum_{x \in \mathcal{X}^\star} v(x)$ 
    \Comment{normalization}
    \State $s \sim \mathrm{Bernoulli}\!\left[\frac{c}{1+c}\right]$
    \If {$s = 1$}
    \State \textbf{Return:} $x \sim q(\cdot)$ 
    \Comment{using input Alias sampler}
    \Else 
    \State $v \gets c \cdot v$
    \State \textbf{Return:} $x \sim v(\cdot)$
    \Comment{using new Alias sampler}
    \EndIf
  \end{algorithmic}
\end{algorithm}

\subsection{Proof of Proposition~\ref{thm:complexity-pert-sampling}: 
complexity of perturbation sampling}
Let us analyze the time complexity of Algorithm~\ref{app-alg:pert-sampling}.
Lines 2--6 are $O(|\mathcal{X}^\star|)$.
By properties of the Alias sampler~\citepApp{vose_linear_1991_app},
line 8 is $O(1)$ and line 11 is $O(|\mathcal{X}^\star|)$.
Thus the overall complexity is $O(|\mathcal{X}^\star|)$.

\section{Further details on the experimental set-up}
\label{app-sec:experiments-setup}
\subsection{Data sets}
We provide a brief description of each data set below. 
All data sets come with some form of ``ground truth'', which we 
use for evaluation purposes. 
However, the ground truth for \texttt{NCVR} and \texttt{SHIW0810} (two 
of the real data sets) may not be error-free as indicated below.

\begin{itemize}
  \item \texttt{ABSEmployee}. A synthetic data set used 
  internally for linkage experiments at the Australian Bureau of Statistics.
  It simulates an employment census and two supplementary 
  surveys (it is not derived from any real data sources).
  We used four categorical attributes: \texttt{MB}, \texttt{BDAY}, 
  \texttt{BYEAR} and \texttt{SEX}.
  \item \texttt{NCVR}. Two snapshots from the North Carolina 
  Voter Registration database taken two months 
  apart~\citepApp{christen_preparation_2014}.
  The snapshots are filtered to include only those voters 
  whose details changed over the two-month period.
  We used \texttt{first\_name}, \texttt{middle\_name} and 
  \texttt{last\_name} as string-type attributes; and 
  \texttt{age}, \texttt{gender} and \texttt{zip\_code} as 
  categorical attributes. 
  Unique voter identifiers are provided, however they are known to contain 
  some errors~\citepApp{christen_preparation_2014}.
  \item \texttt{NLTCS}. A subset of the U.S.\ National Long-Term 
  Care Survey~\citepApp{manton_nltcs_2010} comprising the 
  1982, 1989 and 1994 waves.
  It was necessary to use a subset, as race was subsampled in the other three 
  years, making it unsuitable for ER.
  We used four categorical attributes: \texttt{SEX}, \texttt{DOB}, 
  \texttt{STATE} and \texttt{REGOFF}.
  Unique identifiers are available which are known to be of high quality.
  \item \texttt{SHIW0810}. A subset from the Bank of Italy's 
  Survey on Household Income and Wealth~\citepApp{bancaitalia_2010} 
  comprising the 2008 and 2010 waves.
  We used eight categorical attributes: \texttt{IREG}, \texttt{SESSO}, 
  \texttt{ANASC}, \texttt{STUDIO}, \texttt{PAR}, \texttt{STACIV}, 
  \texttt{PERC} and \texttt{CFDIC}. 
  Unique identifiers were inferred using a deterministic algorithm, 
  which may not be error-free. 
  Further information and open-source code is provided at
  \url{http://github.com/ngmarchant/shiw}.
  \item \texttt{RLdata10000}. A synthetic data set provided 
  with the \texttt{RecordLinkage} R 
  package~\citepApp{sariyar_recordlinkage_2010_app}.
  We used \texttt{fname\_c1} and \texttt{lname\_c1} as string-type 
  attributes and \texttt{bd}, \texttt{bm}, \texttt{by} as categorical 
  attributes. 
  The \texttt{fname\_c2} and \texttt{lname\_c2} were excluded as they 
  have a high fraction of missing values.
\end{itemize}

\subsection{Implementation and hardware}
Our implementation of \dblink\ is written in Scala and depends on 
Apache Spark 2.3.1 (a distributed computing framework).
Since \dblink\ requires control over the partitioning (entities and 
linked records \emph{must} reside on their assigned partitions), we 
used the RDD API with a custom partitioner.
Our custom-built server ran in local (pseudo-cluster) mode, with
$2 \times$ 28-core Intel Xeon Platinum 8180M CPUs for a total of 112 
threads (with HyperThreading); and 128GB of allocated RAM on the driver.

\subsection{Parameter settings and initialization} 
\label{app-sec:parameter-settings}
We used the following parameter settings for all experiments.
\begin{itemize}
  \item The distortion hyperparameters $\alpha_{a}, \beta_{a}$ were set to 
  encode a prior mean distortion probability of approximately 1\%, with the 
  strength varying in proportion to the total number of records $R$:
  \begin{equation*}
  \alpha_{a} = R \times 10\% \times 1\% \text{ and }
    \beta_{a} = R \times 10\% \text{ for all } a.
  \end{equation*}
  \item The size of the latent entity population $E$ was set to $R$.
  This corresponds to a prior mean number of observed entities of 
  $(1 - \euler^{-1}) R \approx 0.63 R$, as shown by 
  \citeApp{steorts_bayesian_2016_app}.
  It is important not to set $E$ too low, as it places an upper bound 
  on the number of entities present in the data. 
  \item The entity attribute distributions $\{\phi_a\}$ were set empirically
  based on the observed record attributes. Specifically, we set 
  \begin{equation*}
    \phi_a(v) = \frac{\sum_{t = 1}^{T} \sum_{r = 1}^{R_t} o_{tra} \1{x_{tra} = v}}
      {\sum_{t = 1}^{T} \sum_{r = 1}^{R_t} o_{tra}} \quad \text{for all } a.
  \end{equation*}
  \item For simplicity, we treated all attributes as either 
  ``categorical-type''  with similarity function $\simfn_{\mathrm{const}}$ or 
  ``string-type'' with similarity function $10.0 \times \simfn_{\mathrm{nEd}}$ 
  (these are defined in Section~\ref{sec:attribute-sim-measure}).
  \item The similarity cut-off for string-type attributes was set to 7.0, 
  following advice in the \texttt{RecordLinkage} R 
  package~\citepApp{sariyar_recordlinkage_2010_app}.
  \item We used the $k$-d tree blocking function as defined in 
  Section~\ref{sec:kd}.
  The \emph{reference set} splitting rule was used for input sets with 30 or 
  fewer elements---the \emph{ordered median} splitting rule was used 
  otherwise.
\end{itemize}

To initialize the Markov chain, we linked each record to a unique entity and 
copied the record attributes into the entity attributes, assuming no 
distortion.
Any entity attributes that were missing after this process (due to missing 
record attributes) were filled by drawing an attribute value from the 
empirical distribution.
We set the thinning interval to 10---i.e. we only saved every tenth step along 
the chain.
This increases the effective sample size for a given storage budget.

\section{Results on Amazon EC2}
\label{app-sec:cloud-results}
We repeated two of the experiments described in 
Section~\ref{sec:efficiency-expts} of the main paper on a cluster 
running in the Amazon Elastic Compute Cloud (EC2).
For the worker (executor) nodes, we used varying numbers of \texttt{m5.xlarge} 
instances with 4 vCores, 16 GiB memory and 32 GiB of Elastic Block Store 
(EBS) storage. 
Due to the increased latency and decreased bandwidth between the 
compute nodes, we expected the efficiency to decrease.
This is indeed what we observed.

Figure~\ref{app-fig:speed-up-vs-num-partitions} plots the speed-up as a function 
of the number of blocks $B$ relative to a baseline with no partitioning.
We observe poorer scaling with $B$ compared to the results we obtained 
on our local server (c.f.~Figure~\ref{fig:speed-up-vs-num-partitions} 
in the main paper).
Figure~\ref{app-fig:speed-up-vs-sampler} plots the efficiency as a function of 
the sampling method with $B=16$.
The results are qualitatively similar to the ones we obtained using our local 
server (c.f.~Figure~\ref{fig:speed-up-vs-sampler} in the main paper).
However, the ESS rate was reduced for all samplers as expected due to 
increased communication costs.

\begin{figure}
  \centering
  \includegraphics{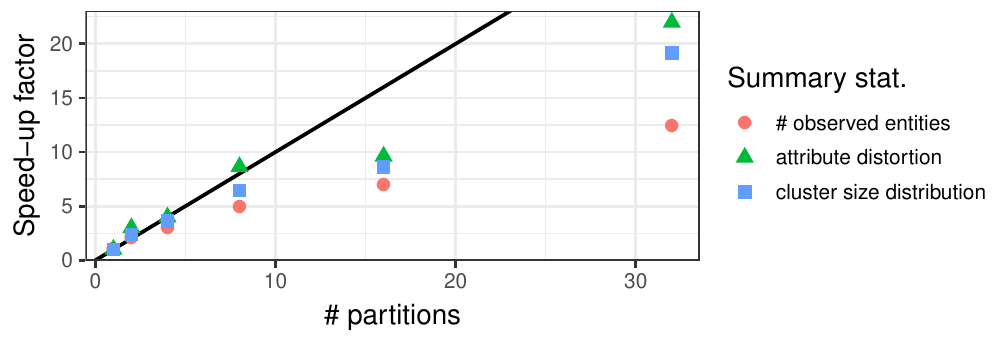}
  \caption{Efficiency of \dblink\ as a function of the
    number of blocks $B$ and summary statistic of interest 
    (larger is better).
    The speed-up measures the ESS rate relative to the ESS rate 
    for $B = 1$ (no partitioning) for the \texttt{NLTCS} data set.}
  \label{app-fig:speed-up-vs-num-partitions}
\end{figure}

\begin{figure}
  \centering
  \includegraphics{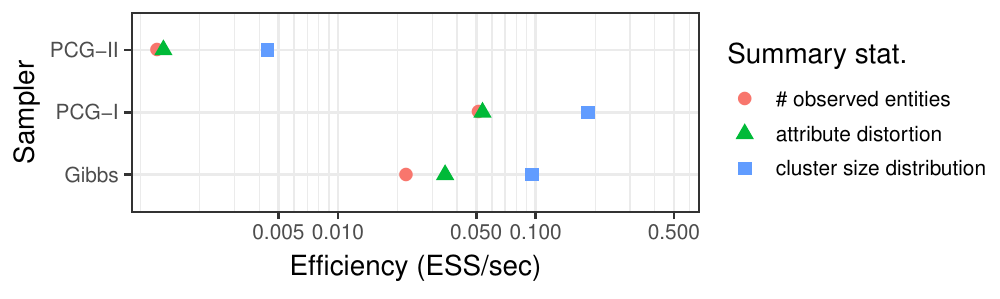}
  \caption{Efficiency of \dblink\ as a function of the 
    sampler and summary statistic of interest (larger is better).
    All measurements are for the \texttt{NLTCS} data set with 
    $B = 16$.
    }
  \label{app-fig:speed-up-vs-sampler}
\end{figure}

\section{Balance of the blocks}
\label{app-sec:balance-partitions}
In Section~\ref{sec:kd}, we proposed a blocking function based on $k$-d 
trees, and argued that it could yield balanced blocks with good entity 
separation.
While running \dblink\ with the $k$-d tree blocking function, we recorded the 
size of the blocks ($|\partset_{b}|$ for all $b$) to assess whether 
they were well-balanced.
Figure~\ref{app-fig:partition-sizes} illustrates the results in terms of the 
relative absolute deviation from the perfectly balanced configuration
(where the entities are divided equally among the blocks). 
We can see that the $k$-d tree partitioner is functioning quite well---the 
deviation from the perfectly balanced configuration is no more than 10\% for 
all data sets.

\begin{figure}
  \centering
  \includegraphics{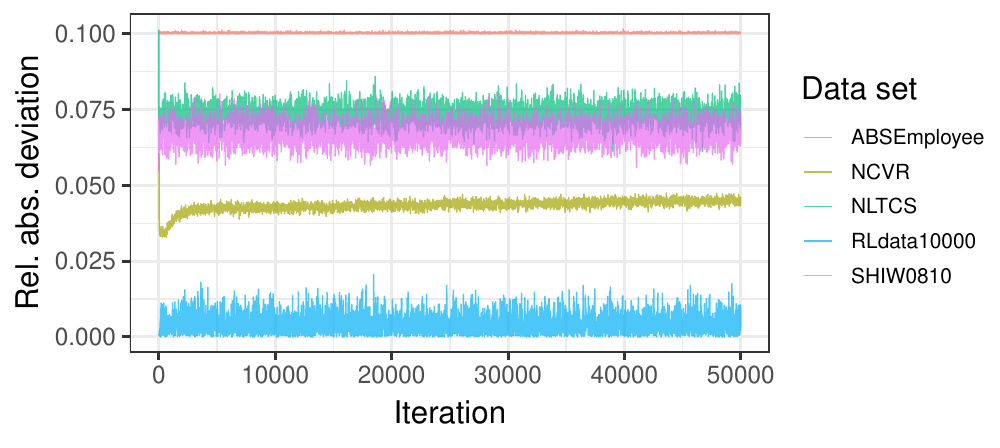}
  \caption{Balance of the blocks for a single run on each
    data set.
    The balance is measured in terms of the relative absolute 
    deviation from the perfectly balanced configuration.
    The number of blocks $B = 64, 64, 16, 2, 8$ for each data set (in the 
    order listed in the legend).
  }
  \label{app-fig:partition-sizes}
\end{figure}

\section{Uncertainty measures}
\label{app-sec:error-num-ents}
\dblink\ allows for measures of uncertainty to be reported, unlike the 
baseline methods, since we have the full posterior distribution.
For example, in Figure~\ref{app-fig:error-num-ents} we compute posterior estimates 
for the number of entities present in each data set, with 95\% 
Bayesian credible intervals.
Note that the posterior estimates are typically quite 
sharp.
This seems to confirm arguments by~\citeApp{steorts_bayesian_2016_app} 
regarding the informativeness of the prior for the linkage 
structure in \blink. 
Research on less informative priors is ongoing~\citepApp{zanella_flexible_2016_app}.

\begin{figure}
  \centering
  \includegraphics{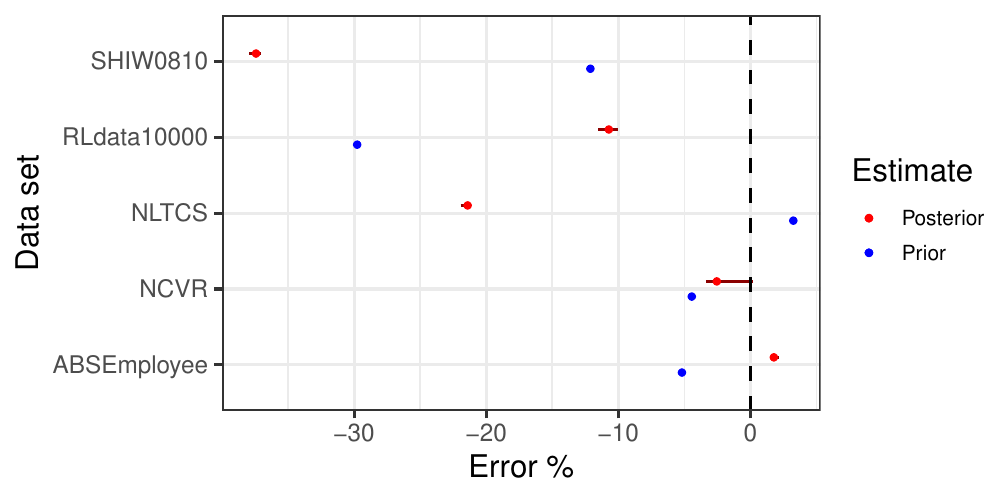}
  \caption{Percentage error in the posterior\slash prior estimates for the 
    number of observed entities for \dblink.
    The posterior estimates are generally sharp and underestimate the true 
    number of observed entities.
  }
  \label{app-fig:error-num-ents}
\end{figure}

\section{Sensitivity analysis}
\label{app-sec:sensitivity}
We conducted an empirical sensitivity analysis for \dblink\ 
using the \texttt{RLdata10000} data set.
We selected this data set as it is relatively small, which made it 
quick to run the inference for various hyperparameter combinations.
The parameters tested were:
\begin{itemize}
  \item $\alpha_\ell, \beta_\ell$: the shape parameters for the Beta prior on 
  the distortion probabilities. 
  We used the same values for all attributes ($a$). 
  \item $E$: the size of the latent population.
  \item $s_{\mathrm{max}}$: the scaling factor for the similarity function. 
  This controls the inverse temperature of the softmax distribution for the 
  distorted attribute values. 
\end{itemize}
We varied each of these parameters in turn, while holding all other parameters 
fixed.
For the Beta prior on the distortion probabilities, we first varied the 
strength while fixing the prior mean to $\sim 1\%$, then we varied the mean 
(1\%, 5\% and 10\%) while fixing $\alpha + \beta$ (related to the strength). 
Table~\ref{tbl:sensitivity} presents the evaluation measures for each 
combination of parameters.
The results indicate that the inferred linkage structure is relatively 
sensitive to all of the parameters, however sensitivity is in general 
predictable, following clear and intuitive trends.
Of particular interest is the fact that the model performs best when the 
Beta prior on the distortion probabilities is sharply peaked near zero. 
It seems that the model has a tendency to overestimate the amount of 
distortion, particularly in the absence of ground truth.

\begin{table}
	\centering
  \caption{Sensitivity analysis for various parameters combinations 
  using \texttt{RLdata10000}. The first group of rows tests the effect of 
  varying the \emph{strength} of the Beta prior, the second group tests the 
  effect of varying the \emph{mean} of the Beta prior, the third group 
  tests the effect of varying the population size, and the fourth group 
  tests the effect of varying the scaling factor for the similarity function.}
	\label{tbl:sensitivity}
	\spacingset{1}
  \footnotesize
  \begin{center}
	\begin{tabular}{*{9}{c}}
		\toprule
		\multicolumn{2}{c}{Distortion} & Pop.\ size & Max.\ sim. & \multicolumn{3}{c}{Pairwise measures} & \multicolumn{2}{c}{Cluster measures} \\
		\cmidrule(lr){1-2} \cmidrule(lr){3-3} \cmidrule(lr){4-4} \cmidrule(lr){5-7} \cmidrule(lr){8-9}
		$\alpha$       & $\beta$          & $E$            & $s_{\mathrm{max}}$ & Precision & Recall & F1-score & ARI & Err. \# clust.\ \\
		\midrule 
    \textbf{0.1}   & \textbf{10.0}    & 10000          & 10.0          & 0.5342 & 0.9990 & 0.6962 & 0.6962 & $-17.47\%$ \\
    \textbf{1.0}   & \textbf{100.0}   & 10000          & 10.0          & 0.5435 & 0.9990 & 0.7040 & 0.7040 & $-16.58\%$ \\
    \textbf{10.0}  & \textbf{1000.0}  & 10000          & 10.0          & 0.6334 & 0.9970 & 0.7747 & 0.7747 & $-10.97\%$ \\
    \textbf{100.0} & \textbf{10000.0} & 10000          & 10.0          & 0.9180 & 0.9850 & 0.9503 & 0.9503 & $-1.595\%$ \\
    \midrule
    \textbf{10.0}  & \textbf{1000.0}  & 10000          & 10.0          & 0.6334 & 0.9970 & 0.7747 & 0.7747 & $-10.97\%$ \\
    \textbf{50.5}  & \textbf{959.5}   & 10000          & 10.0          & 0.6132 & 0.9970 & 0.7593 & 0.7593 & $-11.90\%$ \\
    \textbf{101.0} & \textbf{909.0}   & 10000          & 10.0          & 0.5992 & 0.9970 & 0.7485 & 0.7485 & $-12.90\%$ \\
    \midrule
    10.0           & 1000.0           & \textbf{9000}  & 10.0          & 0.5306 & 0.9970 & 0.6926 & 0.6926 & $-15.65\%$ \\
    10.0           & 1000.0           & \textbf{10000} & 10.0          & 0.6334 & 0.9970 & 0.7747 & 0.7747 & $-10.97\%$ \\
    10.0           & 1000.0           & \textbf{11000} & 10.0          & 0.6999 & 0.9960 & 0.8221 & 0.8221 & $-7.365\%$ \\
    \midrule
    10.0           & 1000.0           & 10000          & \textbf{5.0}  & 0.6927 & 0.9940 & 0.8164 & 0.8164 & $-22.12\%$ \\
    10.0           & 1000.0           & 10000          & \textbf{10.0} & 0.6334 & 0.9970 & 0.7747 & 0.7747 & $-10.97\%$ \\
    10.0           & 1000.0           & 10000          & \textbf{50.0} & 0.2112 & 0.3920 & 0.2745 & 0.2745 & $-12.50\%$ \\
		\bottomrule
  \end{tabular}
  \end{center}
\end{table}

\section{Details of inference for the case study to the 2010 Decennial Census}
\label{app-sec:census-mcmc}
We ran inference for 15,000 iterations using the PCG-I sampler.
After removing 5,000 iterations as burn-in and applying thinning with 
an interval of 10, we obtained 1,000 approximate samples from the 
posterior.
Convergence diagnostics are consistent with those reported for 
the other data sets in Appendix~\ref{app-sec:trace-plots}, 
and are complicated to release due to the fact that the data is 
protected under Title~13. 
Releasing each iteration of a Gibbs sampler could potentially say something 
about individuals in the population, and thus, for privacy reasons, these 
diagnostics are omitted.

\pagebreak
\section{Trace plots}
\label{app-sec:trace-plots}
\subsection{Attribute-level distortion}
The following figures relate to the aggregate distortion per attribute 
for each data set.
On the left are the trace plots, which show the aggregate distortion 
for each attribute (stacked vertically) along the Markov chain.
On the right are the corresponding autocorrelation plots.

\begin{figure}[H]
  \includegraphics[width=0.48\linewidth]{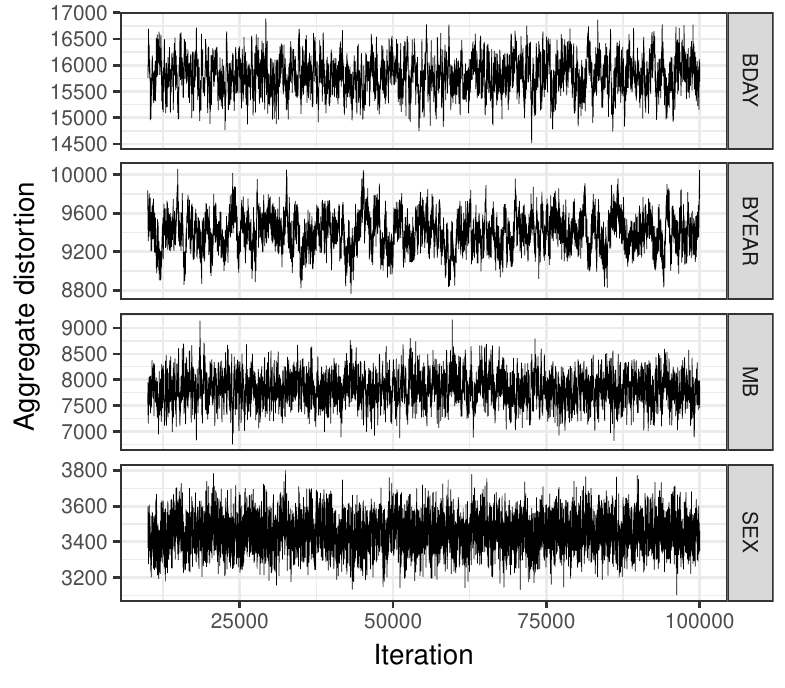} \hfill
  \includegraphics[width=0.48\linewidth]{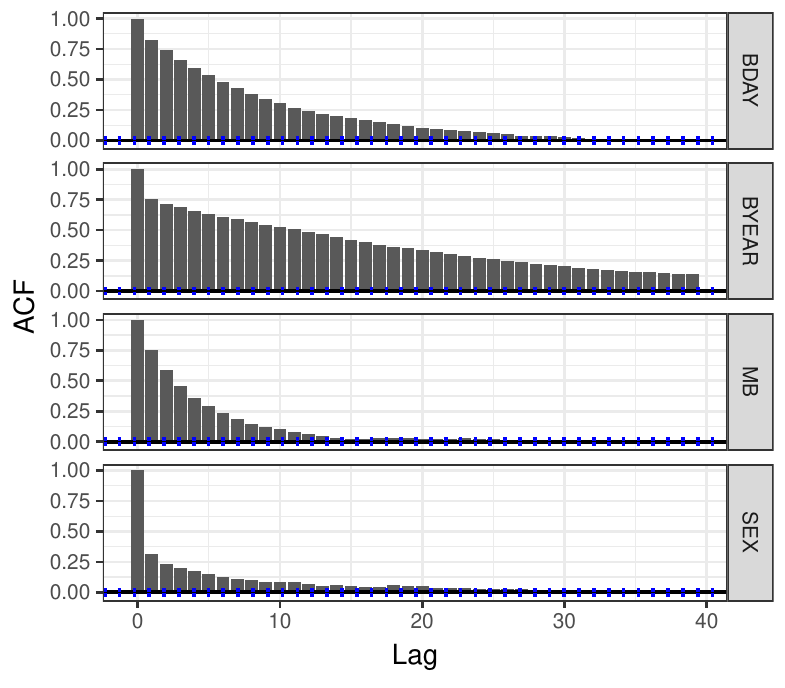}
  \caption{Attribute-level distortion for \texttt{ABSEmployee}}
\end{figure}
\begin{figure}[H]
  \includegraphics[width=0.48\linewidth]{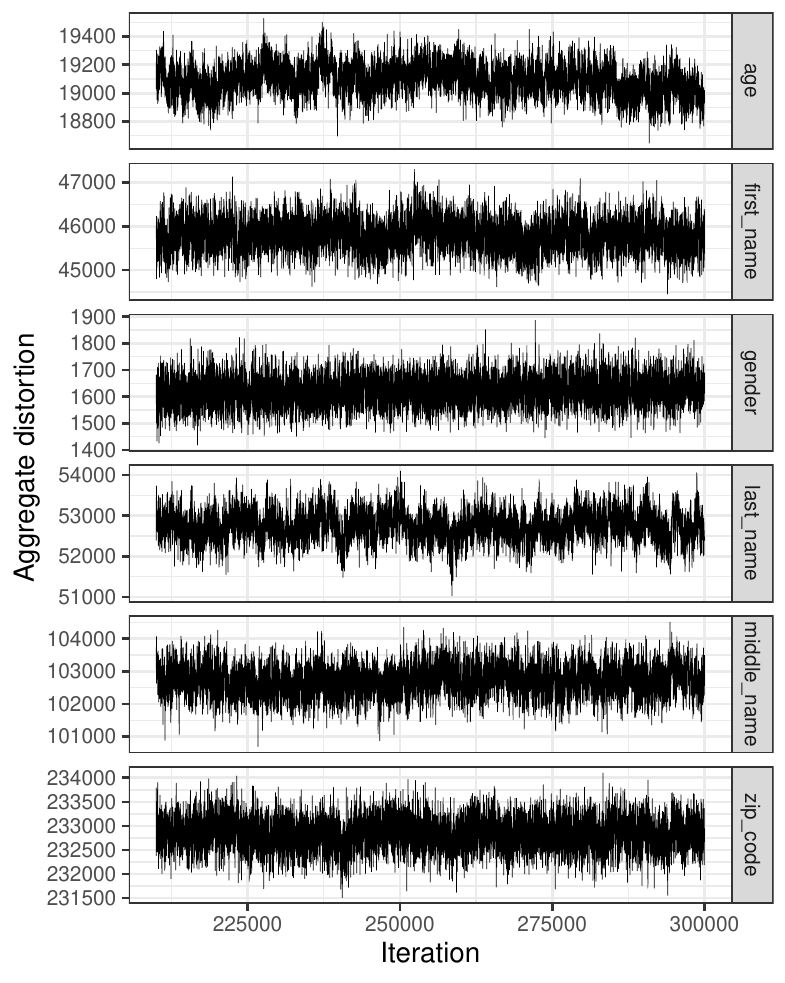} \hfill
  \includegraphics[width=0.48\linewidth]{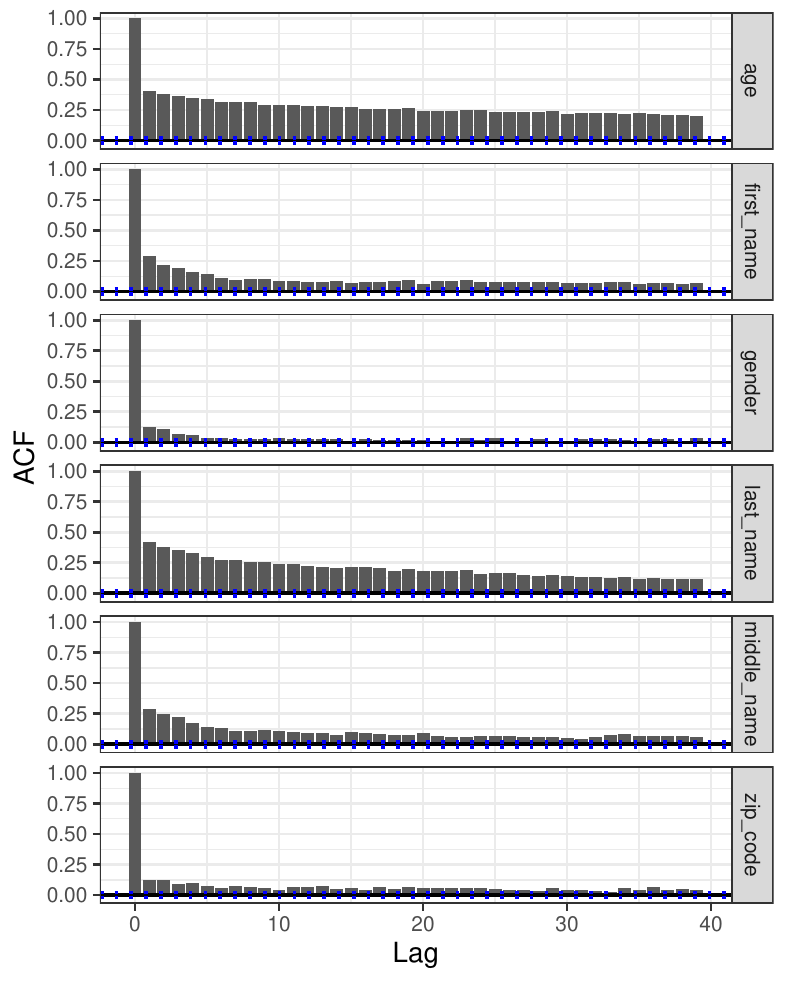}
  \caption{Attribute-level distortion for \texttt{NCVR}}
\end{figure}
\begin{figure}[H]
  \includegraphics[width=0.48\linewidth]{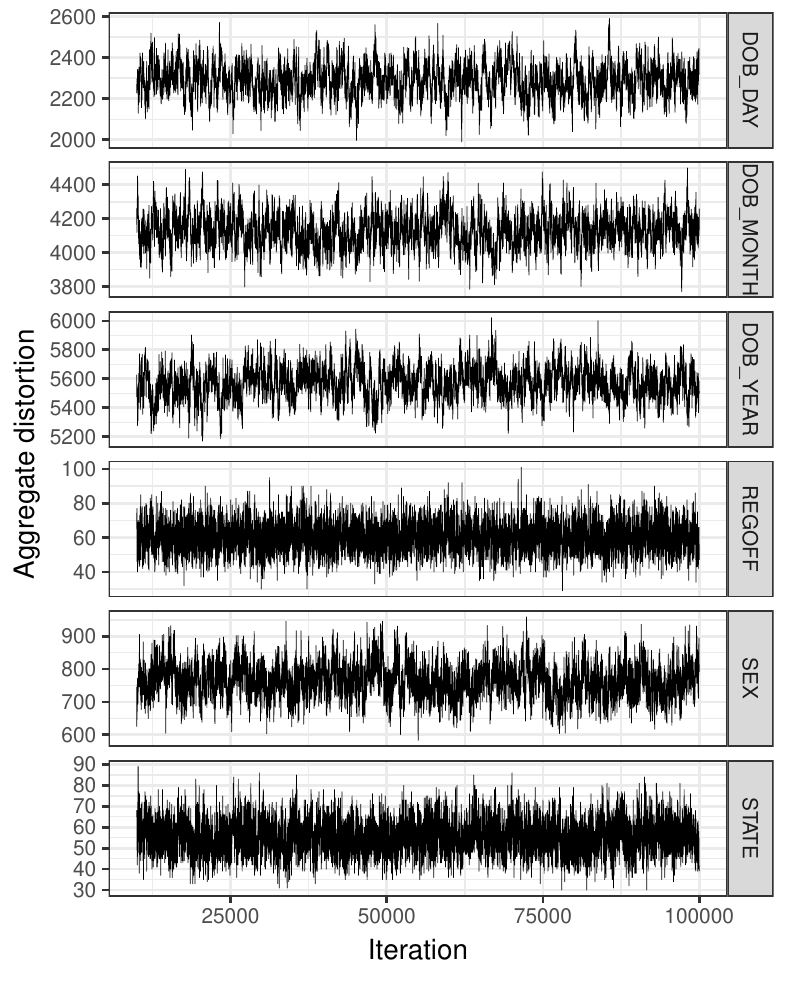} \hfill
  \includegraphics[width=0.48\linewidth]{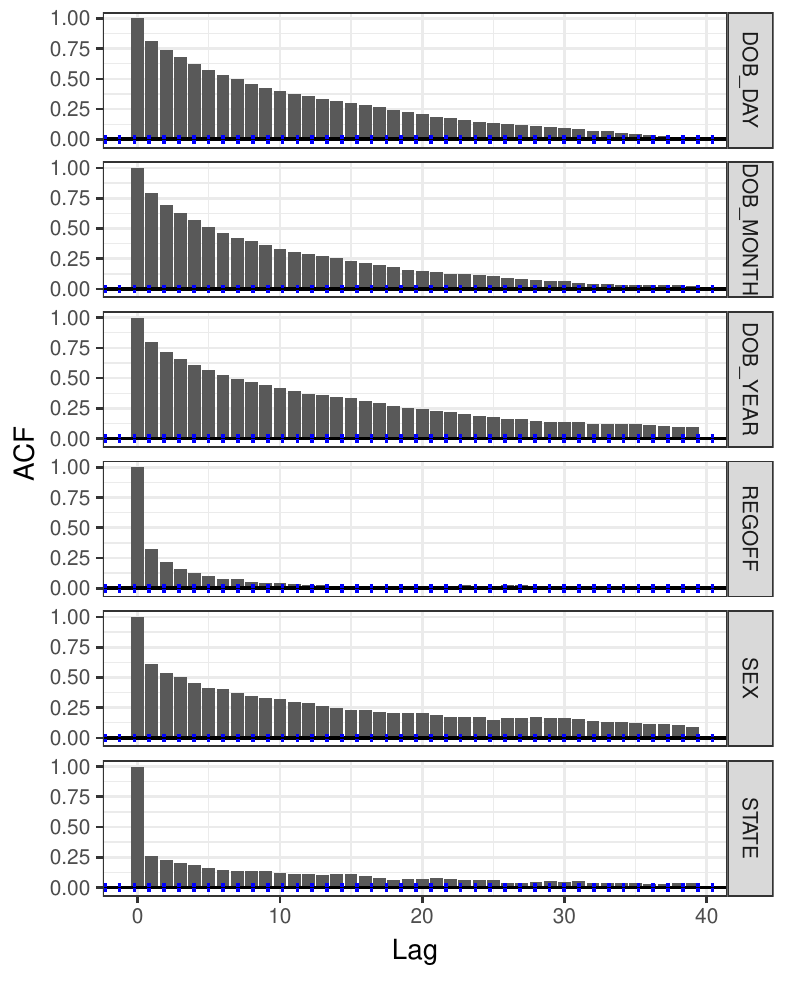}
  \caption{Attribute-level distortion for \texttt{NLTCS}}
\end{figure}
\begin{figure}[H]
  \includegraphics[width=0.48\linewidth]{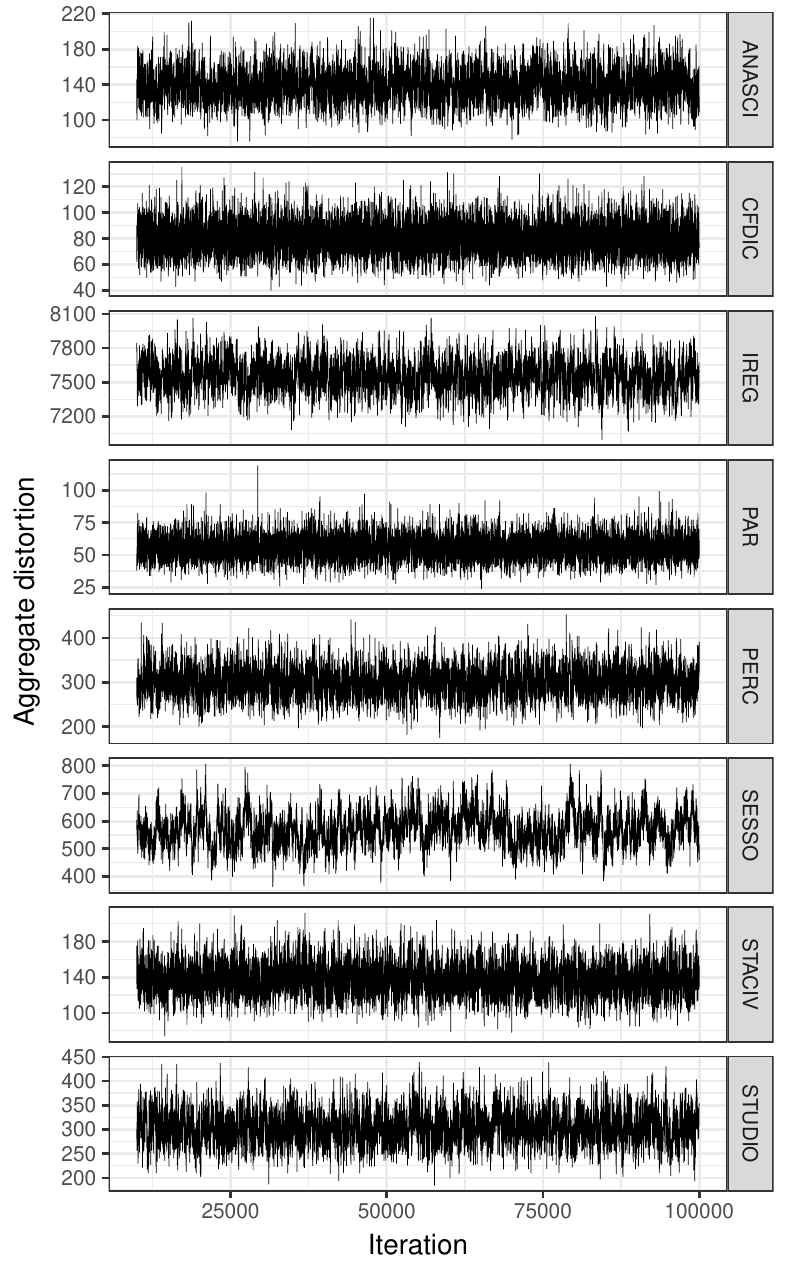} \hfill
  \includegraphics[width=0.48\linewidth]{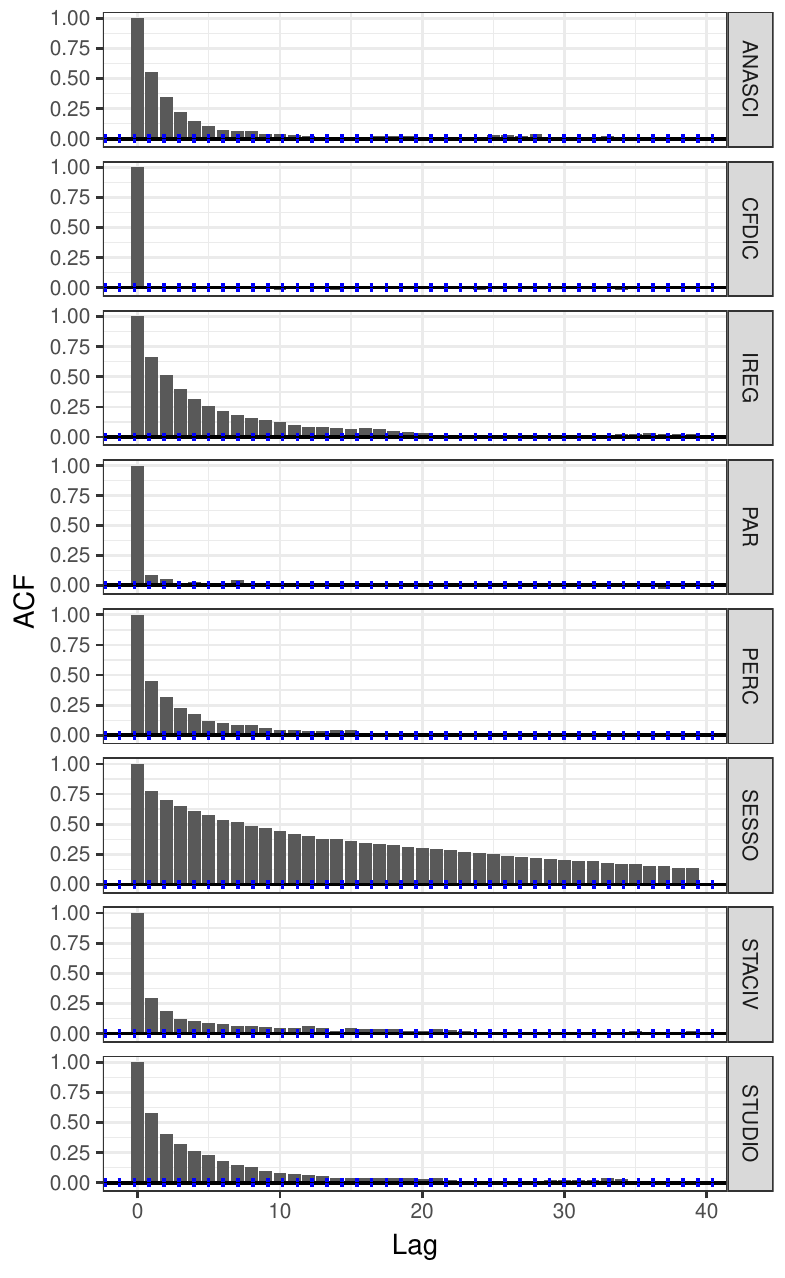}
  \caption{Attribute-level distortion for \texttt{SHIW0810}}
\end{figure}
\begin{figure}[H]
  \includegraphics[width=0.48\linewidth]{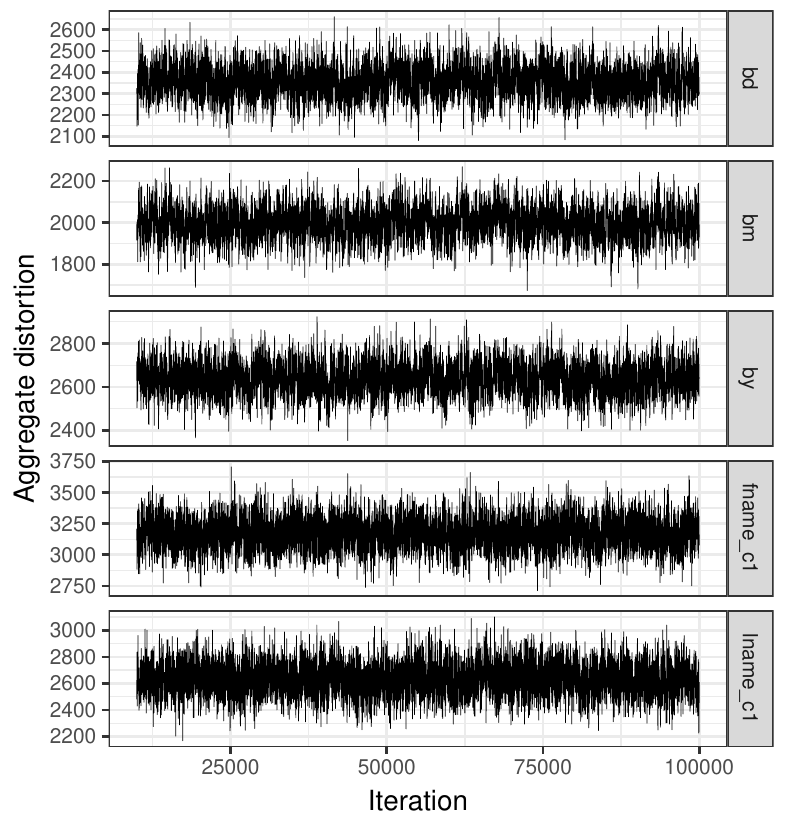} \hfill
  \includegraphics[width=0.48\linewidth]{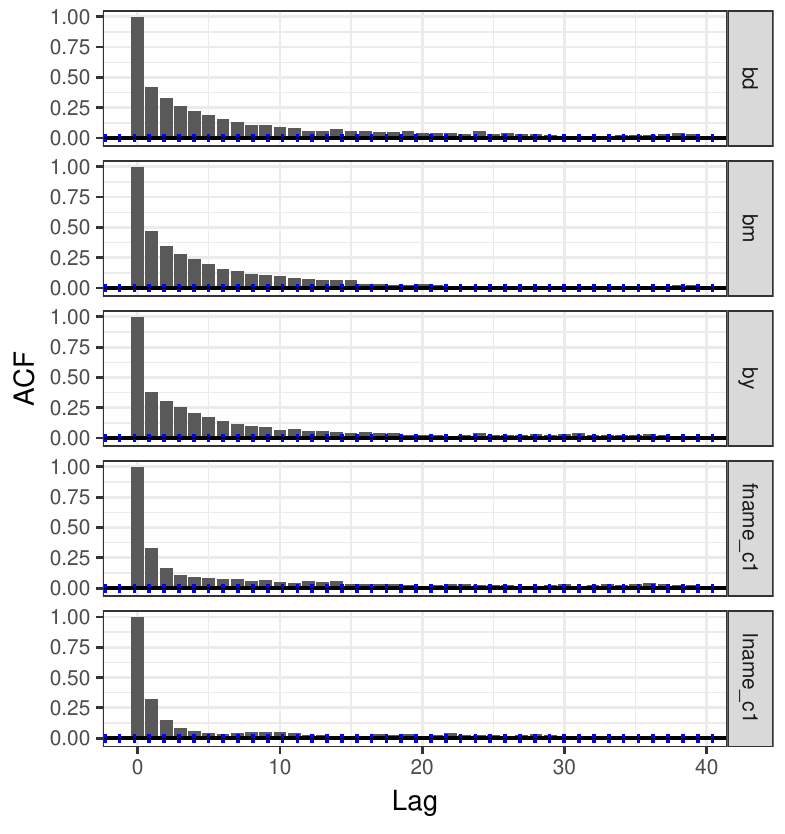}
  \caption{Attribute-level distortion for \texttt{RLdata10000}}
\end{figure}

\subsection{Distribution of record distortion}
The following figures relate to the distribution of record distortion 
for each data set.
Specifically, we count the number of records with 0 distorted attributes, 
1 distorted attribute, 2 distorted attributes, etc.
On the left are the trace plots, which show the record counts for each 
distortion level (stacked vertically) along the Markov chain.
On the right are the corresponding autocorrelation plots.

\begin{figure}[H]
  \includegraphics[width=0.48\linewidth]{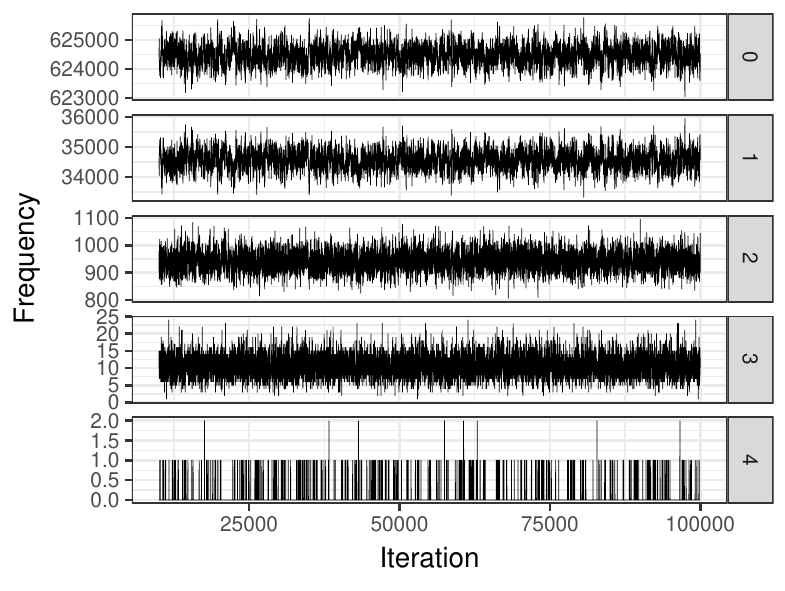} \hfill
  \includegraphics[width=0.48\linewidth]{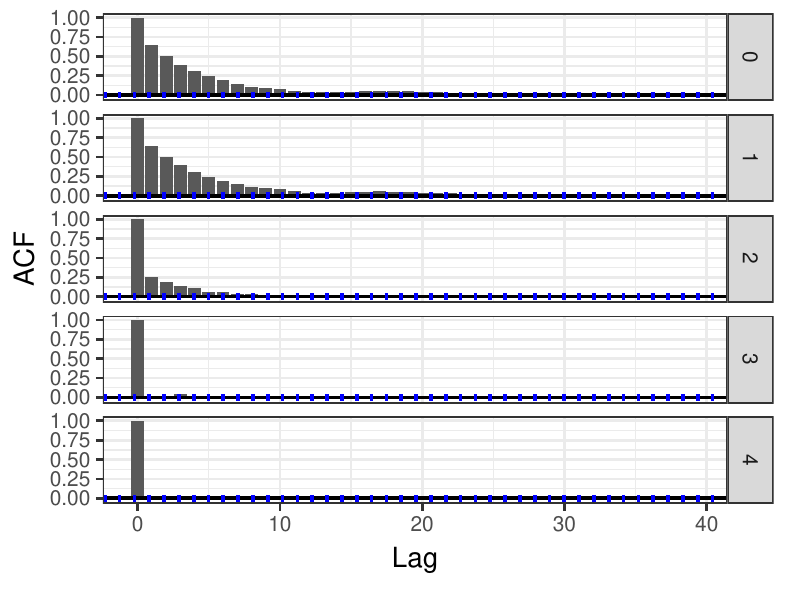}
  \caption{Distribution of record distortion for \texttt{ABSEmployee}}
\end{figure}
\begin{figure}[H]
  \includegraphics[width=0.48\linewidth]{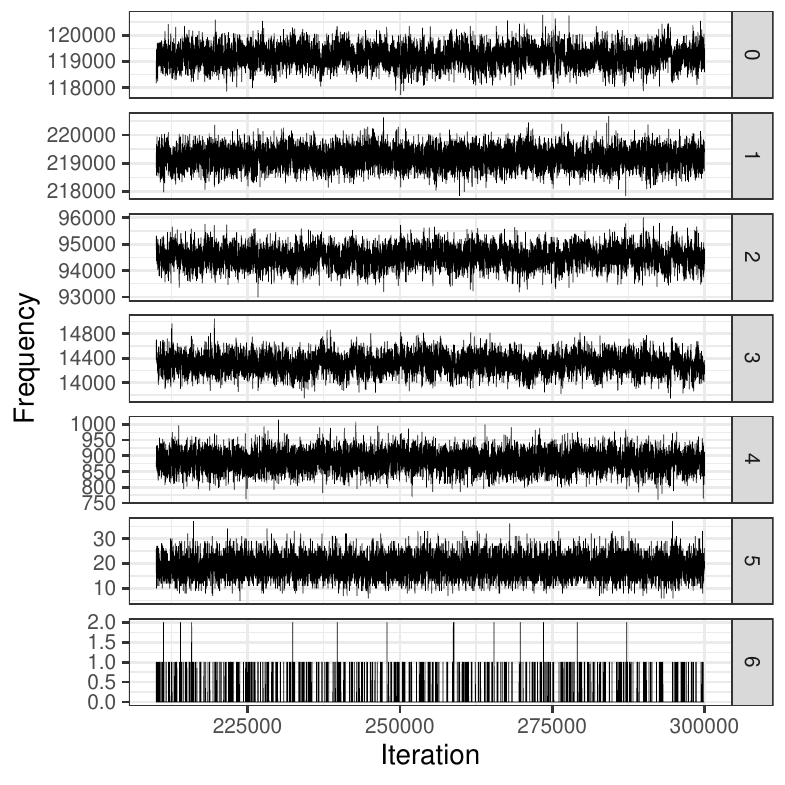} \hfill
  \includegraphics[width=0.48\linewidth]{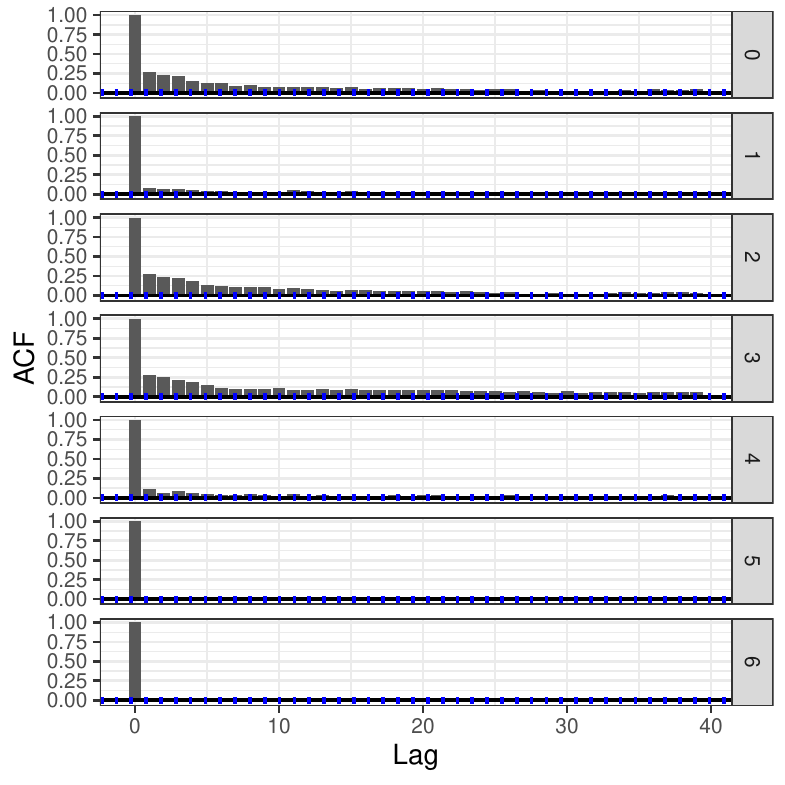}
  \caption{Distribution of record distortion for \texttt{NCVR}}
\end{figure}
\begin{figure}[H]
  \includegraphics[width=0.48\linewidth]{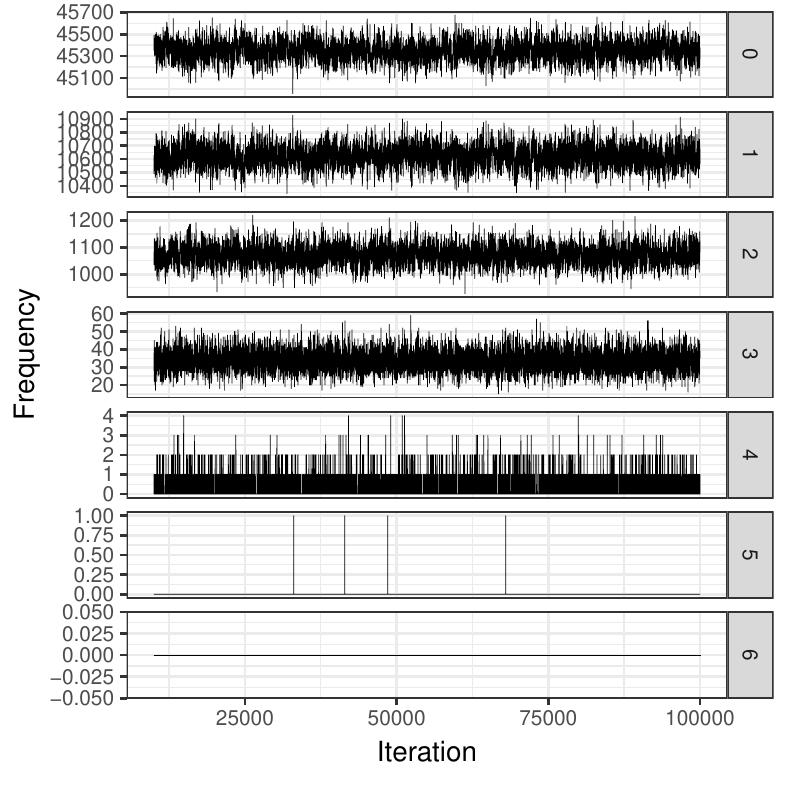} \hfill
  \includegraphics[width=0.48\linewidth]{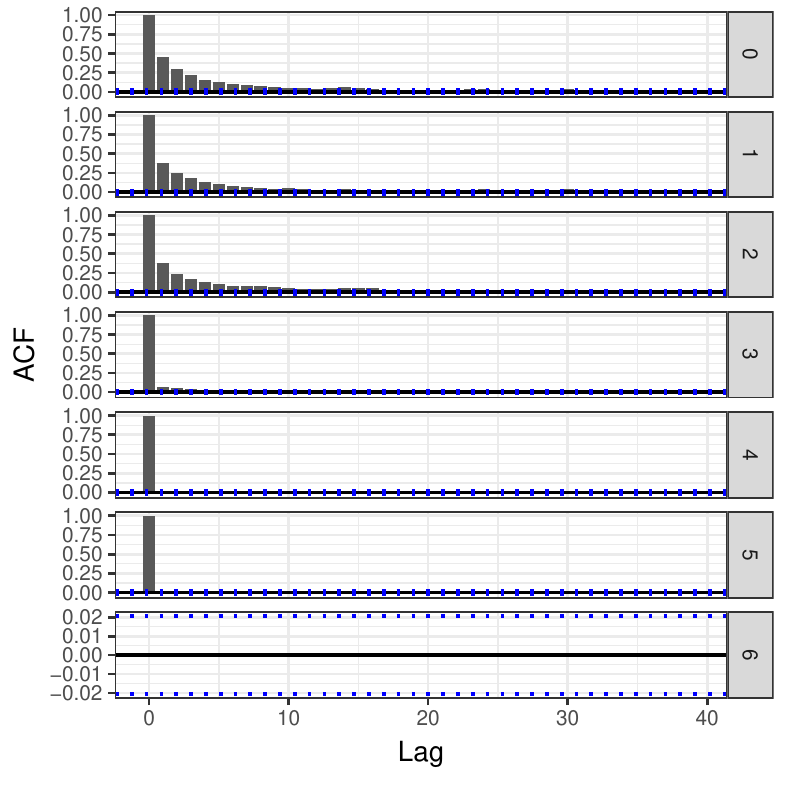}
  \caption{Distribution of record distortion for \texttt{NLTCS}}
\end{figure}
\begin{figure}[H]
  \includegraphics[width=0.48\linewidth]{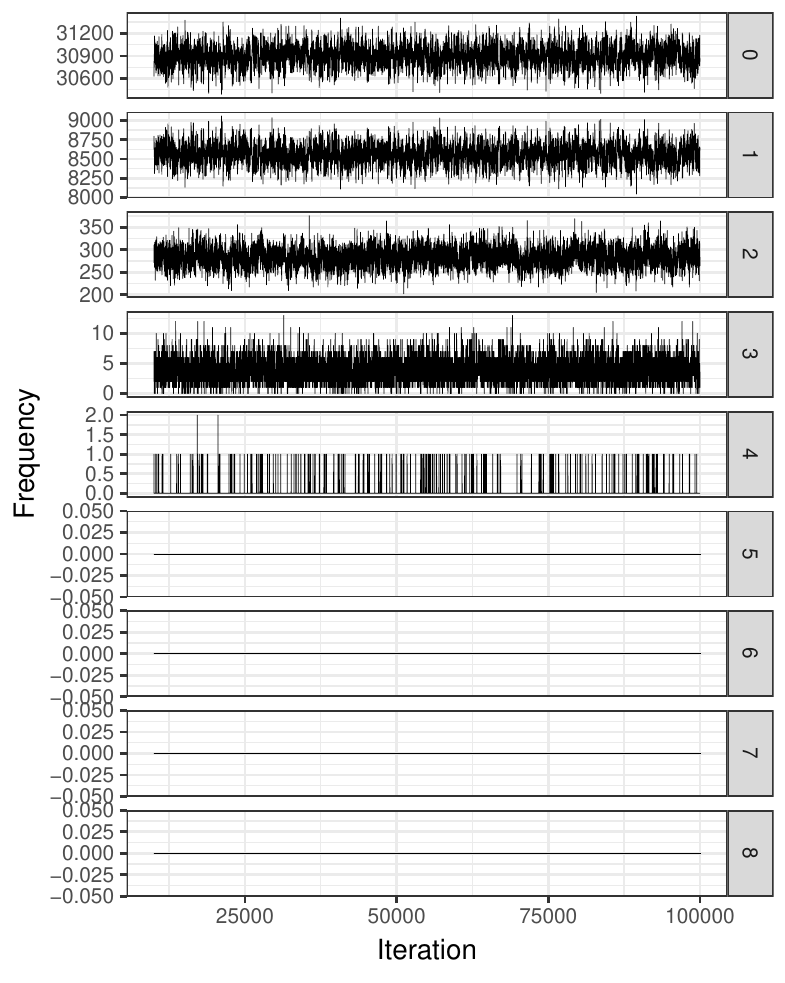} \hfill
  \includegraphics[width=0.48\linewidth]{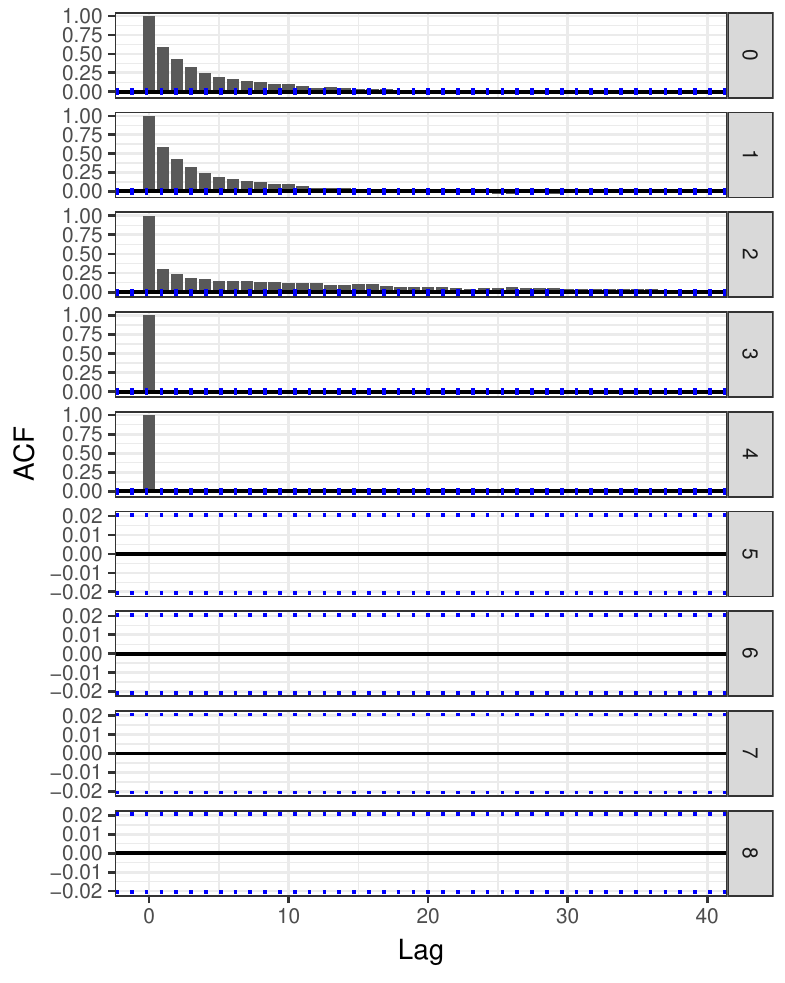}
  \caption{Distribution of record distortion for \texttt{SHIW0810}}
\end{figure}
\begin{figure}[H]
  \includegraphics[width=0.48\linewidth]{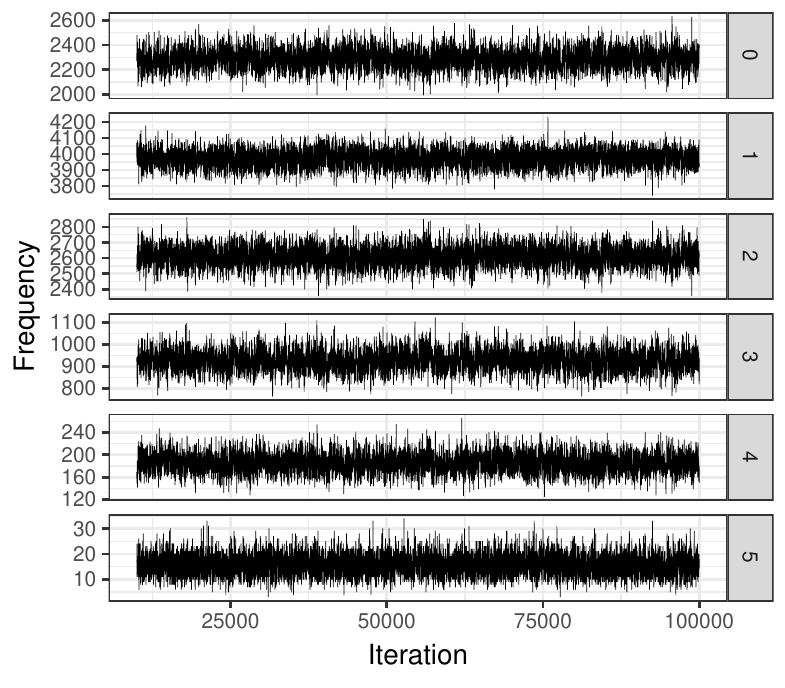} \hfill
  \includegraphics[width=0.48\linewidth]{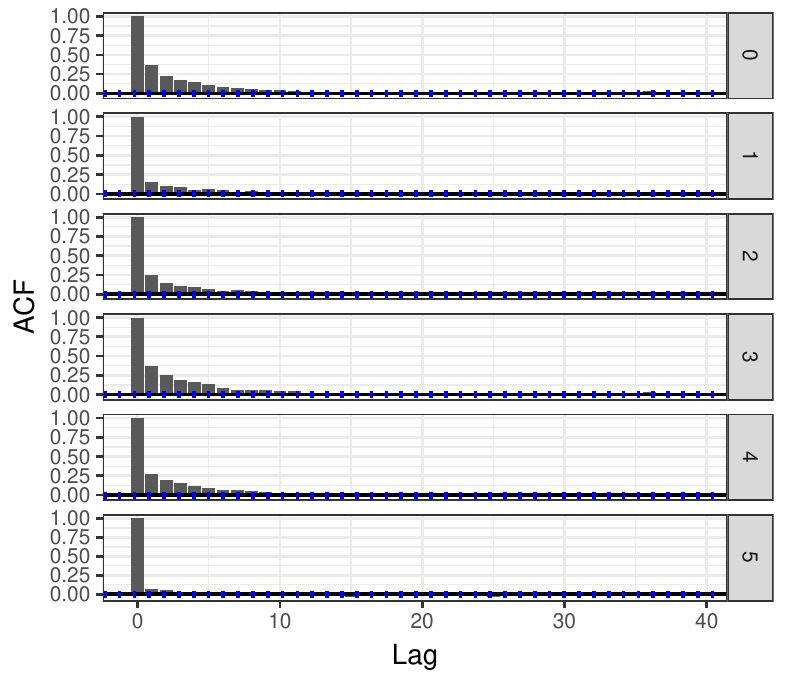}
  \caption{Distribution of record distortion for \texttt{RLdata10000}}
\end{figure}

\subsection{Cluster size distribution}
The following figures relate to the distribution of cluster (entity) sizes 
for each data set.
Specifically, we count the number of entities with 0 linked records, 
1 linked record, 2 linked records, etc.
On the left are the trace plots, which show the counts for each cluster size 
(stacked vertically) along the Markov chain.
On the right are the corresponding autocorrelation plots.

\begin{figure}[H]
  \includegraphics[width=0.48\linewidth]{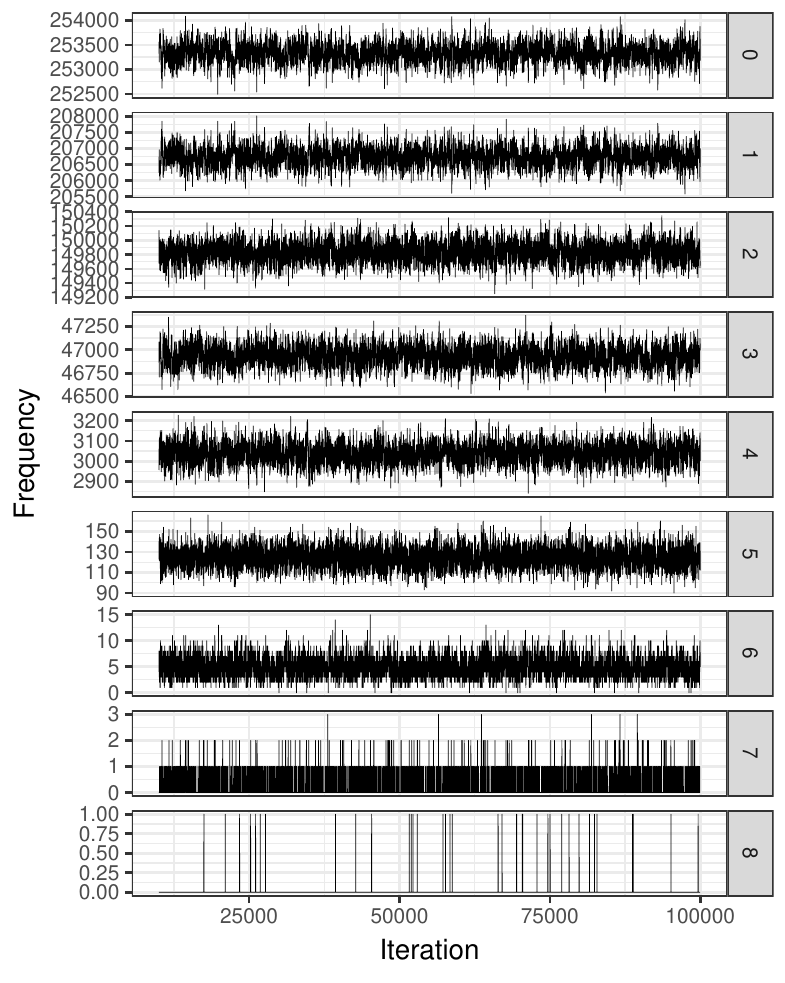} \hfill
  \includegraphics[width=0.48\linewidth]{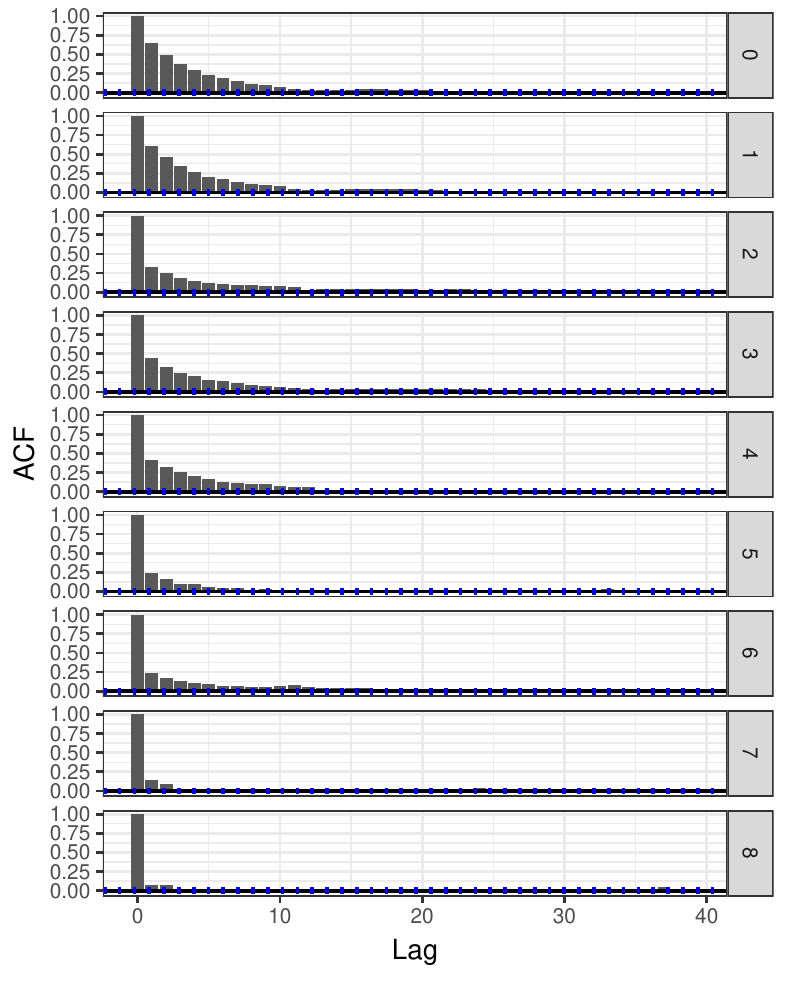}
  \caption{Cluster size distribution for \texttt{ABSEmployee}}
\end{figure}
\begin{figure}[H]
  \includegraphics[width=0.48\linewidth]{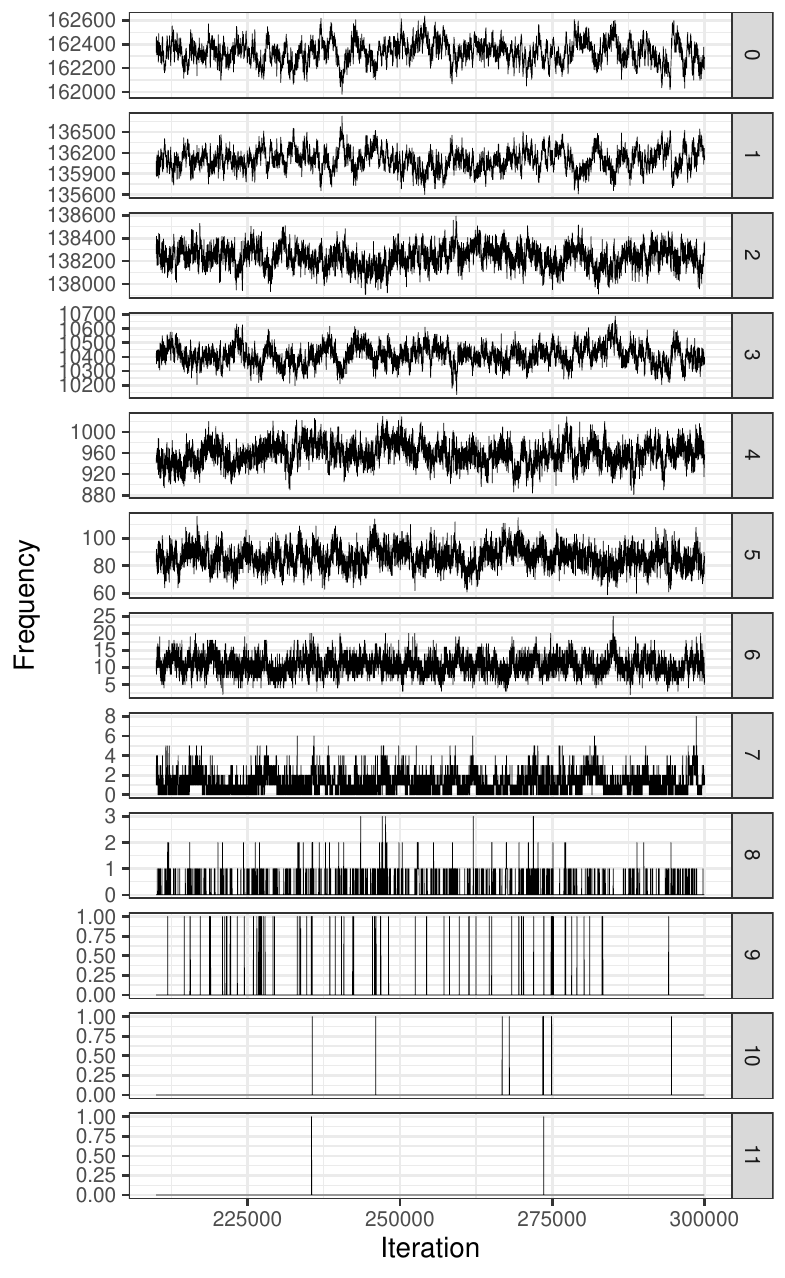} \hfill
  \includegraphics[width=0.48\linewidth]{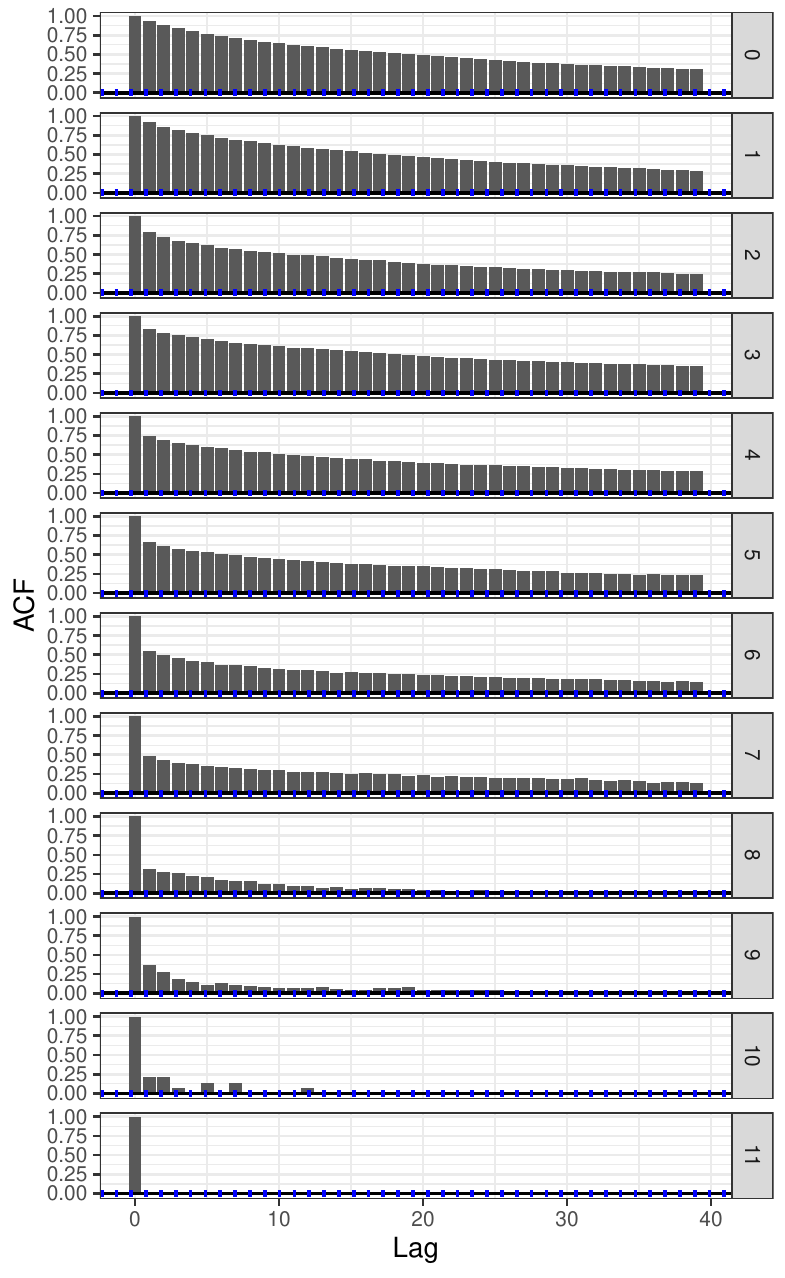}
  \caption{Cluster size distribution for \texttt{NCVR}}
\end{figure}
\begin{figure}[H]
  \includegraphics[width=0.48\linewidth]{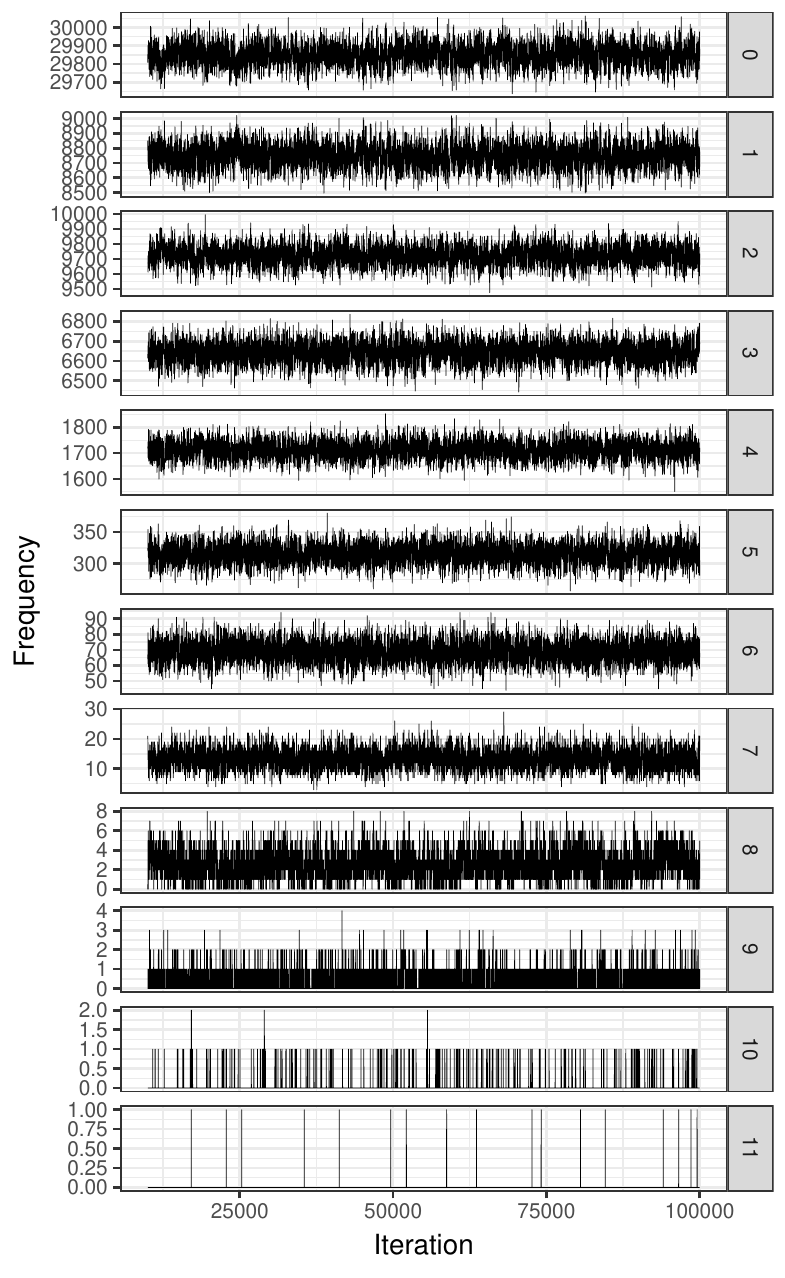} \hfill
  \includegraphics[width=0.48\linewidth]{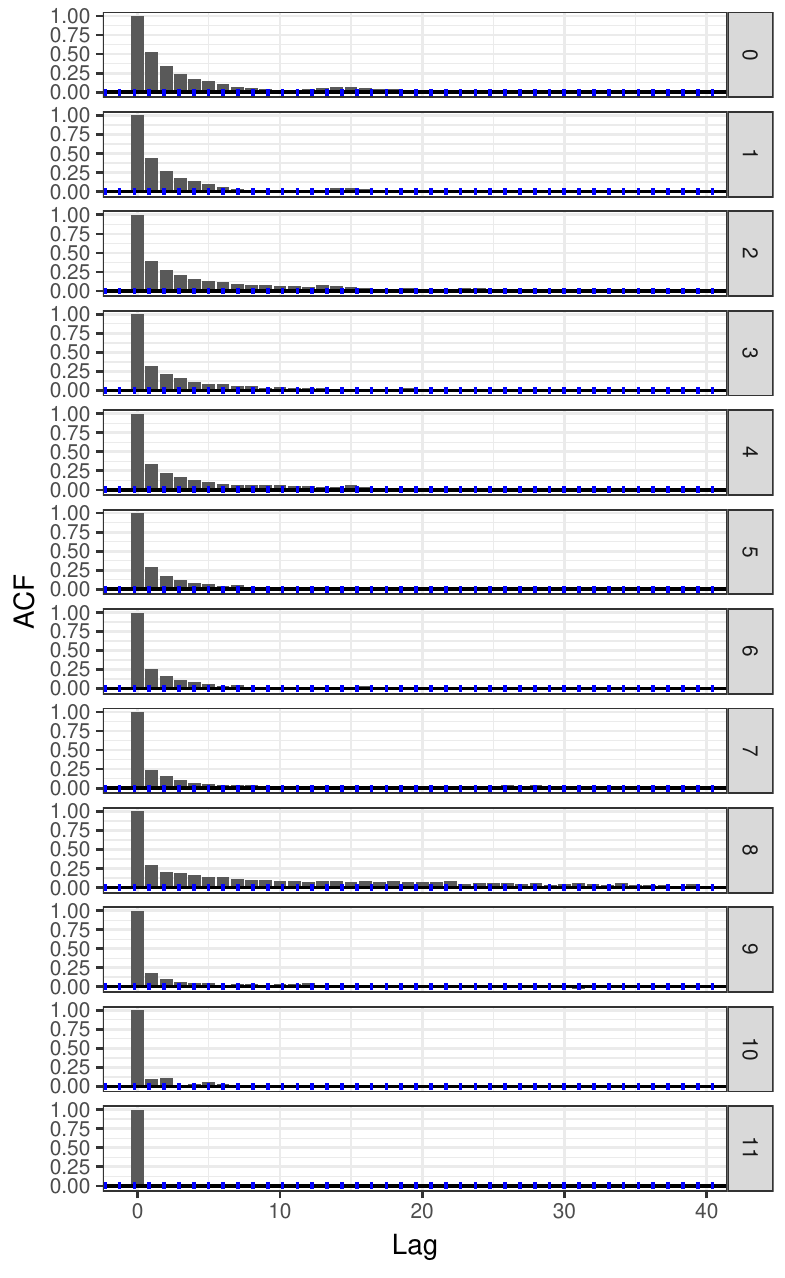}
  \caption{Cluster size distribution for \texttt{NLTCS}}
\end{figure}
\begin{figure}[H]
  \includegraphics[width=0.48\linewidth]{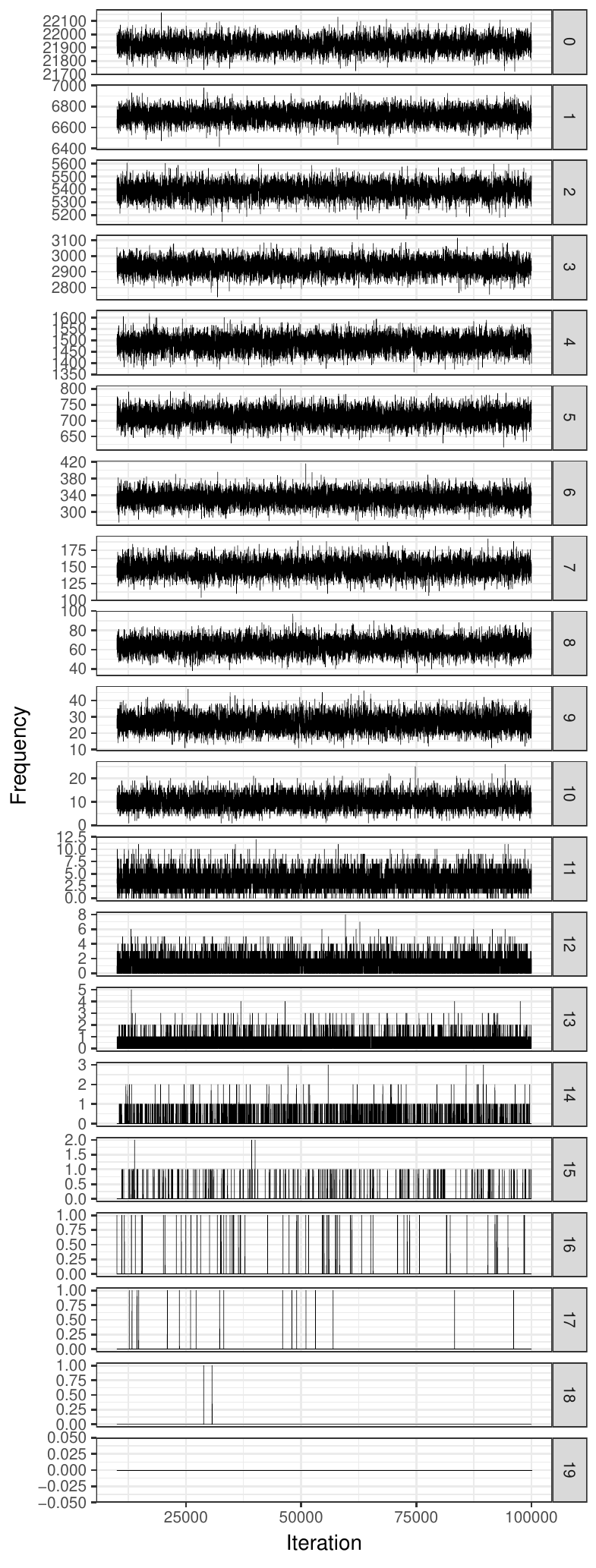} \hfill
  \includegraphics[width=0.48\linewidth]{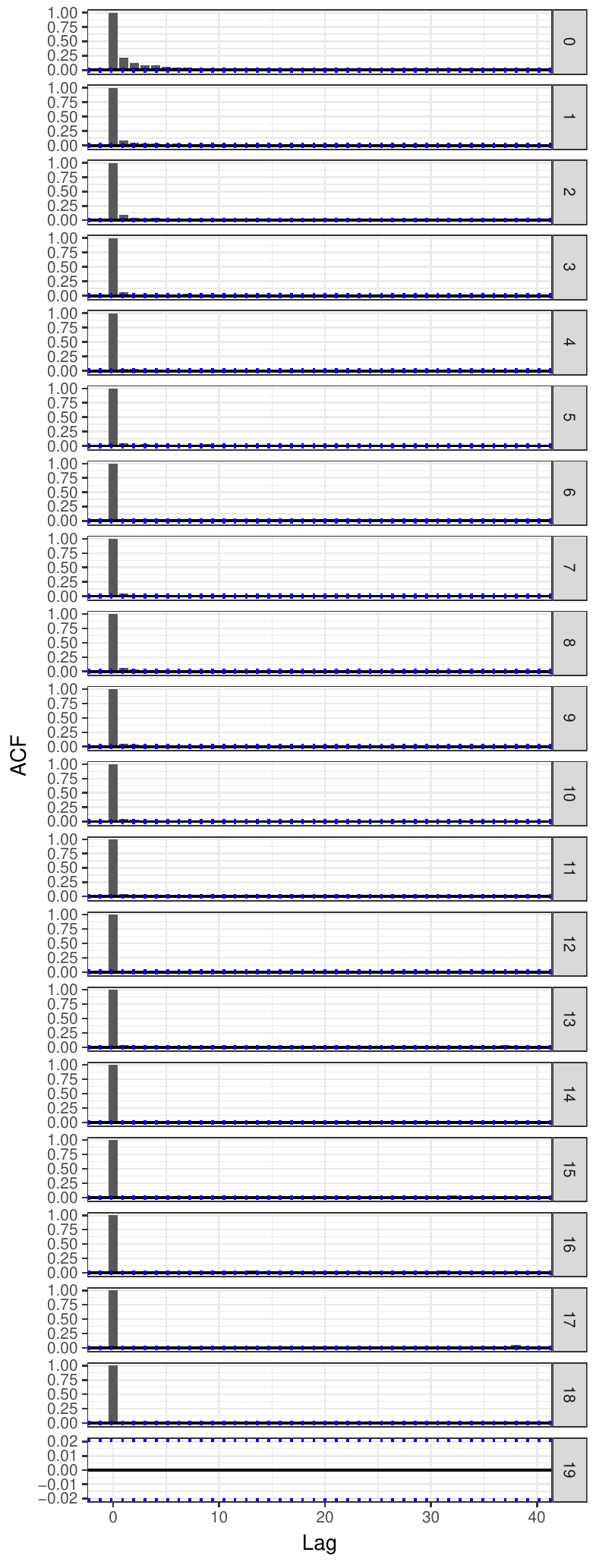}
  \caption{Cluster size distribution for \texttt{SHIW0810}}
\end{figure}
\begin{figure}[H]
  \includegraphics[width=0.48\linewidth]{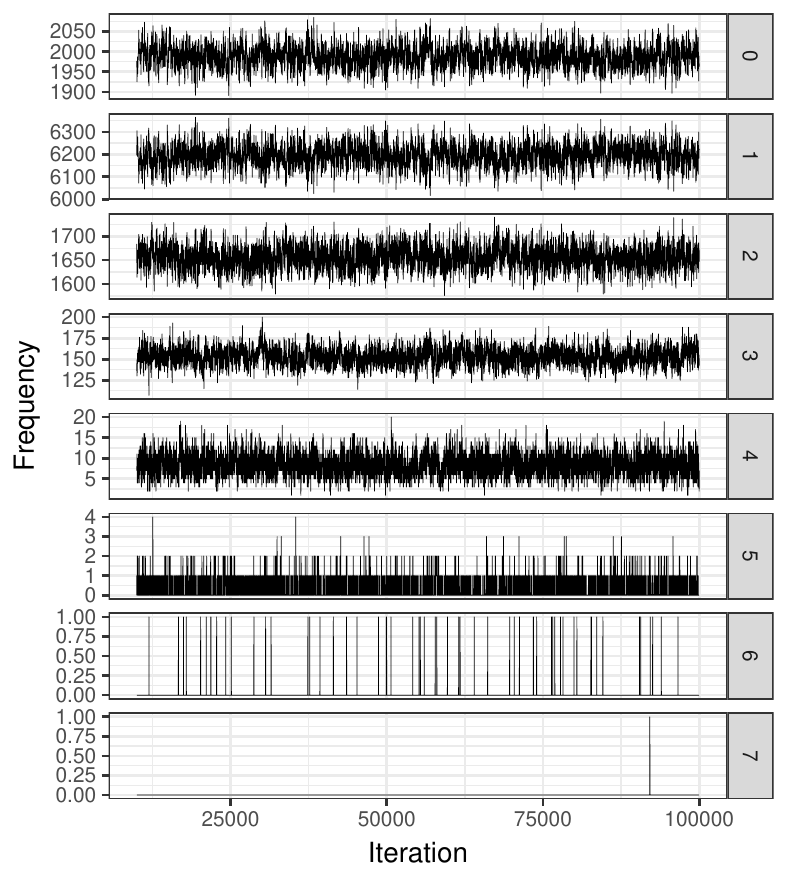} \hfill
  \includegraphics[width=0.48\linewidth]{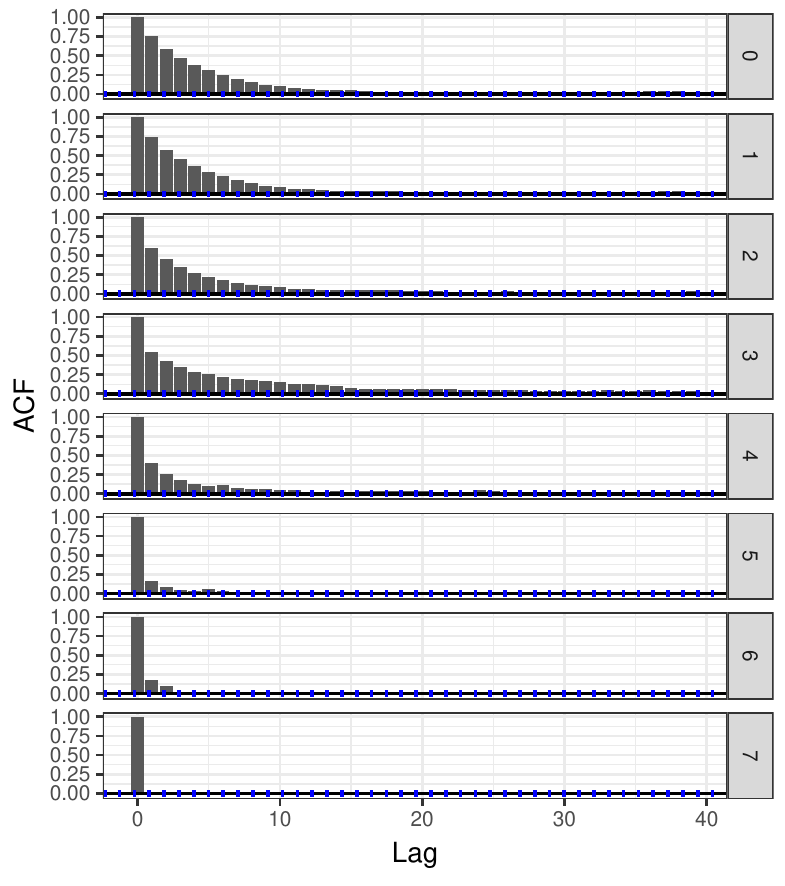}
  \caption{Cluster size distribution for \texttt{RLdata10000}}
\end{figure}

{
\footnotesize
\bibliographystyleApp{jasa}
\bibliographyApp{dblink}
}

\fi

\end{document}